\documentclass[prd,tightenlines,nofootinbib,superscriptaddress]{revtex4}

\usepackage{pre_all}
\usepackage{pre_QI}
\allowdisplaybreaks

\hypersetup{
pdfstartview = {FitH},
}
\hypersetup{
	colorlinks=true,         
	linkcolor=brown,          
	citecolor=red,        
	urlcolor=blue            
}
\begin{document}
\title{Quantum Curved Tetrahedron, Quantum Group Intertwiner Space, and Coherent States}


\author{{\bf Chen-Hung Hsiao}}\email{chsiao2017@fau.edu}
\affiliation{Department of Physics, Florida Atlantic University, 777 Glades Road, Boca Raton, FL 33431, USA}

\author{{\bf Qiaoyin Pan}}\email{qpan@fau.edu}
\affiliation{Department of Physics, Florida Atlantic University, 777 Glades Road, Boca Raton, FL 33431, USA}

\date{\today}

\begin{abstract}

In this paper, we construct the phase space of a constantly curved tetrahedron with fixed triangle areas in terms of a pair of Darboux coordinates called the length and twist coordinates, which are in analogy to the Fenchel-Nielsen coordinates for flat connections, and their quantization. The curvature is identified to the value of the cosmological constant, either positive or negative. The physical Hilbert space is given by the $\UQ$ intertwiner space. We show that the quantum trace of quantum monodromies, defining the quantum length operators, form a fusion algebra and describe their representation theory. We also construct the coherent states in the physical Hilbert space labeled by the length and twist coordinates. These coherent states describe quantum curved tetrahedra and peak at points of the tetrahedron phase space. This works is closely related to 3+1 dimensional Loop Quantum Gravity with a non-vanishing cosmological constant. The coherent states constructed herein serve as good candidates for the application to the spinfoam model with a cosmological constant. 

\end{abstract}

\maketitle
\tableofcontents
\renewcommand\thesection{\Roman{section}}

\section{Introduction}

Quantum tetrahedron is the building block in the theory of Loop Quantum Gravity (LQG). 
It is essential for the construction of both the canonical formalism and covariant formalism, the latter also called the spinfoam model. In the canonical formalism, quantum tetrahedra build up the spacial quantum geometries, which evolve dynamically under the Wheeler-DeWitt equation. In the spinfoam model, five quantum tetrahedra are included on the boundary of a four-simplex, whose amplitude describes the elementary dynamics of a discretized four-manifold. 
In the case of vanishing cosmological constant, classically, \textit{Minkowski theorem} \cite{minkowski1897allgemeine} guarantees that a tetrahedron can be reconstructed from a set of four vectors $\vec{J}_i=a_i\hat{n}_i,i=1,...,4$ satisfying flat closure condition, i.e. $\sum^4_{i=1} \vec{J}_i=\vec{0}$, in which the information of the area $a_i$ and normal $\hat{n}_i$ of each face of the tetrahedron are encoded. The space of such vectors modulo rotation has the structure of a symplectic manifold \cite{10.4310/jdg/1214459218,Bianchi:2010gc} and is known as the \textit{Kapovich-Millson phase space} $S_4(a_1,a_2,a_3,a_4)$ specified by four areas, 
\be
S_{4}(a_1,a_2,a_3,a_4)=\{\hat{n}_i\in (S^2)^{\times 4}| \sum^4_{i=1} \vec{J}_i=\vec{0}\}/\SO(3)\;, \quad \vec{J}_i:=a_i\hat{n}_i\in\R^3\,,
\ee
where the modulo $\SO(3)$ comes from the fact that the flat closure condition is invariant under the simultaneous $\SO(3)$ rotation and those four normals do not lie within the same plane. It is a two-dimensional real space parametrized by a pair of variables $(\mu,\theta)$ with the Poisson bracket $\{\mu,\theta\}=1$. $\mu$ is called the diagonal length and $\theta$ is the dihedral angle between the two triangles as illustrated in fig.\ref{fig:4-gon_flat}. 
\begin{figure}[h!]
\centering
\begin{tikzpicture}
    \coordinate (A) at (0,-2.0);
\coordinate (B) at (2,-0);
\coordinate (C) at (-2,0);
\coordinate (D) at (0,2.0);

\draw[->, thick, >=stealth] (A) -- (B);
\draw[->, thick, >=stealth] (B) -- (D);
\draw[->, thick, >=stealth] (D) -- (C);
\draw[->, thick, >=stealth] (C) -- (A);
\draw[red, ->, thick, >=stealth] (A) -- (D);
\draw (0,0) node[red,anchor=west] {$\mu$};
\draw[->, thick, >=stealth] (0,0.2) -- (0.4,0.6);
\draw[->, thick, >=stealth] (0,0.2) -- (-0.4,0.6);
\draw[<-> ,blue, thick] (0.3,0.5) to[out=120, in=60] (-0.3,0.5);
\draw (0,0.9) node[blue,anchor=west] {$\theta$};
\draw (1,-1) node[anchor=west] {$\Vec{J}_1=a_1\hat{n}_1$};
\draw (1,1) node[anchor=west] {$\Vec{J}_2=a_2\hat{n}_2$};
\draw (-1,1) node[anchor=east] {$\Vec{J}_3=a_3\hat{n}_3$};
\draw (-1,-1) node[anchor=east] {$\Vec{J}_4=a_4\hat{n}_4$};
\draw (A) node[anchor=north] {$1$};
\draw (B) node[anchor=west] {$2$};
\draw (C) node[anchor=east] {$4$};
\draw (D) node[anchor=south] {$3$};
\end{tikzpicture}
\caption{A polygon with four vectors in $\mathbb{R}^3$. Here, $\Vec{J}_i=a_i\hat{n}_i,i=1,...,4$ and all the four normals do not lie within the same plane. $\mu$ ({\it in red}) is the diagonal length and $\theta$ ({\it in blue}) is the dihedral angle between the triangle $f_{123}$ bounded by vertices 1,2,3 and triangle $f_{134}$ bounded by vertices 1,3,4.}
    \label{fig:4-gon_flat}
\end{figure}
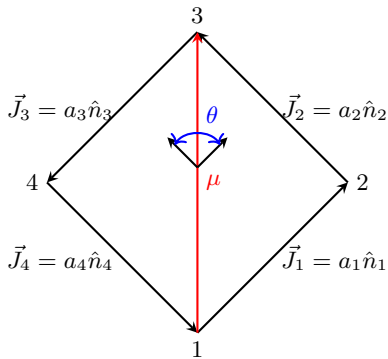
 We refer to it as the space of shapes of a tetrahedron with fixed areas \cite{Bianchi:2010gc}. The quantization of the phase space is the $4$-valent $\SU(2)$ intertwiner space $\mathrm{Inv}_{\SU(2)}(V^{j_1}\otimes V^{j_2}\otimes V^{j_3}\otimes V^{j_4})$, which is the solution space of the quantum flat closure condition. The coherent intertwiner was constructed in \cite{Livine:2007vk}. The expectation values of geometrical operators using these states in the semi-classical limit peak at the point of the phase space \cite{Conrady:2009px}.

In the case of a non-vanishing cosmological constant, the \textit{Kapovich-Millson} description has been generalized to the shapes of constantly curved tetrahedron, where the constant curvature is identified with cosmological constant and the polygon is on $S^3$. The phase space of shapes of a constantly curved tetrahedron is parametrized by the diagonal length $\ln{x}$ and the bending angle $\ln{y}$ (see \eqref{eq:def_theta_phi} and fig.\ref{fig:4-gon}). This phase space is identified with the moduli space of $\SU(2)$ flat connections on a four-punctured sphere, which is the solution space of the classical curved closure condition 
\be
M_{4}M_{3}M_{2}M_{1}=\Id_{\SU(2)},
\ee
where $M_\nu, \;\nu=1,...,4$ is the $\SU(2)$ monodromy around the $\nu$-th face of the homogeneously curved tetrahedron.

The length and the twist coordinates are in analogy to to the complex Fenchel-Nielsen coordinates for flat connections, in that they have the same Poisson bracket given by $\{\ln x,\ln y\}=1$. Inspired by this identification, the diagonal length should be related to the Fenchel-Nielsen length, which, classically, is defined as the Wilson loop along the loop enclosing a pair of punctures.

The quantization of the phase space is the $4$-valent intertwiner space $W^0(K_1,K_2,K_3,K_4)$ of quantum group $\UQ$, which is proved to be the only solution space of the quantum curved closure condition \cite{Han:2023wiu}. The quantum counterpart of the Wilson loop is constructed by the combinatorial quantization \cite{Alekseev:1994au,Alekseev:1994pa,Alekseev:1995rn} (see also \cite{Han:2023wiu}), where the operator algebras, i.e., the graph algebra and loop algebra, on a four-punctured sphere are constructed. The quantum monodromy of the loop enclosing a pair of punctures satisfies the defining relations of a loop algebra, from which one constructs the $q$-deformed Wilson loop operator $c^I_{\ell}$, where $I$ spans all the physical irreducible representations of $\UQ$. However, the elements $c^I_{\ell}$ are no longer central elements, in contrast to those obtained from loops $\ell_\nu, \nu=1,...,4$ around each puncture. Nevertheless, the elements $c^I_{\ell}$ still form a fusion algebra, and due to the first-pinching theorem, the intertwiner space $W^0(K_1,K_2,K_3,K_4)$ is decomposed into $\bigoplus_{J} W^J(K_1,K_2) \otimes W^{\bar{J}}(K_3,K_4)$, where $\Bar{J}$ represents the dual representation to $J$ and the eigenvalue of the element $c^I_{\ell}$ in the intertwiner space is given by the unnormalized $S$-matrix, where $s_{IJ}:=\tr_q^I \otimes \tr_q^J(R'R)$.

Classically, the $\SU(2)$ monodromy matrix is a $2$ by $2$ matrix, i.e., the irreducible representation is taken to be the fundamental representation. In the quantum case, considering the fundamental representation, the eigenvalue $s_{\frac{1}{2} J}$ is taken a simple expression as
 \be
  s_{\f12 J}=e^{\frac{i \pi}{k+2}(2J+1)}+ e^{-\frac{i \pi}{k+2}(2J+1)}\;,
 \ee
 where $k\in \N$. 
 Inspired by the expression above, we define the quantum length operator $\bold{\tilde{x}}$ such that its action on the intertwiner space $W^0(K_1,K_2,K_3,K_4)$ is defined as
\be
\bold{\tilde{x}}|\Psi^J\rangle=e^{\frac{i \pi}{k+2}(2J+1)} |\Psi^J\rangle,\quad \bold{\tilde{y}}|\Psi^J\rangle= |\Psi^{J+1}\rangle\,,
\ee
where $\bold{\tilde{y}}$ is treated as the translational operator. $\bold{\tilde{x}}$ and $\bold{\tilde{y}}$ satisfy the commutation relation:
\be
\bold{\tilde{x}}\bold{\tilde{y}}=q\bold{\tilde{y}}\bold{\tilde{x}}\;.
\ee

 To define the coherent state in the physical Hilbert space $\cH\equiv W^0(K_1,K_2,K_3,K_4)$, we firstly define the auxiliary space, $\mathcal{H}_{\mathrm{aux}}=\mathbb{C}^{2k+4}$, which carries the irreducible representation of the Weyl algebra, $\bold{x}\bold{y}=q^{\frac{1}{2}}\bold{y}\bold{x}$. Their actions are defined as
\be
\bold{x}|\Psi^J\rangle=e^{\frac{i \pi}{k+2}(2J+1)} |\Psi^J\rangle,\quad \bold{y}|\Psi^J\rangle= |\Psi^{J+\frac{1}{2}}\rangle\,.
\ee
It is beneficial to consider the auxiliary space as it can be naturally viewed as the quantization of the phase space of a torus and that a set of coherent states $\psi_{(\Xt_0,\Yt_0)}(x)$ on a torus is well defined \cite{Gazeau:2009zz}.
To get back to the physical Hilbert space, which is a subspace of $\cH_{\text{aux}}$, we define a projector $P:\mathcal{H}_{\mathrm{aux}}\to \cH$. The coherent states in the $\cH$ then are the projected ones from the auxiliary space. In the semi-classical limit, both $k$ and representation labels $K_{\nu}$ increase at the same rate (i.e. $k=\lambda k$, $K_{\nu}=\lambda K_{\nu}$). The expectation of $\bold{\Tilde{x}}$ and $\bold{\Tilde{y}}$ in the projected coherent states peak at a point $(\tilde{X}_0,\tilde{Y}_0)$ of the phase space, with $\tilde{X}_0$ being the logarithm of the length variable and $\tilde{Y}_0$ being the logarithm of the bending angle (up to an imaginary number). More precisely,
\be
\langle\Tilde{\bold{x}}\rangle=e^{i\Xt_0}+O\lb e^{-\lambda}/\sqrt{\lambda}\rb\,,\quad
\langle\Tilde{\bold{y}}\rangle=e^{2i\Yt_0}+ O \lb e^{-\lambda}/\sqrt{\lambda}\rb \;,
\ee
where the position coordinate $\Xt_0$ needs to satisfy the triangle inequality: $2 \max\left(\frac{|K_1-K_2|}{k+2},\frac{|K_3-K_4|}{k+2}\right) \leq \frac{\Xt_0}{\pi} \leq 2 \min\left(\frac{u(K_1,K_2)}{k+2},\frac{u(K_3,K_4)}{k+2}\right)$ with $u(K_i,K_j)= \min(K_i+K_j,k-K_i-K_j)$. Otherwise, the expectation values of $\Tilde{\bold{x}}$ and $\Tilde{\bold{y}}$ exponential decay by $\lambda$: 
\be
\langle\Tilde{\bold{x}}\rangle= O \lb e^{-\lambda}/\sqrt{\lambda}\rb \;,\quad
\langle\Tilde{\bold{y}}\rangle= O \lb e^{-\lambda}/\sqrt{\lambda}\rb \;.
\ee

This paper is organized as follows: Section \ref{sec:FN_coord} reviews the main ideas of the phase space of a constantly curved tetrahedron, which can be described by length and twist coordinates, analogous to complex Fenchel-Nielsen coordinates for flat connections. The length and twist coordinates possess geometrical interpretations in the 4-gon on $S^3$.
In Section \ref{sec:review}, we review some facts about the moduli algebra and its representation theory. For detailed information, we refer to \cite{Fock:1998nu, Alekseev:1994au, Alekseev:1994pa, Alekseev:1995rn} (See also \cite{Han:2023wiu}).
In Section \ref{sec:intro}, we prove that the quantum monodromy of the loop around a pair of punctures satisfies the defining relations of a loop algebra. The fusion algebra $\mathcal{V}(\ell)$ generated by elements $c^I_{\ell}$ can be constructed from these quantum monodromies. However, all the elements $c^I_{\ell}$ remain gauge-invariant; they are no longer central elements.
In Section \ref{sec:rep_fusion}, the intertwiner space has the following decomposition: $W^0(K_1,K_2,K_3,K_4)=\bigoplus_{J} W^J(K_1,K_2)\otimes W^{\bar{J}}(K_3,K_4)$, where $J$ must satisfy the triangle inequality. We calculate the eigenvalue of the elements $c^I_{\ell}$ in the intertwiner space.
In Section \ref{sec:quantize_coordinates}, we define the quantum counterparts of the twist and length operators, which form a Weyl algebra, and their actions on the intertwiner space. The commutation relation of the Weyl algebra serves as the quantization of the Poisson bracket of the length and twist coordinates.
In Section \ref{sec:coherent_state}, we introduce the auxiliary space, which carries the irreducible representations of the Weyl algebra quantizing the torus. Here, we construct the coherent state in the auxiliary space. Moving on to Section \ref{sec:projection_and_exp}, we define a set of projectors, and the coherent states in the intertwiner space are obtained through projections. Lastly, we calculate the expectation values of the length and twist operators in the semiclassical limit, which peak at phase space points.


\section{Phase space of a tetrahedron and the length-twist coordinates}
\label{sec:FN_coord}

In this section, we briefly review the phase space of a convex constantly curved tetrahedron, which we will denote as a curved tetrahedron for conciseness, as the moduli space of $\PSU(2)\cong\SO(3)$ flat connection on a four-punctured sphere and describe the length and twist coordinates as a pair of Darboux coordinates of the phase space. More details of the former can be found in \cite{Han:2023wiu} and references therein. 

We identify the curvature of the curved tetrahedron, which could be positive or negative, to be the cosmological constant $\Lambda$. Denote $s=\sgn(\Lambda)$ and the $n$-dimensional constantly curved space as $\bE^{n,s}$ with $\bE^{n,+}=S^n$ and $\bE^{n,-}=\bH^n$. Each boundary triangle of the tetrahedron is flatly embedded in a 2D subspace $\bE^{2,s}$ of the same $\bE^{3,s}$, as illustrated in fig.\ref{fig:tetrahedra}.  
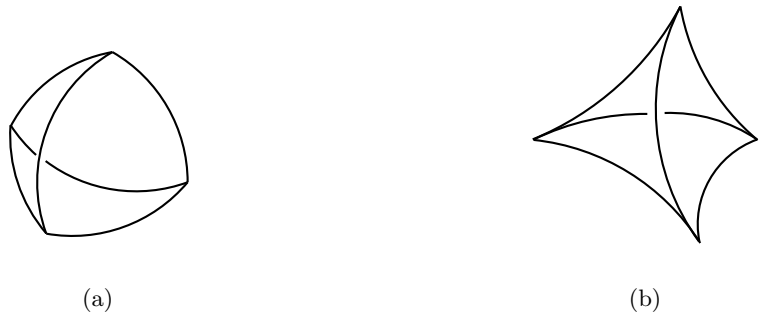
\begin{figure}[h!]
\centering
\begin{subfigure}{0.4\linewidth}
     \begin{tikzpicture}[scale=2]
  \draw[black, thick] (3,3) arc (0:60:1);
 \draw[black, thick] (2.5,3.866) arc (100:152.5:0.95);
 \draw[black, thick] (2.5,3.866) arc (120:200:1);
 \draw[black, thick] (2.06,2.66) arc (-139.5:-184:1);
 \draw[black, thick] (3,3) arc (-40:-100:1);
 \draw[black, thick] (3,3) arc (-70:-127:1);
 \draw[black, thick] (1.823,3.380) arc (-147:-132:1);
 \end{tikzpicture} 

	\subcaption{}
	\label{fig:SphericalTetra}
\end{subfigure}
\begin{subfigure}{0.4\linewidth}
    \begin{tikzpicture}[scale=3]
  \draw[black,thick] (-3,-3) arc (-27:36:-1);
 \draw[black,thick] (-3,-3) arc (-28.2:-67.8:1.3);
 \draw[black,thick] (-3,-3) arc (10:50:-1);
 \draw[black,thick] (-2.66,-3.59) arc (-124:-88:-0.69);
 \draw[black,thick] (-3.65,-3.59) arc (-64.1
 :-90.5:-1.13);
 \draw[black,thick] (-2.66,-3.59) arc (-70:11:-0.4);
 \draw[black,thick] (-2.912,-4.049) arc (-147.5:-96:-1);
  \end{tikzpicture}
	\subcaption{}
	\label{fig:HyperbolicTetra}
 \end{subfigure}
\caption{
{\it (a)} A tetrahedron flatly embedded in $S^3$. {\it (b)} A tetrahedron flatly embedded in $\bH^3$.}
\label{fig:tetrahedra}
\end{figure}
A (non-degenerate) curved tetrahedron can be uniquely reconstructed from the so-called {\it closure condition}. This is described by the {\it curved Minkowski Theorem} \cite{Haggard:2015ima}. 
In short, it states that, given a cosmological constant $\Lambda$  
and four $\PSU(2)$ group elements $\{\Mt_1,\Mt_2,\Mt_3,\Mt_4\}$ satisfying the closure condition
\be
\Mt_4\Mt_3\Mt_2\Mt_1=\Id_{\PSU(2)}\,,
\label{eq:closure}
\ee
we interpret these group elements as holonomies based at the same vertex ${\mathfrak b}$ of a curved tetrahedron and each $\Mt_\nu\, (\nu=1,\cdots4)$ is along a {\it simple path}\footnote{Given the closure condition \eqref{eq:closure}, the simple path for $\Mt_\nu$ with $\nu=1,2,3$ is simply a path starting from ${\mathfrak b}$ along the boundary of the triangle $\Mt_\nu$ encloses. On the other hand, the simple path for $\Mt_4$ is made of three parts. It starts from ${\mathfrak b}$ along the edge shared by the simple paths for $\Mt_1$ and $\Mt_3$, called the special edge, and arrives at a vertex on the remaining triangle, then along the boundary of the triangle, and finally going back to ${\mathfrak b}$ along the special edge. See \eg fig.2 of \cite{Han:2023wiu} for an illustration. } around a triangle in the same orientation of the triangle as a Riemann surface. Then these four holonomies uniquely determine a curved tetrahedron up to isometry and  
\be
\Mt_\nu=\exp\lb  \frac{\Lambda}{3} a_\nu \hat{n}_\nu \cdot\vec{\tau}\rb
\label{eq:param_holo}
\ee 
encodes the area $a_\nu$ and the outward-pointing normal $\hat{n}_\nu$ of the $\nu$-th triangle, where $\vec{\tau}=\f{1}{2i}\vec{\sigma}$ is the vector of the $\su(2)$ generators. We refer to \cite{Haggard:2015ima} (see also \cite{Han:2023wiu}) for a more detailed description of the curved Minkowski theorem. 

On the other hand, the holonomies $\{\Mt_\nu\}_{\nu=1,\cdots,4}$ satisfying the closure condition can be interpreted as the fundamental group of a four-punctured sphere, denoted as $\Sfour$, represented in $\PSU(2)$ group, which describes the moduli space of $\PSU(2)$ flat connection on $\Sfour$:
\be
\cM_{\Flat}^0(\Sfour, \PSU(2))=\{\Mt_1,\Mt_2,\Mt_3,\Mt_4\in\PSU(2):\Mt_4\Mt_3\Mt_2\Mt_1=\Id_{\PSU(2)}\}/\PSU(2)\,,
\label{eq:flat_connection_def}
\ee
where each $\Mt_\nu$ is now the holonomy of the loop $\ell_\nu$ around the $\nu$-th puncture and the quotient is by the conjugate action of $\PSU(2)$. $\cM_{\Flat}^0(\Sfour, \PSU(2))$ is a Poisson manifold but not a symplectic one. 
We fix the conjugacy classes of the holonomies $\{\Mt_\nu\}_{\nu=1,2,3,4}$, each of which is labelled by the eigenvalue $\tilde{\lambda}_\nu$. Then it defines a 2-dimensional symplectic space, denoted as $\cM_{\Flat}(\Sfour, \PSU(2))$. Given the holonomies parametrized as in \eqref{eq:param_holo} in the geometrical interpretation of a curved tetrahedron, $\tilde{\lambda}_\nu$ encodes the area of the $\nu$-th triangle by $\f12\lb \tilde{\lambda}_\nu+\tilde{\lambda}^{-1}_\nu\rb=\cos\lb\f{|\Lambda|}{6}a_\nu \rb$. In this sense, we also call $\cM_{\Flat}(\Sfour, \PSU(2))$ the phase space of a curved tetrahedron with fixed triangle areas. 

For convenience of quantization, we lift each $\Mt_\nu\in\PSU(2)$ to $M_\nu\in\SU(2)$ whose eigenvalue $\lambda_\nu=\epsilon_\nu\sqrt{\tilde{\lambda}_\nu}$ is randomly chosen to be the positive ($\epsilon_\nu=+$) or negative ($\epsilon_\nu=-$) square root of that of $\Mt_\nu$. Then the corresponding symplectic space is
\be
\cM_{0,4}^{\vec{\lambda}}=\left\{M_1,M_2,M_3,M_4\in\SU(2):M_\nu=G_\nu \mat{cc}{\lambda_\nu&0\\0&\lambda_\nu^{-1}}G^{-1}_\nu,\,G_\nu\in\SU(2); \, M_4M_3M_2M_1=\Id_{\SU(2)}\right\}/\SU(2)\,,
\label{eq:flat_connection_def}
\ee
with $\{\lambda_\nu\}_{\nu=1,\cdots,4}$ fixed. The dimension of phase space $\cM_{0,4}^{\vec{\lambda}}$ is 2. 
We are interested in the Darboux coordinates of this phase space, (at least some of) which have interpretations of holonomies on the 4-punctured sphere. Such coordinates can be provided by the {\it length-twist coordinates} based on the pants decomposition of $\Sfour$.  

A $\Sfour$ can be decomposed into a pair of pants by cutting along a closed curve, say $c$, enclosing (any) two punctures, as illustrated in fig.\ref{fig:pants}.
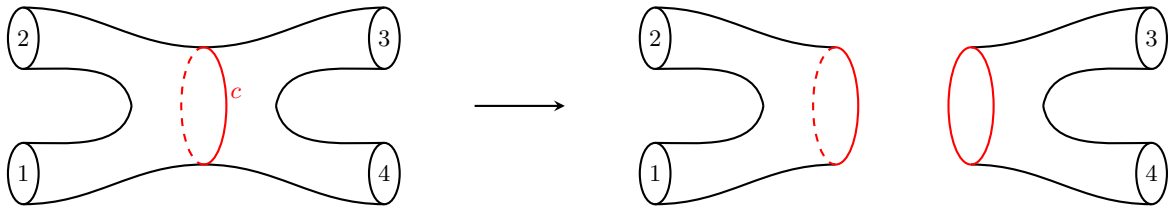
\begin{figure}[h!]
\begin{tikzpicture}
[scale=1.2,one end extended/.style={shorten >=-#1},
 one end extended/.default=1cm]
\coordinate (A) at (0,0.34);
\coordinate (B) at (0,-0.34);
\coordinate (C) at (0,-1.16);
\coordinate (D) at (0,-1.84);
\coordinate (O) at (2,-0.1);
\coordinate (P) at (2,-1.4);
\coordinate (R) at (1.2,-0.75);
\coordinate (S) at (2.8,-0.75);
\coordinate (E) at (4,0.34);
\coordinate (F) at (4,-0.34);
\coordinate (G) at (4,-1.16);
\coordinate (H) at (4,-1.84);
\draw[thick, shift={(0,0)}] (0,0) node{2} ellipse (0.17cm and 0.34cm);
\draw[thick, shift={(0,-1.5)}] (0,0) node{1} ellipse (0.17cm and 0.34cm);
\draw[thick, shift={(4,-1.5)}] (0,0) node{4} ellipse (0.17cm and 0.34cm);
\draw[thick, shift={(4,0)}] (0,0) node{3} ellipse (0.17cm and 0.34cm);
 \draw[thick] (A) to[out=0,in=180](O);
 \draw[thick] (D) to[out=0,in=180]
 (P);
 \draw[thick] (E) to[out=180,in=0]
 (O);
 \draw[thick] (H) to[out=180,in=0]
 (P);
 \draw[thick] (B) to[out=0,in=100]
 (R);
 \draw[thick] (C) to[out=0,in=-100]
 (R);
 \draw[thick] (F) to[out=180,in=80]
 (S);
 \draw[thick] (G) to[out=180,in=-80]
 (S);
\draw[thick, shift={(0,0)},red] (2,-1.4) arc (-90:90:0.25cm and 0.65cm);
  \draw[thick, dashed, shift={(0,0)},red] (2,-0.1) arc (90:270:0.25cm and 0.65cm);
 \draw[->, thick, >=stealth] (5,-0.75) -- (6,-0.75);
 \draw (2.2,-0.6) node[red,right]{$c$};
\coordinate (a) at (7,0.34);
\coordinate (b) at (7,-0.34);
\coordinate (c) at (7,-1.16);
\coordinate (d) at (7,-1.84);
\coordinate (o) at (9,-0.1);
\coordinate (p) at (9,-1.4);
\coordinate (x) at (10.5,-0.1);
\coordinate (y) at (10.5,-1.4);
\coordinate (r) at (8.2,-0.75);
\coordinate (s) at (11.3,-0.75);
\coordinate (e) at (12.5,0.34);
\coordinate (f) at (12.5,-0.34);
\coordinate (g) at (12.5,-1.16);
\coordinate (h) at (12.5,-1.84);
\draw[thick, shift={(0,0)}] (7,0) node{2} ellipse (0.17cm and 0.34cm);
\draw[thick, shift={(0,-1.5)}] (7,0) node{1} ellipse (0.17cm and 0.34cm);
\draw[thick, shift={(5.5,-1.5)}] (7,0) node{4} ellipse (0.17cm and 0.34cm);
\draw[thick, shift={(5.5,0)}] (7,0) node{3} ellipse (0.17cm and 0.34cm);
 \draw[thick] (a) to[out=0,in=180] 
 (o);
 \draw[thick] (d) to[out=0,in=180]
 (p);
 \draw[thick] (e) to[out=180,in=0]
 (x);
 \draw[thick] (h) to[out=180,in=0]
 (y);
 \draw[thick] (b) to[out=0,in=100]
 (r);
 \draw[thick] (c) to[out=0,in=-100]
 (r);
 \draw[thick] (f) to[out=180,in=80]
 (s);
 \draw[thick] (g) to[out=180,in=-80]
 (s);
\draw[thick, shift={(0,0)},red] (9,-1.4) arc (-90:90:0.25cm and 0.65cm);
  \draw[thick, dashed, shift={(0,0)},red] (9,-0.1)  arc (90:270:0.25cm and 0.65cm);
  \draw[thick, shift={(0,0)},red] (10.5,-1.4) arc (-90:90:0.25cm and 0.65cm);
  \draw[thick, shift={(0,0)},red] (10.5,-0.1) arc (90:270:0.25cm and 0.65cm);
  \end{tikzpicture}
\caption{A pants decomposition of a 4-punctured sphere by cutting a closed curve $c$ ({\it in red}) surrounding the 1st and 2nd (or the 3rd and 4th) punctures. 
}
\label{fig:pants}
\end{figure}

An $\SU(2)$ flat connection on $\Sigma_{0,4}$ defines a holonomy $H$ along $c$ with an eigenvalue denoted by $x$. $x$ is called the {\it length coordinate} and can be used as the canonical coordinate of the phase space $\cM_{0,4}^{\vec{\lambda}}$. Its conjugate momentum $y$, is called the {\it twist coordinate}. The Atiyah-Bott-Goldman symplectic 2-form on $\cM_{0,4}^{\vec{\lambda}}$ can be written in terms of $(x,y)$ as
\be
\omega=\f{\delta y}{y}\wedge\f{\delta x}{x}\,,
\ee
where $\delta$ denotes the standard differential on $\cM_{0,4}^{\vec{\lambda}}$. It induces a canonical Poisson bracket on the logarithmic length-twist coordinates:
\be
\{\ln x,\ln y\}=1\,.
\label{eq:Poisson_length_twist}
\ee
The length-twist coordinates are the analogy of the Fenchel-Nielsen coordinates \cite{fenchel2011discontinuous}, which are in general symplectic coordinates of the Techm\"uller space for a Riemann surface associated to a pants decomposition. 

In correspondence to a curved tetrahedron, which is described by $\cM_\Flat(\Sfour,\PSU(2))$, we choose the lifts of the logarithmic length-twist coordinates such that $\ln x^2,\ln y^2\in i(0,2\pi)$ (zero can not be reached when the tetrahedron under consideration is non-degenerate), \ie $\ln x,\ln y\in i(0,\pi)$. 
Then these coordinates possess a clear geometrical interpretation in the 4-gon on $S^3\cong\SU(2)$ whose edges are geodesic curves. 
Consider 4 points $\{v_i\}_{i=1}^4$ on $S^3$ located at 
\be
v_1=\Id_{\SU(2)}\,,\quad
v_2=M_1\,,\quad
v_3=M_2M_1\,,\quad
v_4=M_3M_2M_1\,.
\ee
\begin{figure}[h!]
    \centering
    \includegraphics[width=0.4\textwidth]{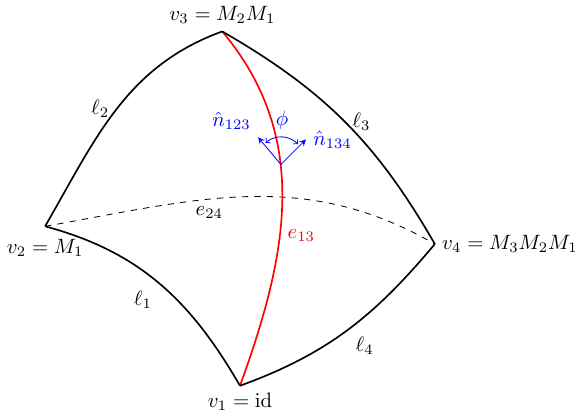}
    \caption{A 4-gon on $\SU(2)\cong S^3$ formed by geodesic curves $\{\ell_i\}_{i=1}^4$ connecting four points $v_1=\Id,v_2=M_1,v_3=M_2M_1,v_4=M_3M_2M_1$ in cyclic order. The geodesic curve ({\it in red}) $e_{13}$ connecting $v_1$ and $v_3$ has length $\theta$. Further connecting $v_2$ and $v_4$ with a geodesic curve $e_{24}$ ({\it dashed}) forms a curved tetrahedron on $S^3$ whose faces are geodesics. $\hat{n}_{123}$ and $\hat{n}_{134}$ ({\it one-way arrows in blue}) are outward-pointing (relative to the tetrahedron) normal vectors of the geodesic triangle $f_{123}$ bounded by $\ell_1,\ell_2,e_{13}$ and the geodesic triangle $f_{134}$ bounded by $\ell_3,\ell_4,e_{13}$ respectively. $\phi\in[0,\pi]$ is the dihedral angle between $f_{123}$ and $f_{134}$ hinged by $e_{13}$. }
    \label{fig:4-gon}
\end{figure}
A 4-gon is formed by 4 geodesic curves $\ell_1\equiv e_{12}\,,\,\ell_2\equiv e_{23}\,,\,\ell_3\equiv e_{34}\,,\,\ell_4\equiv e_{41}$ where $e_{ij}$ is the geodesic connecting $v_{i}$ and $v_{j}$, as shown in fig.\ref{fig:4-gon}. The geodesic length ${\mathfrak l}_\nu\in[0,\pi]$ of $\ell_\nu$ satisfies 
\be
\cos({\mathfrak l}_\nu)=(\lambda_\nu+\lambda_\nu^{-1})/2,
\ee
for $\nu=1,\cdots,4$. Denote $\theta\in[0,\pi]$ to be the length of the diagonal geodesic curve $e_{13}$ connecting $v_1$ and $v_3$, which separates the 4-gon into two (curved) triangles $f_{123}$ bounded by $\ell_1,\ell_2,e_{13}$ and $f_{134}$ bounded by $\ell_3,\ell_4,e_{13}$. 
On the other hand, $\phi\in[0,\pi]$ describes the bending angle between the two triangles. 
Adding the other diagonal geodesic curve, one forms a curved tetrahedron in $S^3$. 

Given fixed lengths $\{{\mathfrak l}_1,{\mathfrak l}_2,{\mathfrak l}_3,{\mathfrak l}_4\}$ of the four geodesic curves of the 4-gon, $\theta$ and $\phi$ uniquely determine the shape of this convex curved tetrahedron embedded in $S^3$ \cite{Nekrasov:2011bc}.  
As side lengths of two triangles, $\theta$ is restricted by the following triangular inequality 
\be
\max(|{\mathfrak l}_1-{\mathfrak l}_2|,|{\mathfrak l}_3-{\mathfrak l}_4|)\leq\theta\leq\min({\mathfrak l}_1+{\mathfrak l}_2,{\mathfrak l}_3+{\mathfrak l}_4)\,,\quad
\label{eq:range_theta_phi}
\ee
while $\phi\in[0,\pi]$ can be freely chosen. When the curved tetrahedron constructed from the 4-gon is non-degenerate, $\{{\mathfrak l}_\nu\},\theta,\phi\in(0,\pi)$. Interpret the length and twist coordinates as
\be
\ln x = i\theta\,,\qquad
\ln y = i\phi\,.
\label{eq:def_theta_phi}
\ee
Then the length and twist coordinates encode the Darboux coordinates $(\phi,\theta)$ of the 4-gon. 
It is shown in \cite{Han:2023wiu} that the phase space $\cM_{0,4}^{\vec{\lambda}}$ can be quantized to be the Hilbert space of the moduli algebra (see next section). 
In the remaining of this paper, we aim to quantize the length and twist coordinates into operators in this Hilbert space and construct a coherent state with these operators. 


\section{Moduli algebra $\mathfrak{M}^{K_1,K_2,K_3,K_4}_{0,4}$ and its representation}
\label{sec:review}


In this section, we provide a concise review of the mathematical tools, namely the moduli algebra on $\Sfour$, denoted by $\mathfrak{M}^{K_1,K_2,K_3,K_4}_{0,4}$, and its representation theory, for constructing the quantum theory of the classical phase space $\cM_{\Flat}(\Sfour,\SU(2))$. We aim to give the necessary formulas for the succeeding sections and skip details and derivations. For details, we refer to the seminal works in \cite{Fock:1998nu,Alekseev:1994au,Alekseev:1994pa,Alekseev:1995rn} and the companion paper \cite{Han:2023wiu}.

Classically, the solution space of the classical closure condition, $M_4M_3M_2M_1=\Id$, is the moduli space of $\SU(2)$ flat connections on $\Sigma_{0,4}$. 
The Poisson structure of $\cM_{\Flat}(\Sfour,\SU(2))$ is defined by the use of the so-called {\it classical $r$-matrix}. A natural quantization is given by quantizing the classical monodromy $M_\nu$ to the quantum monodromy operator $\bM_\nu$, quantizing the commutation relations in terms of the {\it quantum monodromies}, which are operator matrices, and a {\it quantum $\cR$-matrix}, which is the quantum version of the $r$-matrix. 
 The quantum theory requires the notion of quasitriangular ribbon (quasi) Hopf algebra $\UQ$, which is a deformation of the universal enveloping algebra of $\su(2)$ with a deformation parameter $q\in\bC$ and some extra algebraic structures, \eg the $*$-structure. In this paper, we only focus on the case when the deformation parameter $q^{k+2}=1$ is a root-of-unity, where $k\in\N_+$ is the Chern-Simons level. 
 In this case, $\UQ$ has finitely many  
 irreducible representations labeled by spins $I=0,\f12,\cdots, \f{k}{2}$. In other words, the Chern-Simons level produces a truncation on the representations $I\in \N/2$. We call these representations after truncation the {\it physical representations}. 
 
The collection of the matrix elements of the quantum monodromy around one single puncture form an algebra called the {\it loop algebra} denoted by $\cL_{0,1}$. When the quantum monodromy is around the $\nu$-th puncture, we also denote the loop algebra by $\cL_\nu$. 
{Similarly, the collection of the matrix elements of the $\{\bM_\nu\}_{\nu=1,\cdots,4}$ form an algebra on the four-punctured sphere called the {\it graph algebra} denoted by $\cL_{0,4}$. In general, the graph algebra on $m$-punctured sphere is denoted $\cL_{0,m}$.}
On top of that, one can define a ``quantum gauge action'' on the graph algebra. The algebra of quantum observables, called the {\it invariant algebra} denoted as $\cA_{0,4}$, is described by the subalgebra of $\cL_{0,4}$ that is invariant under the quantum gauge action. 
More precisely, $\cA_{0,4}$ is defined as the subalgebra of $\cL_{0,4}$ containing all elements in $\cL_{0,4}$ that are invariant with respect to the action of $\UQ$, \ie 
\be
\cA_{0,4}=\{A\in\cL_{0,4}|\xi(A)=A\epsilon(\xi),\, \xi\in\UQ\}\,,
\ee
where $\epsilon:\UQ\rightarrow\bC$ is the counit of $\UQ$. 
It can be shown that elements of $\cA_{0,4}$ are linear combinations of the form \cite{Alekseev:1995rn} 
\be
\tr_{q}^{I}(C[I_1I_2I_3I_4|I]\bM_{1}^{I_1}\bM_{2}^{I_2}\bM_{3}^{I_3}\bM_{4}^{I_4}C[I_1I_2I_3I_4|I]^*)\,,
\label{eq:generator_A04}
\ee
where $\tr_q^I$ is the quantum trace and $C[I_1I_2I_3I_4|I]$ is the intertwiner that maps the tensor product of representation spaces of $\UQ$ to another representation space of $\UQ$, \ie $C[I_1I_2I_3I_4|I]:V^{I_1}\otimes V^{I_2}\otimes V^{I_3}\otimes V^{I_4} \to V^{I}$ and $C[I_1I_2I_3I_4|I]^*:V^{I}\to V^{I_1}\otimes V^{I_2}\otimes V^{I_3}\otimes V^{I_4}$ is a dual map\footnote{$C[I_1I_2I_3I_4|I]^*$ is called the dual map of $C[I_1I_2I_3I_4|I]$ only in a loose way. Strictly speaking, $C[I_1I_2I_3I_4|I]^*\circ C[I_1I_2I_3I_4|I]\neq \id_{V^I}$. See \eg \cite{Han:2023wiu} for more details.}. 
As a special case, a {\it central element} $c^I_\nu$ defined as
\be
c^I_\nu :=\kappa^I\tr_q^I(\bM_\nu^I)\,,\quad
\kappa^I=q^{-\frac{1}{2}I(I+1)}\,,
\label{eq:central_element}
\ee
and it is an element in $\cA_{0,4}$. 

The moduli algebra $\mathfrak{M}^{K_1,K_2,K_3,K_4}_{0,4}$ is defined from $\cA_{0,4}$ with the help of the {\it quantum character} $\chi^I_{\nu}$, which we now define.  
Firstly, when $q$ is a root-of-unity, one can define a symmetric and invertible ``$S$-matrix'' $S_{IJ}$ in terms of the $\cR$-matrix of $\UQ$, denoted as $R\equiv\sum_a R^{(1)}_a\otimes R^{(2)}_a \in\UQ\otimes\UQ$ and $R'\equiv\sum_a R^{(2)}_a\otimes R^{(1)}_a$: 
\be
S_{IJ}:=\cN (\tr^I_q\otimes\tr^J_q)(R'R)\,,\quad
\cN=\f{1}{\lb\sum_I d_I^2\rb^\f12}\,,
\label{eq:S-matrix_def}
\ee
{where $d_I=[2I+1]_q$ is the quantum dimension of $V^I$ with $[n]_q:=\f{q^{\frac{n}{2}}-q^{-\frac{n}{2}}}{q^{\frac{1}{2}}-q^{-\frac{1}{2}}}$ being a quantum number.} 
The summation in \eqref{eq:S-matrix_def} runs through all the physical representations $I=0,\f12,\cdots,\f{k}{2}$. 
Indeed, if $I$ can take up to $\infty$, which is the case for $q$ not a root-of-unity, $\cN\rightarrow 0$ hence $S_{IJ}$ is ill-defined.  
The $S$-matrices satisfy the following properties \cite{Frohlich:1990ww,Alekseev:1994au}.
\be
S_{IJ}=S_{JI}\,,\quad S_{0J}=\mathcal{N}d_{J}\,,\quad
\sum_{J}S_{IJ}S_{JK}=\delta_{IK}\,,\quad
\sum^{u(I,J)}_{K=|I-J|}S_{KL}=\f{S_{JL}S_{IL}}{\cN d_{L}}\,,
\label{eq:S_matrix}
\ee
where $u(I,J)=\min(I+J,k-I-J)$.
The properties above show that the inverse element of the $S$-matrix is $S_{IJ}$ itself. 

With the help of the $S$-matrix, the {\it quantum characters} $\chi_\nu^J$ in $\cA_{0,4}$ with $\nu=0,1,2,3,4$ is defined as 
\be
\chi_{\nu}^{J}=\cN d_{J}\sum_K S_{JK}c_{\nu}^{\Bar{K}}\equiv\mathcal{N}d_J\sum_K S_{J\Bar{K}}c^{K}_{\nu}\,,
\label{eq:character_def}
\ee
where $\bar{J}$ denotes the dual representation of representation $J$. 
A quantum character is indeed a central element. It is also easy to prove that it is also an orthogonal projector in $\cA_{0,4}$ satisfying
\be
\chi_{\nu}^{I}\chi_{\nu}^{J}=\delta_{IJ}\chi_{\nu}^{I}\,,\quad
(\chi_{\nu}^{J})^{*}=\chi_{\nu}^{J}\,.
\ee
Specially, for $\nu=0$, $\bM_{0}^{J}:=\kappa_{J}^{3}\bM_{4}^{J}\bM_{3}^{J}\bM_{2}^{J}\bM_{1}^{J}$ is the quantum monodromy around all four punctures and  $c^J_0:=\kappa_J\tr_{q}^{J}(\bM_{0}^{J})$. Then, the corresponding quantum character $\chi_0^0$ is. 
\be
\chi_0^0:= \cN d_0 \sum_J S_{0J} c^{\bar{J}}_0= \cN^2\sum_J d_J\kappa_J^4 \tr_q^J\lb\bM_{4}^{J}\bM_{3}^{J}\bM_{2}^{J}\bM_{1}^{J}\rb\,.
\label{eq:bM_0}
\ee

We are now ready to define the moduli algebra for the root-of-unity case.  
The moduli algebra $\fM_{0,4}^{K_\nu}$ on a four-punctured sphere, each puncture of which is associated with a physical representation $K_{\nu}\,(\nu=1,2,3,4)$, is a $*$-algebra defined as \cite{Alekseev:1995rn}
\be
\fM_{0,4}^{K_{\nu}}:=\chi_{0}^{0}\chi_{1}^{K_{1}}\chi_{2}^{K_{2}}\chi_{3}^{K_{3}}\chi_{4}^{K_{4}}\cA_{0,4}\,.
\label{eq:def_moduli_algebra}
\ee
The expression above means that each element in $\fM_{0,4}^{K_{\nu}}$ is obtained by an element in the invariant algebra $\cA_{0,4}$ multiplied by the five quantum characters
$\chi_{0}^{0}\,,\,\chi_{1}^{K_{1}}\,,\,\chi_{2}^{K_{2}}\,,\,\chi_{3}^{K_{3}}\,,\,\chi_{4}^{K_{4}}\in\cA_{0,4}$.
 
The representation of $\cL_{0,4}$ is realized on the tensor product space $V^{K_1}\otimes V^{K_2}\otimes V^{K_3}\otimes V^{K_4}$ which admits a decomposition\footnote{The decomposition \eqref{eq:decompose_reps} is based on the semi-simplicity of the Hopf algebra, which is broken for $q$ a root-of-unity. However, in this case, $\UQ$ is canonically associated with a truncated Hopf algebra which admits semi-simplicity. \eqref{eq:decompose_reps} is realized through such the construction. See \cite{Alekseev:1995rn,MACK1992185} and \cite{Han:2023wiu} for more details. } 
\be
V^{K_1}\otimes V^{K_2}\otimes V^{K_3}\otimes V^{K_4}=\bigoplus_{J} V^J\otimes W^J(K_1,K_2,K_3,K_4)\,,
\label{eq:decompose_reps}
\ee
where $J$ runs through all the admissible physical representations and $W^J(K_1,K_2,K_3,K_4)$ is the multiplicity space, which can be shown to be the carrier space of $\cA_{0,4}$ \cite{Alekseev:1995rn}.

The representation of $\fM^{K_\nu}_{0,4}$ is realized in the invariant subspace $W^0(K_1,K_2,K_3,K_4)$ onto which the five quantum characters project. 

Importantly, 
given four punctures labeled by representations $K_1,\cdots,K_4$ respectively, $W^0(K_1,K_2,K_3,K_4)$ is the {\it only} irreducible $*$-representation space of the moduli algebra $\fM_{0,4}^{K_\nu}$ \cite{Alekseev:1995rn}. The representation space of the moduli algebra can be understood as the quantization of the moduli space of flat connections that satisfy the closure condition. In this sense, we call the intertwiner space $W^0(K_1,K_2,K_3,K_4)$ the {\it physical Hilbert space} as it is the solution space to the ``quantum closure condition''.

\section{Fusion algebra from quantum monodromies around two punctures }
\label{sec:intro}

In the previous section, we have only considered the quantum monodromies around single punctures. However, the length coordinate that we are interested in in this paper is the trace of the monodromy around two punctures, whose quantization is still a quantum monodromy but the algebra it forms needs to be further investigated. 
Quite neatly, we will see that the matrix element of this quantum monodromy still forms a loop algebra and its gauge-invariant elements form a fusion algebra, also called a Verlinde algebra. 

To simplify the calculation, we fix a linear order on the loops around single punctures such that $\ell_1 \prec \ell_2 \prec \ell_3 \prec \ell_4$ by introducing a {\it cilium} between $\ell_1$ and $\ell_4$ as illustrated in fig. \ref{fig:ciliated_graph}. 

Let $\ell=\ell_m \ell_{(m-1)}$ with $m=2,3,4$ be a composed loop formed by two loops neighbouring in the linear order. The (quantum) monodromy $\bM^I_{\ell}$ along $\ell$ surrounds the $m$-th and the $(m-1)$-th punctures and it can be expressed in terms of monodromies $\bM^I_{\ell_m}$ and $\bM^I_{\ell_{m-1}}$ as 
\be
\bM_\ell^{I}=\kappa_I \bM_{\ell_{m}}^I \bM_{\ell_{m-1}}^I\,.
\label{eq:M_ell}
\ee

Let us first study the properties of $\bM_\ell^{I}$. In the following, we will use extensively the standard notations that $\MlIl  :=\bM_{\ell}^{I}\otimes \Id_{V^J}$ and $\MlJr  :=\Id_{V^I}\otimes \bM_{\ell}^{J}$ and similarly for other quantum matrices with a number above. 

\begin{prop}
\label{prop:loop_algebra}

The matrix elements of the quantum monodromy $\{\bM_\ell^{I}\}_I$ defined by \eqref{eq:M_ell}, where $I$ runs through all the physical representations of $\UQ$, generate a loop algebra $\cL_{0,1}$ by satisfying the following defining exchange relations of $\cL_{0,1}$. 
 
\begin{subequations}
\begin{align}
\MlIl  R^{IJ}\MlJr  &=\sum_K C[IJ|K]^{*}\bM_{\ell}^{K}C[IJ|K].
\label{eq:M_ell_property_1}\\
(R^{-1})^{IJ}\MlIl  R^{IJ}\MlJr  &=\MlJr  (R')^{IJ}\MlIl  ((R')^{-1})^{IJ}\;, 
\label{eq:M_ell_property_2}\\
(\bM^I_\ell)^*&=\sigma_\kappa\lb R^I(\bM^I_\ell)^{-1}(R^{-1})^I\rb\;,
\label{eq:M_ell_property_3}
\end{align}
\label{eq:loop_exchange_relations}
\end{subequations}
where $\sigma_{\kappa}(\bM_\ell^I):=\kappa^{-1}\bM^I_\ell\kappa$ with $\kappa\in\UQ$ defined by the ribbon element $v\in\UQ$ through $\kappa^2=v$ and $R^I=(\rho^I\otimes \Id)R\in \End(V^I)\otimes \UQ$. 
\end{prop}
\begin{proof} 
Note that the matrix elements of $\{\bM_{\ell_m}^I,\bM_{\ell_{m-1}}^I\}_I$ are generators of $\cL_{0,4}$. By definition, they satisfy the following exchange relations as $\ell_{m-1}\prec\ell_m$ \cite{Alekseev:1994au,Alekseev:1995rn} (see also \cite{Han:2023wiu} and Lemma \ref{lemma:commu_linear_order} in Appendix \ref{app:detail_calc}).
\begin{subequations}
\begin{align}
(R^{-1})^{IJ}\MlmmIl R^{IJ}\,\MlmJr &= \MlmJr(R^{-1})^{IJ}\MlmmIl R^{IJ}\,,
\label{eq:commutation_rep_1}\\
(R')^{IJ}\MlmIl(R'^{-1})^{IJ}\,\MlmmJr &= \MlmmJr(R')^{IJ}\MlmIl(R'^{-1})^{IJ}\,.
\label{eq:commutation_rep_2}
\end{align}
\label{eq:graph_exchange_relations_lmlmm}
\end{subequations}
On the other hand, the matrix elements of $\{\bM_{\ell_m}^I\}_I$ (\resp $\{\bM_{\ell_{m-1}}^I\}_I$) generate $\cL_{0,1}$, so the exchange relations \eqref{eq:loop_exchange_relations} also hold when one replaces $\bM_\ell$ to $\bM_{\ell_m}$ or $\bM_{\ell_{m-1}}$ therein. Then 
\eqref{eq:loop_exchange_relations} can all be derived through direct calculation.  
Firstly, the {\it l.h.s.} of \eqref{eq:M_ell_property_1} is expanded as 
\be
\begin{split} 
\MlIl R^{IJ} \MlJr 
&=\kappa_I\kappa_{J}\MlmIl\MlmmIl R^{IJ} \MlmJr\MlmmJr\\
&=\kappa_I\kappa_{J}\MlmIl R^{IJ}\MlmJr  (R^{-1})^{IJ}\MlmmIl  R^{IJ}\MlmmJr \\ 
&=\kappa_I\kappa_{J}\sum_{K}C[IJ|K]^{*}\bM_{\ell_m}^KC[IJ|K](R^{-1})^{IJ}\sum_{L}C[IJ|L]^{*}\bM_{\ell_{m-1}}^LC[IJ|L]\\
&=\sum_{K}C[IJ|K]^{*}\kappa_K \bM_{\ell_m}^K\bM_{\ell_{m-1}}^KC[IJ|K]\\
&=\sum_{K}C[IJ|K]^{*}\bM_{l}^KC[IJ|K]\,,
\end{split}
\label{eq:M_ell_property_proof_1}
\ee
where \eqref{eq:commutation_rep_1} is used to obtain the second line, \eqref{eq:M_ell_property_1} is used twice to obtain the third line, one for $M_{\ell_m}^I$ and one for $M_{\ell_{m-1}}^J$, and the normalization condition of Clebsch-Gordan maps, i.e. $C[IJ|K](R^{-1})^{IJ}C[IJ|L]^{*}=\frac{\kappa_K}{\kappa_I\kappa_J}\delta_{KL}$ is used to obtain the second last line. 

We next prove \eqref{eq:M_ell_property_2}.
Explicitly, 
\be
\begin{split}
(R^{-1})^{IJ}\MlIl  R^{IJ}\MlJr  &=\kappa_{I}\kappa_{J}(R^{-1})^{IJ}\MlmIl \MlmmIl R^{IJ}\MlmJr \MlmmJr \\
&=\kappa_{I}\kappa_{J}(R^{-1})^{IJ}\MlmIl R^{IJ}\MlmJr  (R^{-1})^{IJ}\MlmmIl  R^{IJ}\MlmmJr \\
&=\kappa_{I}\kappa_{J}\MlmJr  (R')^{IJ}\MlmIl ((R')^{-1})^{IJ}\MlmmJr  (R')^{IJ}\MlmmIl ((R')^{-1})^{IJ}\\
&=\kappa_{I}\kappa_{J}\MlmJr  \MlmmJr  (R')^{IJ}\MlmIl ((R')^{-1})^{IJ}(R')^{IJ}\MlmmIl ((R')^{-1})^{IJ}\\
&=\MlJr  (R')^{IJ}\MlIl  ((R')^{-1})^{IJ}\;,\\
\end{split}
\label{eq:M_ell_property_proof_2}
\ee
where \eqref{eq:commutation_rep_1} is used to obtain the second line, \eqref{eq:M_ell_property_2} is used twice, one for $M_{\ell_m}^I$ and one for $M_{\ell_{m-1}}^J$, to obtain the third line and \eqref{eq:commutation_rep_2} is used to obtain the fourth line. 

{Lastly, we use the property of the $*$-structure that $(AB)^*=B^*A^*$ to derive \eqref{eq:M_ell_property_3} and get
\be
\begin{split}
\lb\kappa_I\bM^I_{\ell_m}\bM^I_{\ell_{m-1}}\rb^*&=\kappa^{-1}_I(\bM^I_{\ell_{m-1}})^*(\bM^I_{\ell_m})^*
 =\kappa^{-1}\kappa^{-1}_IR^I(\bM^I_{\ell_{m-1}})^{-1}(R^{-1})^IR^I(\bM^I_{\ell_m})^{-1}(R^{-1})^I\kappa
 =\sigma_{\kappa}\lb R^I \bM^I_{\ell} (R^{-1})^I\rb\;,
\end{split}
\label{eq:M_ell_property_proof_3}
\ee}
where \eqref{eq:M_ell_property_3} is used for both $\bM^I_{\ell_m}$ and $\bM^I_{\ell_{m-1}}$ to obtain the second equation. 
\end{proof}
\begin{figure}
\begin{subfigure}[t]{0.3\linewidth}
\begin{tikzpicture}
[scale=1.2]
\draw[blue,thick] (0,0) -- (0,0.5);
\draw[blue,thick] (1.70,0) -- (2.2,0);
\draw[blue,thick] (-1.70,0) -- (-2.2,0);
\draw [red, postaction={decorate},decoration={markings,mark=at position 0.8 with {\arrow[scale=1.5,>=stealth]{>}}}] (0,0) .. controls (-2.25,2.15) and(-2.25,-2.15) .. (0,0) ;
\draw [red, postaction={decorate},decoration={markings,mark=at position 0.25 with {\arrow[scale=1.5,>=stealth]{>}}}] (0,0) .. controls (-2.25,2.15) and(-2.25,-2.15) .. (0,0) ;
\draw[red, postaction={decorate},decoration={markings,mark=at position 0.8 with {\arrow[scale=1.5,>=stealth]{>}}}] (0,0) .. controls (2.25,-2.15) and (2.25,2.15) .. (0,0) ;
\draw[red, postaction={decorate},decoration={markings,mark=at position 0.25 with {\arrow[scale=1.5,>=stealth]{>}}}] (0,0) .. controls (2.25,-2.15) and (2.25,2.15) .. (0,0) ;
\filldraw[color=black!80, fill=black](0,0) circle (0.08);
\draw (0,0) node[anchor=north]{$x$};
\filldraw[color=black!80, fill=black](1.70,0) circle (0.08);
\draw (1.7,0) node[anchor=east]{$y$};
\draw (1.7,0.2) node[anchor=west]{$\ell_{e'}$};
\filldraw[color=black!60, fill=black](-1.70,0) circle (0.08);
\draw (-1.7,0) node[anchor=west]{$y'$};
\draw (-1.7,0.2) node[anchor=east]{$\ell_e$};
\draw (-1,0.8) node[thick]{C};
\draw (-1,-0.8) node[thick]{$C'$};
\draw (1,0.8) node[thick]{$\widetilde{C}'$};
\draw (1,-0.8) node[thick]{$\widetilde{C}$};
\filldraw[color=black, fill=white, very thick](0.85,0) circle (0.12);
\filldraw[color=black, fill=white, very thick](-0.850,0) circle (0.12);
\end{tikzpicture}
\caption{} 
\label{fig_decomp_add_twovertices}
\end{subfigure}
\begin{subfigure}[t]{0.3\linewidth}
 \begin{tikzpicture}
     \pun{0,0}{0.55,0.42}{1,0.75}{2,1.7}{2,-1.7}{1,-0.75}{0.55,-0.42}; 
 \pun{0,0}{-0.42,0.55}{-0.75,1}{-1.7,2}{1.7,2}{0.75,1}{0.42,0.55};
 \pun{0,0}{-0.55,-0.42}{-1,-0.75}{-2,-1.7}{-2,1.7}{-1,0.75}{-0.55,0.42};
 \pun{0,0}{0.42,-0.55}{0.75,-1}{1.7,-2}{-1.7,-2}{-0.75,-1}{-0.42,-0.55}; 
 \draw[thick,blue] (0,0) -- (-0.5,0.5);
 \draw (-0.55,0.55) node[anchor=south east]{cilium};
 \filldraw[black] (0,0) circle(2pt);
 \draw (0,1.7) node[anchor=south] {$\ell_1$};
 \filldraw[black] (0,1.2) circle(1pt);
 \draw (0,-1.7) node[anchor=north] {$\ell_3$};
 \filldraw[black] (0,-1.2) circle(1pt);
  \draw (1.7,0) node[anchor=west] {$\ell_2$};
 \draw (-1.7,0) node[anchor=east] {$\ell_4$};
 \filldraw[color=black, fill=white, very thick](1.2,0) circle (0.12);
\filldraw[color=black, fill=white, very thick](-1.2,0) circle (0.12);
\filldraw[color=black, fill=white, very thick](0,-1.2) circle (0.12);
\filldraw[color=black, fill=white, very thick](0,1.2) circle (0.12);
 \end{tikzpicture} 
 \caption{}
\label{fig:ciliated_graph}
\end{subfigure}
\caption{{\it (a)} Two loops ({\it in red}), each enclosing a single puncture ({\it black hollow circle}). Both loops $\ell_{e}$ and $\ell_{e'}$ start and end in the same vertex $x$. $\ell_{e}$ and $\ell_{e'}$ are decomposed into curves $C$, $C'$, and $\widetilde{C}$, $\widetilde{C}'$ respectively with two additional vertices $y'$ and $y$ added to the loops. A cilium ({\it in blue}) is added to each vertex to fix the linear order, e.g. $\widetilde{C}'\prec -\widetilde{C}\prec C' \prec -C$ at vertex $x$. {\it (b)} Fundamental group generators $\{\ell_\nu\}_{\nu=1,\cdots,4}$ for a four-punctured sphere. A cilium ({\it in blue}) is added between loops $\ell_1$ and $\ell_4$ to fix the linear order $\ell_1\prec\ell_2\prec\ell_3\prec\ell_4$.}
\end{figure}
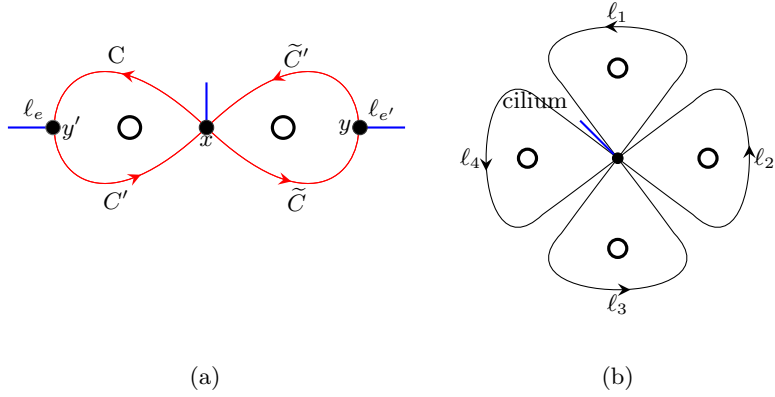

We can also define the quantum monodromy along $\ell$ in an inverse direction (relative to the orientation of the manifold), denoted as $\bM_{-\ell}$, through 
\be
\bM_{\ell}\bM_{-\ell}= \bM_{-\ell}\bM_{\ell}=e\,,
\ee
where $e$ is the identity element of $\UQ$. 

Moreover, the matrix elements of $\bM_\ell^I$ together with those of the quantum monodromies along loops not included in $\ell$ generate a graph algebra. If other loops are along single punctures, the graph algebra is $\cL_{0,3}$. This can be shown by the following proposition.

\begin{prop}
The quantum monodromies $\bM^I_\ell$  around two punctures satisfy the following defining relations of a graph algebra $\cL_{0,3}$ 
 \be
 \begin{split}
 (R^{-1})^{IJ}\MlIl R^{IJ}\MlmuJr &=\MlmuJr (R^{-1})^{IJ}\MlIl R^{IJ}\,, \quad \text{if }\, \ell \prec \ell_{\mu}\;,\\
 (R')^{IJ}\MlIl (R'^{-1})^{IJ}\MlmuJr &=\MlmuJr (R')^{IJ}\MlIl (R'^{-1})^{IJ}\,, \quad  \text{if }\,  \ell \succ \ell_{\mu}\;.
 \label{eq:graph_exchange_relations}
 \end{split}
 \ee
\end{prop}
\begin{proof}
 Since the loop $\ell$ contains both loops $\ell_{m}$ and $\ell_{m-1}$, both $\ell_{m}$ and $\ell_{m-1}$ have the same linear order with respect to loop $\ell_{\mu}$.  
 Both monodromies $\bM^I_{\ell_{m}}$ and $ \bM^I_{\ell_{m-1}}$ satisfy the commutation  relations of $\cL_{0,4}$. The commutation relation of $\bM^I_{\ell}$ and $\bM^J_{\ell_{\mu}}$ can be derived by the use of the commutation relations of $\cL_{0,4}$. That is, \eqref{eq:graph_exchange_relations} holds when one replaces $\bM_\ell$ to $\bM_{\ell_m}$ or $\bM_{\ell_{m-1}}$.  Therefore, when $\ell\prec \ell_\mu$, $\ell_m\prec\ell_\mu$ and $\ell_{m-1}\prec\ell_\mu$, and hence
\be
\begin{aligned}
(R^{-1})^{IJ}\MlIl  R^{IJ}\MlmuJr  &=(R^{-1})^{IJ}\kappa_{I}\MlmIl  \MlmmIl  R^{IJ}\MlmuJr  \\
&=(R^{-1})^{IJ}\kappa_{I}\MlmIl  R^{IJ}\MlmuJr  (R^{-1})^{IJ}\MlmmIl  R^{IJ}\\
&=(R^{-1})^{IJ}\kappa_{I}R^{IJ}\MlmuJr  (R^{-1})^{IJ}\MlmIl  R^{IJ}(R^{-1})^{IJ}\MlmmIl  R^{IJ}\\
&=\MlmuJr  (R^{-1})^{IJ}\kappa_{I}\MlmIl  \MlmmIl  R^{IJ}\\
&=\MlmuJr  (R^{-1})^{IJ}\MlIl  R^{IJ}\;.
\end{aligned}
\ee
The exchange relation for $\ell\succ \ell_{\mu}$ is derived in the same way. 
$\ell\succ\ell_{\mu}$ implies $\ell_{m}\succ \ell_{\mu}$ and $\ell_{m-1}\succ \ell_{\mu}$. Then
\be
\begin{aligned}
(R')^{IJ}\MlIl  ((R')^{-1})^{IJ}\MlmuJr  &=(R')^{IJ}\kappa_{I}\MlmIl  \MlmmIl  ((R')^{-1})^{IJ}\MlmuJr  \\
&=(R')^{IJ}\kappa_{I}\MlmIl  (R'^{-1})^{IJ}\MlmuJr  (R')^{IJ}\MlmmIl  \left((R')^{-1}\right)^{IJ}\\
&=(R')^{IJ}\kappa_{I}(R'^{-1})^{IJ}\MlmuJr  (R')^{IJ}\MlmIl  \left((R')^{-1}\right)^{IJ}(R')^{IJ}\MlmmIl  \left((R')^{-1}\right)^{IJ}\\
&=\MlmuJr  (R')^{IJ}\kappa_{I}\MlmIl  \MlmmIl  \left((R')^{-1}\right)^{IJ}\\
&=\MlmuJr  (R')^{IJ}\MlIl  \left((R')^{-1}\right)^{IJ}\;.
\end{aligned}
\ee
\end{proof}

Given the monodromy $\bM_{\ell}^{I}$, an element $c_{\ell}^{I}:=\kappa_I \tr_q^I(\bM_{\ell}^I)$ can be constructed for each physical representation. Different from the central elements \eqref{eq:central_element} for a quantum monodromy around a single puncture, $c_{\ell}^{I}$ is not central. Nevertheless, the set $\{c_\ell^I\}_I$ has similar properties as the central elements, which are collected in Lemma \ref{lemma:c_x} and Proposition \ref{prop:central_element} below.
\begin{lemma}
 The element $c^I_\ell$ is independent of the choice of starting point.  
 \label{lemma:c_x}
\end{lemma}
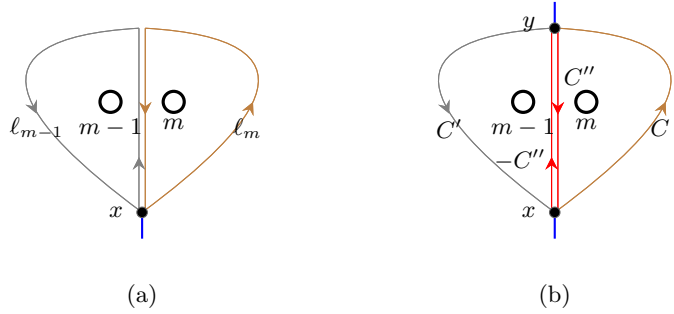
\begin{figure}[h!]
\begin{subfigure}[t]{0.3\linewidth}
\begin{tikzpicture}
[scale=1.4]
\draw[blue,thick] (0,0) -- (0,-0.25);
\draw[brown, postaction={decorate},decoration={markings,mark=at position 0.5 with {\arrow[scale=1.5,>=stealth]{>}}}] (0,0) .. controls (1.15,0.85) and (1.75,1.65) .. (0.03,1.75) ;
\draw[brown, postaction={decorate},decoration={markings,mark=at position 0.5 with {\arrow[scale=1.5,>=stealth]{>}}}] (0,0) .. controls (1.15,0.85) and (1.75,1.65) .. (0.03,1.75) ;
\draw[gray, postaction={decorate},decoration={markings,mark=at position 0.5 with {\arrow[scale=1.5,>=stealth]{<}}}] (0,0) .. controls (-1.15,0.85) and (-1.75,1.65) .. (-0.03,1.75) ;
\draw[gray, postaction={decorate},decoration={markings,mark=at position 0.5 with {\arrow[scale=1.5,>=stealth]{<}}}] (0,0) .. controls (-1.15,0.85) and (-1.75,1.65) .. (-0.03,1.75) ;
\draw [brown, postaction={decorate},decoration={markings,mark=at position 0.6 with {\arrow[scale=1.5,>=stealth]{<}}}] (0.03,0) -- (0.03,1.75); 
\draw [brown, postaction={decorate},decoration={markings,mark=at position 0.6 with {\arrow[scale=1.5,>=stealth]{<}}}] (0.03,0) -- (0.03,1.75);
\draw [gray, postaction={decorate},decoration={markings,mark=at position 0.3 with {\arrow[scale=1.5,>=stealth]{>}}}] (-0.03,0) -- (-0.03,1.75); 
\draw [gray, postaction={decorate},decoration={markings,mark=at position 0.3 with {\arrow[scale=1.5,>=stealth]{>}}}] (-0.03,0) -- (-0.03,1.75);
\filldraw[color=black!60, fill=black](0,0) circle (0.05);
\draw (-0.1,0) node[anchor=east]{$x$};
\filldraw[color=black, fill=white, very thick](0.3,1.050) circle (0.1);
\filldraw[color=black, fill=white, very thick](-0.3,1.050) circle (0.1);
\draw (0.3, 0.9) node[below=-0.1]{$m$};
\draw (-0.3, 0.9) node[below=-0.1]{$m-1$};
\draw (-1,0.8) node[thick]{$\ell_{m-1}$};
\draw (1,0.8) node[thick]{$\ell_{m}$};
\end{tikzpicture}
\caption{}
\label{fig:_cycle_l}
\end{subfigure}
\begin{subfigure}[t]{0.3\linewidth}
\begin{tikzpicture}
[scale=1.4]
[scale=2]
\draw[blue,thick] (0,0) -- (0,-0.25);
\draw[blue,thick] (0,1.75) -- (0,2.0);
\draw[ brown, postaction={decorate},decoration={markings,mark=at position 0.5 with {\arrow[scale=1.5,>=stealth]{>}}}] (0,0) .. controls (1.15,0.85) and (1.75,1.65) .. (0.03,1.75) ;
\draw[brown, postaction={decorate},decoration={markings,mark=at position 0.5 with {\arrow[scale=1.5,>=stealth]{>}}}] (0,0) .. controls (1.15,0.85) and (1.75,1.65) .. (0.03,1.75) ;
\draw[gray, postaction={decorate},decoration={markings,mark=at position 0.5 with {\arrow[scale=1.5,>=stealth]{<}}}] (0,0) .. controls (-1.15,0.85) and (-1.75,1.65) .. (-0.03,1.75) ;
\draw[gray, postaction={decorate},decoration={markings,mark=at position 0.5 with {\arrow[scale=1.5,>=stealth]{<}}}] (0,0) .. controls (-1.15,0.85) and (-1.75,1.65) .. (-0.03,1.75) ;
\draw [red, postaction={decorate},decoration={markings,mark=at position 0.6 with {\arrow[scale=1.5,>=stealth]{<}}}] (0.03,0) -- (0.03,1.75); 
\draw [red, postaction={decorate},decoration={markings,mark=at position 0.6 with {\arrow[scale=1.5,>=stealth]{<}}}] (0.03,0) -- (0.03,1.75);
\draw [red, postaction={decorate},decoration={markings,mark=at position 0.3 with {\arrow[scale=1.5,>=stealth]{>}}}] (-0.03,0) -- (-0.03,1.75); 
\draw [red, postaction={decorate},decoration={markings,mark=at position 0.3 with {\arrow[scale=1.5,>=stealth]{>}}}] (-0.03,0) -- (-0.03,1.75);
\filldraw[color=black!60, fill=black](0,0) circle (0.05);
\draw (-0.1,0) node[anchor=east]{$x$};
\filldraw[color=black!60, fill=black](0,1.750) circle (0.05);
\draw (-0.1,1.790) node[anchor=east]{$y$};
\draw (0,1.30) node[thick,anchor=west]{$C''$};
\draw (0,0.50) node[thick,anchor=east]{$-C''$};
\filldraw[color=black, fill=white, very thick](0.3,1.050) circle (0.1);
\filldraw[color=black, fill=white, very thick](-0.3,1.050) circle (0.1);
\draw (0.3, 0.9) node[below=-0.1]{$m$};
\draw (-0.3, 0.9) node[below=-0.1]{$m-1$};
\draw (-1,0.8) node[thick]{$C'$};
\draw (1,0.8) node[thick]{$C$};
\end{tikzpicture}
\caption{}
\label{fig_c}
\end{subfigure}
\caption{{\it (a)} 
Loop $\ell$ starting at vertex $x$ as a composition of loops $\ell_m$ ({\it in brown}) and $\ell_{m-1}$ ({\it in gray}), which enclose the $m$-th and the $(m-1)$-th punctures ({\it illustrated as black hollow circles}) respectively.  
{\it (b)} {Loops $\ell_m$ and $\ell_{m-1}$ are decomposed into curves $C$ ({\it in brown}) and $C''$ ({\it in red}), and $-C''$ ({\it in red}) and $C'$ ({\it in gray}) respectively by adding an additional vertex $y$ to the loop $\ell$.} A cilium ({\it in blue}) is added at each vertex, pointing out of $\ell$, to fix the linear order.}
\end{figure}
\begin{proof} 
The loop $\ell$ is composed of loops $\ell_{m}$ and $\ell_{m-1}$ as shown in fig.\ref{fig:_cycle_l}. We break the loop $\ell_{m}$ into $C$ and $C''$ and $\ell_{m}$ into $-C''$ and $C'$ by adding another (randomly chosen) vertex $y$ to $\ell$ as illustrated in fig.\ref{fig_c}. The quantum monodromies $\bM^I_{\ell_{m}}$ and $\bM^I_{\ell_{m-1}}$ can be expressed in terms of quantum holonomies along curves $C,C'$ and $C''$: $\bM^I_{\ell_{m}}=\kappa^{-1}_I\bold{U}^I_{C}\bold{U}^I_{C''}$ and $\bM^I_{\ell_{m-1}}=\kappa^{-1}_I\bold{U}^I_{-C''}\bold{U}^I_{C'}$. Combining these two expressions we have $\bM^I_{\ell}=\kappa^{-1}_I\bold{U}^I_{C'}\bold{U}^I_{C}$ if starting at $y$ while $\bM^I_{\ell}=\kappa^{-1}_I\bold{U}^I_{C}\bold{U}^I_{C'}$ if starting at $x$. To prove that $c^I_\ell$ is independent of the starting point, it is adequate to prove that 
\be
 \tr_q^I(\kappa^{-1}_I\bold{U}^I_{C'}\bold{U}^I_{C})=\tr_q^I(\kappa^{-1}_I\bold{U}^I_{C}\bold{U}^I_{C'}) \,.
\ee
Classically, the Poisson bracket of holonomies $U^I_{C}$ and $U^I_{C'}$ is \cite{Alekseev:1994pa},
 \be
 \{\ucl,\ucpr\}=\ucpr (r')^{II}_{x}\ucl -\ucl r^{II}_{y}\ucpr  \,,
 \ee
 which is quantized to the commutation relation of $\bold{U}^I_{C}$ and $\bold{U}^J_{C'}$ 
 \be
\stackrel{2\hspace{1em}}{\bold{U}_{C'}^I}  (R')^{II}_{x} \UCl =\UCl  R^{II}_{y} \stackrel{2\hspace{1em}}{\bold{U}_{C'}^I}\,.
 \label{eq:commutation_U}
 \ee
 There is a group-like element $g\in\UQ$ with the intertwining property 
 \be
 gS(\xi)=S^{-1}(\xi)g,\,\forall\,\xi\in\UQ\,,
 \label{eq:g_intertwining}
 \ee
 where $S^{-1}$ is the inverse of the antipode such that $S^{-1}\circ S(\xi)=\xi$. 
Moreover, $g$ can be decomposed as \cite{Alekseev:1994pa,Alekseev:1994au}. 
 \be
 g=u^{-1}v\,,\quad u:=\sum_aS(R_a^{(2)})R_a^{(1)}\,,
 \ee
 where $v$ is the central element $v$ of $\UQ$. Introduce $g$ as a function of points on $\Sfour$ and let $g_x=g(x)$ and $g_y=g(y)$.
We multiply \eqref{eq:commutation_U} by $\gxl\equiv \rho^I(g_x)\otimes \Id_{V^J}$ from the left and $\gyl\equiv \rho^I(g_y)\otimes \Id_{V^J}$ from the right, then obtain 
 \be
\gxl\,\stackrel{2\hspace{1em}}{\bold{U}_{C'}^I}  (R')^{II}_{x} \UCl\, \gyl=\gxl\,\UCl  R^{II}_{y} \stackrel{2\hspace{1em}}{\bold{U}_{C'}^I}\, \gyl\,.
 \ee
  Here we set $R^{-1}=\sum_i R^1_i\otimes R^2_i$ and $(S\otimes id)(R)=R^{-1}$ , $(S^{-1}\otimes id)(R')=R'^{-1}$. Then
 \be
 \begin{aligned}
\gxl\,\stackrel{2\hspace{1em}}{\bold{U}_{C'}^I}  \left(\sum_{i}S(R^2_i)\otimes R^1_i\right)^{II}_{x} \UCl \,\gyl=\gxl\,\UCl  \left(\sum_{j}S^{-1}(R^1_j)\otimes R^2_j\right)^{II}_{y} \stackrel{2\hspace{1em}}{\bold{U}_{C'}^I} \,\gyl\,.
\end{aligned}
\label{eq:gURUg}
 \ee
Using the intertwining relation \eqref{eq:g_intertwining} of the group-like element $g$, \eqref{eq:gURUg} becomes
 \be
\stackrel{2\hspace{1em}}{\bold{U}_{C'}^I}  \left(\sum_{i}S^{-1}(R^2_i)\otimes R^1_i\right)^{II}_{x} \gxl\,\UCl\, \gyl=\gxl\,\UCl \,\gyl \left(\sum_{j}S(R^1_j)\otimes R^2_j\right)^{II}_{y} \stackrel{2\hspace{1em}}{\bold{U}_{C'}^I} \, .
\ee
Reordering the expression produces
 \be
   \stackrel{2\hspace{1em}}{\bold{U}_{C'}^I} \lb\sum_j \stackrel{2\hspace{1.5em}}{(R_{jx}^{(1)})^I}\stackrel{1}{(S^{-1}(R_{jx}^{(2)}))^I}\rb \gxl\,\UCl\, \gyl 
   =\gxl\,\UCl \,\gyl \lb\sum_i \stackrel{1}{(S(R_{iy}^{(1)}))^I}\stackrel{2\hspace{1.5em}}{(R_{iy}^{(2)})^I}\rb \stackrel{2\hspace{1em}}{\bold{U}_{C'}^I}\, .
   \label{eq:UrrgUg}
 \ee
 To proceed, we first multiply
the two matrix components in the last equation and then multiply both sides with $v^2_I$ (which is merely a complex number), then
 take the trace $\tr^I$. Using the fact that $u^I=v^2_I (R_{jx}^{(1)}S^{-1}(R_{jx}^{(2)}))^I=v^2_I(S(R_{iy}^{(1)})R_{iy}^{(2)})^I$, we get
\be
   \tr^I(\bold{U}^I_{C'} u^I_x g^I_x \bold{U}^I_{C} g_y^I)=\tr^I(g_x^I\bold{U}^I_{C}g^I_yu^I_y \bold{U}^I_{C'}) 
\ee
{Finally, since $gu=ug=v$ as $v$ is central and the quantum trace is defined as $\tr^I_q({\bf U}^I):=\tr({\bf U}^I g^I)$, we multiply both side by $\kappa_I^{-1}$ and get (recall $\kappa_I^2=v^I$).}
\be 
 \tr_q^I(\kappa^{-1}_I\bold{U}^I_{C'}\bold{U}^I_{C})=\tr_q^I(\kappa^{-1}_I\bold{U}^I_{C}\bold{U}^I_{C'}) \,.   
\ee
This is equivalent to 
\be
 \tr_q^I(\kappa_I \bM^I_{\ell_{m-1}}\bM^I_{\ell_m})=\tr_q^I(\kappa_I\bM^I_{\ell_m}\bM^I_{\ell_{m-1}})\,.
\ee
\end{proof}
\begin{prop}
\label{prop:central_element}
$c_\ell^{I}$ as an element of $\cL_{0,4}$  satisfies the following properties.
\begin{enumerate}
	\item They are gauge-invariant but not central elements. \\
\item They form a fusion algebra $\mathcal{V}(\ell)$:
\be
c_\ell^Ic_\ell^J=\sum_K N^{IJ}_Kc_\ell^K \;,\quad
 (c_\ell^I)^{*}=c^{\Bar{I}}_\ell \;.
 \label{eq:c_fusion}
 \ee
\end{enumerate}
\end{prop}
\begin{proof}
The gauge-invariant property can be proved by showing that $c^I_{\ell}$ is in the invariant sub-algebra $\mathcal{A}_{0,4}$. Recall that any element in $\mathcal{A}_{0,4}$ can be expressed as a linear combination of the form \eqref{eq:generator_A04}.
In the following ,we are going to show that  
\be
\kappa_J\tr^J_q(\kappa_J\bM^J_{\ell_{m-1}}\bM^J_{\ell_{m}})
= \kappa_J\tr^J_q(C[I0|J]\kappa_I\bM^I_{\ell_{m-1}}\bM^I_{\ell_{m}}C[I0|J]^*)\,,
\label{eq:c_l_A}
\ee
which is a special case of \eqref{eq:generator_A04} when $\bM_{\ell_m}$ and $\bM_{\ell_{m-1}}$ take the same representation while the other two quantum holonomies take trivial representations.  

The Clebsh-Gordon maps $C[I0|J]$ and $C[I0|L]^*$ are intertwiners and commute with $\rho^I(\xi), \forall \,\xi \in \mathcal{U}_q(\mathfrak{su}(2))$: 
\be
\begin{split}
C[I0|J]\rho^I(\xi)&=C[I0|J](\rho^I\otimes \epsilon)\Delta(\xi)=\rho^J(\xi) C[I0|J]\,,\\
C[I0|L]^*\rho^L(\xi)&=(\rho^I\otimes \epsilon)\Delta'(\xi)C[I0|L]^*=\rho^I(\xi_i^{(2)})\epsilon(\xi_i^{(1)})C[I0|L]^*=\rho^I(\xi_i^{(2)}\epsilon(\xi_i^{(1)}))C[I0|L]^*=\rho^I(\xi)C[I0|L]^*\;.
\end{split}
\ee
By Schur's lemma, both $C[I0|L]$ and $C[I0|L]^*$ are equal to identity $e^I$ up to some complex factor. To be consistent with the normalization condition we set for CG maps, the factor must be compatible with the normalisation condition.
Suppose 
$C[I0|J]=\lambda \delta_{I,J}$ and $C[I0|L]^*=\lambda' \delta_{L,I}$,
\be
C[I0|J](R')^{I0}C[I0|L]^*=C[I0|J]C[I0|L]^*= \delta_{J,L}\;,
\ee
where the first equality is obtained by using quasi-triangularity, i.e. $(id \otimes \epsilon)(R')=e$.
From the derivation above we then have $\lambda \lambda'=1$, which shows
\be
C[I0|J]\kappa_I\bM^I_{\ell_{m-1}}\bM^I_{\ell_{m}}C[I0|J]^*=\kappa_J\bM^J_{\ell_{m-1}}\bM^J_{\ell_{m}}.
\ee
Combining the derivation above with Lemma \ref{lemma:c_x}, it is confirmed that $c^I_{\ell}$ is indeed an element of $\mathcal{A}_{0,4}$.

 It can be proved by using Lemma \ref{lemma:trace_identity} in Appendix \ref{app:detail_calc} that elements $c_\ell^I$ still commute with $\bold{M}^L_{\nu},$ for $\nu\neq m,\;m-1$. However, they are not central elements of $\cL_{0,4}$ and can be verified by the fact that $c^I_\ell =\kappa_I \tr_q^I(\kappa_I \bM_{\ell_m}^I \bM^I_{\ell_{m-1}})$ does not commute with $\bM_{\ell_m}^J$ or $\bM_{\ell_{m-1}}^K$, using the defining relations \eqref{eq:loop_exchange_relations} of a loop algebra and those \eqref{eq:graph_exchange_relations} of a graph algebra. 

Lastly, the matrix elements of $\bM_\ell^{I}=\kappa_I \bM_{\ell_m}^I \bM_{\ell_{m-1}}^I$ are the generators of the loop algebra as shown in Proposition \ref{prop:loop_algebra}. Then the fusion rule can be derived with the help of the functionality condition of $\bM^I_\ell$: 
\be
\MlIl R^{IJ}\MlJr=\sum_K C[IJ|K]^*\bM^K_\ell C[IJ|K]\,.
\ee
The derivation and proof for the fusion rule and $*$-structure of $c^I_{\ell}$ are the same as those shown in \cite{Alekseev:1994au} hence we skip it here.
 \end{proof}
  To construct the quantum character of the loop $\ell$, we use the $S$-matrix defined in \eqref{eq:S-matrix_def} to define the quantum character $\chi^I_{\ell}$:
  \be
  \chi^I_\ell:=\mathcal{N}d_I S_{IJ}c^{\Bar{J}}_{\ell}=\mathcal{N}d_I S_{I\Bar{J}}c^{J}_{\ell}\,.
  \ee
  It also satisfies the projecting property just as the quantum character for quantum monodromies around single loops. That is,
  \be
  \chi^I_\ell  \chi^J_\ell=\delta_{IJ}  \chi^J_\ell \quad,\quad (\chi^I_\ell)^{*}=\chi^{I}_\ell\,.
  \ee  
  The detailed proof can also be found in \cite{Alekseev:1994au}. 
  
\medskip

Let us clarify that, although we are working on $\UQ$ with $q$ a root-of-unity in this paper, the above derivations were done firstly in an {unrealistic} case by assuming that there is no truncation for representations and then truncating away the unphysical representation. 
{This means the representations we use above are for the truncated Hopf algebra $\mathcal{U}^T_q(\mathfrak{su}(2))$,  which is canonically associated with $\UQ$ with $q$ a root-of-unity, ensuring that the irreducible representations precisely correspond to the physical representations of $\UQ$ \cite{Alekseev:1994au,Alekseev:1995rn}.} Thanks to the semi-simplicity of $\mathcal{U}^T_q(\mathfrak{su}(2))$, the derivations are neat as above. 
Nevertheless, one can directly work on the representations of $\UQ$ then the above results are still valid with the so-called substitution rule (see \eqref{eq:substitution_rule}) applied. We refer to Appendix \ref{com_relation_truncated} for this approach. 
 
To end this section, we would like to comment on the case when $q$ is a phase but not a root-of-unity, i.e. $|q|=1, q^{k+2}\neq 1$. Almost all the proofs and defining relations we gave in this section still hold in this case except that the quantum character $\chi^I_{\ell}$ is not well-defined due to the fact that the $S$-matrix is ill-defined for $|q|=1,q^{k+2}\neq 1$. More details about $|q|=1,q^{k+2}\neq 1$ case can be found in \cite{Han:2023wiu}.

\section{Representation of the fusion algebra $\mathcal{V}(\ell)$}
\label{sec:rep_fusion}

Having defined the fusion algebra $\mathcal{V}(\ell)$, we proceed to construct the representation of $\mathcal{V}(\ell)$. To be concrete and with no loss of generality, we set $\ell$ to be the loop that encloses the first and second punctures, i.e. $\ell=\ell_2 \ell_{1}$.  
Then $\bM^I_\ell=\kappa_I\bM^I_{\ell_2}\bM^I_{\ell_1}$. 
As mentioned in the preceding section, $\{c^I_\ell\}_I$ form a fusing algebra $\mathcal{V}(\ell)$ but they are not central elements. This means the commutant of $\mathcal{V}(\ell)$ in the $\mathfrak{M}^{K_1,K_2,K_3,K_4}_{0,4}$ is non-trivial, unlike that of $\mathcal{V}(\ell_\nu)$, and that the representations of $\cV(\ell)$ in the multiplicity space $W^0(K_1,K_2,K_3,K_4)$ are reducible. 
In this section, we describe the representation theory of $\cV(\ell)$ (Theorem \ref{theorem:rep_fusion}) which is based on the {\it first pinching theorem} (Lemma \ref{lemma:first_pinching_theorem}) \cite{Alekseev:1995rn} and study the corresponding representation space. 

\begin{lemma}
  The commutant $\mathcal{C}(\mathcal{V}(\ell),\mathfrak{M}^{K_1,K_2,K_3,K_4}_{0,4})$ of $\mathcal{V}(\ell)$ in $\mathfrak{M}^{K_1,K_2,K_3,K_4}_{0,4}$ can be split into the direct sum of products of two moduli algebras, each corresponding to a $3$-punctured sphere. 
\be
\mathcal{C}(\mathcal{V}(\ell),\mathfrak{M}^{K_1,K_2,K_3,K_4}_{0,4})=\bigoplus_J \mathfrak{M}^{K_1,K_2,J}_{0,3}\otimes \mathfrak{M}^{\Bar{J},K_3,K_4}_{0,3}\,,
\ee
where $J$ runs through all the physical irreducible representations of $\UQ$, \ie $\max \lb |K_1-K_2|,|K_3-K_4| \rb\leq J\leq \min\lb u(K_1,K_2),u(K_3,K_4) \rb$ with spacing 1,  and $\bar{J}$ is the dual representation of $J$. However, {for non-admissible $J$, the algebra is trivial since $N^{K_1K_2}_J=0$.}
\label{lemma:first_pinching_theorem}
\end{lemma}
A proof of Lemma \ref{lemma:first_pinching_theorem} can be found in \cite{Alekseev:1995rn}. According to the first pinching theorem, to study the representation theory of $\cV(\ell)$, one can separate $\Sfour$ into a pair of 3-punctured spheres $\Sigma_{0,3}$'s by doing a pants decomposition as illustrated in fig.\ref{fig:pants}, and study the representation theories of the moduli algebra on a $\Sigma_{0,3}$ individually. We make it precise in the following theorem. 

\begin{theorem}
\label{theorem:rep_fusion}
Given any set of physical representations $K_1,K_2,K_3,K_4$, each labeling a puncture of $\Sfour$, there exists a unique irreducible $*$-representation of the fusion algebra $\mathcal{V}(\ell)$ on the space $W^{J}(K_1,K_2)\otimes W^{\bar{J}}(K_3,K_4)$ for each admissible physical representation $J$ of $\UQ$.  
\end{theorem}
\begin{proof}
 The crucial part of the proof is to show that $\mathfrak{M}^{K_1,K_2,K_3,K_4}_{0,4}$ is isomorphic to the algebra $\chi^{K_4}_{0}\prod^{3}_{\nu=1}\chi_\nu^{K_\nu}\cA_{0,3}$, \ie
\be
\mathfrak{M}^{K_1,K_2,K_3,K_4}_{0,4}\cong \chi^{K_4}_{0}\prod^{3}_{\nu=1}\chi_\nu^{K_\nu}\cA_{0,3}\,,
\label{eq:iso_pinching}
\ee
which we now establish. 
Firstly, we realize that $W^{K_4}(K_1,K_2,K_3)$ is isomorphic to $W^{0}(K_1,K_2,K_3,K_4)$ as they are both multiplicity spaces and have the same dimension $\sum_{J}N^{K_1K_2}_J N^{JK_3}_{K_4}$ with $N^{IJ}_K=1$ if $I,J,K$ are all physical representations and they satisfy the triangular inequality while $N^{IJ}_K=0$ otherwise. The former is the representation space of the algebra $\chi^{K_4}_{0}\prod^{3}_{\nu=1}\chi_\nu^{K_\nu}\cA_{0,3}$ while the latter is the representation space of $\mathfrak{M}^{K_1,K_2,K_3,K_4}_{0,4}$. 
The isomorphism of the two representation spaces then leads to the isomorphism of the two algebras because the representations of a moduli algebra are irreducible and faithful. Since the representation of moduli algebra is realized in only one space, the faithfulness and irreducibility of representation do guarantee the isomorphism of moduli algebra \cite{Alekseev:1995rn}. This establishes the validity of \eqref{eq:iso_pinching}.

The unit of the moduli algebra $\mathfrak{M}^{K_1,K_2,K_3,K_4}_{0,4}$ is defined as $\chi^{0}_{0}\prod^{4}_{\nu=1} \chi^{K_\nu}_{\nu}$ \cite{Alekseev:1994au}. Similarly, the unit of $\mathfrak{M}^{K_1,K_2,J}$ and $\mathfrak{M}^{\Bar{J},K_3,K_4}$ is expressed as $\chi^{J}_{0}\chi^{K_1}_1\chi^{K_2}_2$ and $\chi^{\Bar{J}}_{0}\chi^{K_3}_3\chi^{K_4}_4$ respectively. 
By the first pinching theorem, the decomposition of the commutant of the fusion algebra determines the following decomposition of the unit 
\be
e^{K_1,K_2,K_3,K_4}_{0,4}=\sum_J e^{K_1,K_2,J}_{0,3}\otimes e^{\bar{J},K_3,K_4}_{0,3}\,.
\label{eq:decompose}
\ee
 The decomposition is the consequence of the completeness of the characters $\sum_{J}\chi^J_{\ell}=\Id$. 
 {We now show the decomposition \eqref{eq:decompose}.} 
 {Denote the mapping $\phi:\cL_{0,3}\otimes\cL_{0,3}\to \cL_{0,4}$. Then the embedding $\chi^0_0\phi$ as defined in \cite{Alekseev:1995rn} maps every matrix algebra $\mathfrak{M}^{K_1,K_2,J}\otimes\mathfrak{M}^{\Bar{J},K_3,K_4}$ into  $\mathfrak{M}^{K_1,K_2,K_3,K_4}_{0,4}$.} For the {\it l.h.s.} of the decomposition \eqref{eq:decompose}, the unit is $\chi^{0}_{0}\prod^{4}_{\nu=1} \chi^{K_\nu}_{\nu}$, for the {\it r.h.s.} of the decomposition, the embedded unit becomes
\be
\begin{split}
&\sum_{J}\chi^0_0\chi^{J}_{\ell}\chi^{K_1}_1\chi^{K_2}_2 \chi^{\Bar{J}}_{-\ell}\chi^{K_3}_3\chi^{K_4}_4
=\sum_{J}\chi^0_0\chi^{J}_{\ell}\chi^{K_1}_1\chi^{K_2}_2 \chi^{J}_{\ell}\chi^{K_3}_3\chi^{K_4}_4
=\sum_{J}\chi^{J}_{\ell}\chi^{J}_{\ell}\chi^0_0\chi^{K_1}_1\chi^{K_2}_2 \chi^{K_3}_3\chi^{K_4}_4\\
&=\sum_{J}\chi^{J}_{\ell}\chi^0_0\chi^{K_1}_1\chi^{K_2}_2 \chi^{K_3}_3\chi^{K_4}_4
=\chi^0_0\chi^{K_1}_1\chi^{K_2}_2 \chi^{K_3}_3\chi^{K_4}_4\;,
\end{split}
\label{eq:decompose_proof}
\ee
where we have used the completeness property of quantum characters $\sum_J \chi^J_{\ell}=\Id$. This validates \eqref{eq:decompose}.

The element $\chi^J_{\ell}=\mathcal{N}d_JS_{JL}c^{\bar{L}}_{\ell}$ is constructed from $c^I_{\ell}$.
Based on the representation theory of the moduli algebra constructed in \cite{Alekseev:1995rn}, there is a unique irreducible $*$-representation of $\mathfrak{M}^{K_1,K_2,K_3,K_4}_{0,4}$ on the space $W^0(K_1,K_2,K_3,K_4)$, which is the space projected by the unit of $\mathfrak{M}^{K_1,K_2,K_3,K_4}_{0,4}$. 
By the isomorphism \eqref{eq:iso_pinching}, there is also a unique irreducible $*$-representation of the commutant on the space projected by $\chi^{J}_{0}\chi^{K_1}_1\chi^{K_2}_2 \otimes \chi^{\Bar{J}}_{0}\chi^{K_3}_3\chi^{K_4}_4$ for each $J$. This irreducible representation of the commutant is realized on the multiplicity space $W^{J}(K_1,K_2)\otimes W^{\Bar{J}}(K_3,K_4)$ for every $J$. The element $c^I$ is a central element in $\mathfrak{M}^{K_1,K_2,J}$ if it is constructed from the contractible loop, as discussed in \cite{Alekseev:1994au}. In our case, the element $c^I_{\ell}$ is constructed from the contractible loop that encloses the first and second punctures after pants decomposition hence it is a central element in $\mathfrak{M}^{K_1,K_2,J}$. This means $c^I_{\ell}$ commutes with $\sum_J \mathfrak{M}^{K_1,K_2,J}_{0,3}\otimes \mathfrak{M}^{\bar{J},K_3,K_4}_{0,3}$. Therefore, the representation of $c^I_{\ell}$ should also be realized in the multiplicity space $W^{J}(K_1,K_2)\otimes W^{\Bar{J}}(K_3,K_4)$ for every $J$.
\end{proof}

The representation of $\{c_\ell^I\}_I$, as it is realized in the space $\bigoplus_{J}W^{J}(K_1,K_2)\otimes W^{\Bar{J}}(K_3,K_4)$, can be expressed in terms of a basis of the space. We define the basis in the space $W^{J}(K_1,K_2)\otimes W^{\Bar{J}}(K_3,K_4)$ as 
\be
{\mathfrak B}^J\equiv\sum_m e_m^{J}\otimes e_m^{\Bar{J}}:= \sum_{m}(-1)^{J-m}q^{m}e_{m}^{J}\otimes e_{-m}^{J}\,,\quad \max\{|K_1-K_2|,|K_3-K_4|\}\leq J\leq \min\{u(K_1,K_2),u(K_3,K_4)\} \;,
\label{eq:basis_fusion}
\ee
where $e^J_m$ is the orthonormal basis of $V^J$ and $e^{\bar{J}}_m\equiv (-1)^{J-m}q^{m}e_{-m}^{J}$ is the orthonormal basis of $V^{\bar{J}}$\footnote{The definition of $e^{\bar{J}}_m$ is motivated by the scalar product between two orthonormal bases of $V^J$, which is in terms of the $q$-deformed Clebsh-Gordon coefficient $\mat{ccc}{J& K & 0 \\m & n & 0}_q:=\delta_{KJ} \delta_{m,-n} \frac{(-1)^{J-m}q^m}{\sqrt{[2K+1]_q}}$ and maps two vectors to a scalar. 
More precisely, the scalar product of $e^J_m\in V^J$ and $e^K_n\in V^K$ is $\la e^K_n, e^J_m\ra :=  \sqrt{[2J+1]_q} \mat{ccc}{J& K & 0 \\m & n & 0}_q \, e^J_me^K_n = (-1)^{J-m} q^m \delta_{KJ}\delta_{n,-m}$. This motivates us to define $e^{\bar{J}}_m=(-1)^{J-m} q^m e^{J}_{-m}$ so that the scalar product can be expressed as the ``direct product'' of $e^{\bar{K}}_n$ and $e^J_m$. The same way to define a basis in the dual representation was used in \cite{Dupuis:2013lka}.}.

\begin{lemma}
The basis ${\mathfrak B}^J$ of $W^{J}(K_1,K_2)\otimes W^{\Bar{J}}(K_3,K_4)$ as defined in \eqref{eq:basis_fusion} is an invariant state, i.e.
\be
\rho^J\boxtimes \rho^J\lb \xi\rb  {\mathfrak B}^J =  \epsilon(\xi) {\mathfrak B}^J\,,\quad
\forall\,\xi\in\UQ\,,
\label{eq:invariant_basis}
\ee
where $\epsilon$ is the counit of $\UQ$, which is also the trivial representation of $\UQ$, and $\rho^J\boxtimes \rho^J\lb \xi\rb:=\rho^J\otimes \rho^J\lb \Delta(\xi)\rb$.  
\end{lemma}

\begin{proof}
To prove \eqref{eq:invariant_basis} for all $\xi\in\UQ$, it is enough to prove its validity when $\xi$ is any generator of $\UQ$, i.e. for any $\xi\in\{q^{\f{H}{2}}, X,Y\}$. We first note that 
\be
\epsilon(q^{\f{H}{2}})=1\,,\quad 
\epsilon(X)=0=\epsilon(Y)\,,
\ee
and that the actions of the generators on the basis read
\be
\rho^J(q^{\frac{H}{2}})e^J_{m}=q^{m}\,e^J_{m}\,,\quad
\rho^{J}(X)e_{m}^{J}=\sqrt{[J-m]_{q}[J+m+1]_{q}}\,e_{m+1}^{J}\,,\quad
\rho^{J}(Y)e_{m}^{J}=\sqrt{[J+m]_{q}[J-m+1]_{q}}\,e_{m-1}^{J} \,.
\label{eq:rho_generators}
\ee

We first consider $\xi=q^{\frac{H}{2}}$. Then 
\be
\begin{split}
\rho^J\boxtimes \rho^J\lb q^{\f{H}{2}} \rb{\mathfrak B}^J &  =\sum_{m}(-1)^{J-m}q^{m}\rho^J( q^{\frac{H}{2}}) e_{m}^{J}\otimes \rho^J( q^{\frac{H}{2}})e_{-m}^{J}\\
&=\sum_{m}q^{m}q^{-m}(-1)^{J-m}q^{m}e_{m}^{J}\otimes e_{-m}^{J}
  \equiv\sum_{m}(-1)^{J-m}q^{m}e_{m}^{J}\otimes e_{-m}^{J} \,, 
\end{split}
\ee   
which satisfies \eqref{eq:invariant_basis}. 
We next consider $\xi=X$:
\be
\begin{split}
\rho^J\boxtimes \rho^J\lb X \rb{\mathfrak B}^J 
  =&\sum_{m}(-1)^{J-m}q^{m}\lb \rho^J(X)\otimes \rho^J( q^{\frac{H}{2}})+\rho^J(q^{-\frac{H}{2}})\otimes\rho^J( X)\rb (e_{m}^{J}\otimes e_{-m}^{J})\\
  =&\sum_{m=-J}^{J-1}\sqrt{[J-m]_q[J+m+1]_q}q^{-m}(-1)^{J-m}q^{m}e_{m+1}^{J}\otimes e_{-m}^{J}\\
  &+\sum_{m=-J+1}^Jq^{-m}\sqrt{[J+m]_q[J-m+1]_q}(-1)^{J-m}q^{m}e_{m}^{J}\otimes e_{-m+1}^{J}\,.
  \end{split}
  \ee

Replacing the second summation for $m$ in the last expression to $m':=m+1$, the two terms in the last expression cancel. Then $\rho^J\boxtimes \rho^J\lb X \rb  {\mathfrak B}^J =0\equiv\epsilon(X) {\mathfrak B}^J$.

The same analysis for the action of $\Delta(Y)\equiv Y\otimes q^{\frac{H}{2}}+q^{-\frac{H}{2}}\otimes Y$ on ${\mathfrak B}$ gives that $\rho^J\boxtimes \rho^J\lb Y \rb  {\mathfrak B}^J =0\equiv\epsilon(Y) {\mathfrak B}^J$. Therefore, \eqref{eq:invariant_basis} is valid for all the generators of $\UQ$ hence any $\xi\in\UQ$.  
\end{proof}
 \begin{theorem}
The representation of $c_\ell^I$ in the subspace $W^J(K_1,K_2)\otimes W^{\Bar{J}}(K_3,K_4)$ of $V^{K_1}\otimes V^{K_2}\otimes V^{K_3}\otimes V^{K_4}$ is expressed as
\be
D^{K_1,K_2,K_3,K_4}(c^I_\ell)|_{W^J(K_1,K_2)\otimes W^{\Bar{J}}(K_3,K_4)}= \frac{[(2I+1)(2J+1)]_q}{d_J} \Id_{W^J(K_1,K_2)}\otimes \Id_{W^{\Bar{J}}(K_3,K_4)}\,.
\ee
Here the range of $J$ is $\max \lb |K_1-K_2|,|K_3,K_4| \rb\leq J\leq \min\lb u(K_1,K_2),u(K_3,K_4) \rb$. 
\label{Thm:eigenvalue_c_l}
\end{theorem}
\begin{proof}
 We first derive the representation $D^{K_1, K_2, K_3, K_4}$ of $c^I_\ell$ in the larger space $V^{K_1}\otimes V^{K_2}\otimes V^{K_3}\otimes V^{K_4}$. 
 We make use of the representation of quantum monodromies around single punctures \cite{Alekseev:1995rn} (see also \cite{Han:2023wiu}):
 \be
D^{K_1,K_2,K_3,K_4}(\bM_{\nu}^{I})=\kappa_{I}^{-1} \lb R'_{12}R'_{13}\cdots R'_{1\nu} R'_{1,\nu+1}R_{1,\nu+1}R^{'\,-1}_{1\nu}\cdots R^{'\,-1}_{13}R^{'\,-1}_{12}\rb^{IK_1\cdots K_\nu}\otimes e^{K_{\nu+1}}\otimes\cdots\otimes e^{K_{4}}\,.
\label{eq:rep_M_explicit}
\ee
Then a direct calculation gives
 \be
 \begin{split}
  D^{K_1, K_2, K_3, K_4}(c^I_\ell )&=D^{K_1, K_2, K_3, K_4}\lb \kappa_I \tr_q^I\lb \kappa_I M_{\ell_2}^IM^I_{\ell_1}\rb\rb
  = \kappa_I \tr_q^I\lb \kappa_I D^{K_1, K_2, K_3, K_4}\lb M_{\ell_2}^IM^I_{\ell_1}\rb\rb\\
  &= \kappa_I^2 \tr_q^I\lb  D^{K_1, K_2, K_3, K_4}\lb M_{\ell_2}^I\rb D^{K_1, K_2, K_3, K_4}\lb M^I_{\ell_1}\rb\rb\\
  &=\tr_q^I\left[\lb\lb R_{12}'R'_{13}R_{13}R_{12}'^{-1}\rb^{IK_1K_2}\otimes e^{K_3}\otimes e^{K_4}\rb\lb \lb R_{12}'R_{12}\rb^{IK_1K_2}\otimes e^{K_3}\otimes e^{K_4}\rb\right]\\
  &=\tr_q^I\left[\lb R_{12}'R'_{13}R_{13}R_{12}'^{-1}R_{12}'R_{12}\rb^{IK_1K_2}\otimes e^{K_3}\otimes e^{K_4}\right]\\
 &=\tr_q^I\left[\rho^I\otimes\rho^{K_1}\otimes\rho^{K_2}\lb\Id\otimes \Delta \rb\lb R'R \rb \right]\otimes e^{K_3}\otimes e^{K_4}\\
 &\equiv \tr_q^I\left[\rho^I\otimes\lb \rho^{K_1}\boxtimes\rho^{K_2}\rb\lb R'R \rb \right]\otimes \lb\rho^{K_3}\boxtimes \rho^{K_3}\rb(e)\,.
 \end{split}
 \ee
{This expression is written in a non-truncated, or unrealistic, version. As we are now working on $\UQ$ with $q$ a root-of-unity, we need to use the substitution rule $R\rightarrow \cR$ and $R'\rightarrow \cR'$ (see \eqref{eq:substitution_rule}) from the beginning of the derivation and the result is given \cite{Alekseev:1995rn}:
\be
 D^{K_1, K_2, K_3, K_4}(c^I_\ell )=\tr_q^I\left[\rho^I\otimes\lb \rho^{K_1}\boxtimes\rho^{K_2}\rb\lb R'R \rb \right]\otimes \lb\rho^{K_3}\boxtimes \rho^{K_3}\rb(e)\,,
\ee
where $R',R$ satisfy quasi-Yang Baxter equation instead of Yang-Baxter equation as for the non-truncated case. 
The calculation with substitution rule \eqref{eq:substitution_rule} gives the same result
\be
 \begin{split}
  &\phantom{ = }  D^{K_1, K_2, K_3, K_4}(c^I_\ell )\\
  &=D^{K_1, K_2, K_3, K_4}\lb \kappa_I \tr_q^I\lb \kappa_I M_{\ell_2}^IM^I_{\ell_1}\rb\rb
  = \kappa_I \tr_q^I\lb \kappa_I D^{K_1, K_2, K_3, K_4}\lb M_{\ell_2}^IM^I_{\ell_1}\rb\rb\\
  &= \kappa_I^2 \tr_q^I\lb  D^{K_1, K_2, K_3, K_4}\lb M_{\ell_2}^I\rb D^{K_1, K_2, K_3, K_4}\lb M^I_{\ell_1}\rb\rb\\
  &=\tr_q^I\left[\lb\lb \varphi_{123}R_{12}'\varphi^{-1}_{213}R'_{13}R_{13}\varphi_{213}R_{12}'^{-1}\varphi^{-1}\rb^{IK_1K_2}\otimes e^{K_3}\otimes e^{K_4}\rb\lb \lb \varphi_{123}R_{12}'\varphi^{-1}_{213}\varphi_{213}R_{12}\varphi^{-1}\rb^{IK_1K_2}\otimes e^{K_3}\otimes e^{K_4}\rb\right]\\
  &=\tr_q^I\left[\lb \varphi_{123}R_{12}'\varphi^{-1}_{213}R'_{13}R_{13}\varphi_{213}R_{12}'^{-1}\varphi^{-1}\varphi_{123}R_{12}'\varphi^{-1}_{213}\varphi_{213}R_{12}\varphi^{-1}\rb^{IK_1K_2}\otimes e^{K_3}\otimes e^{K_4}\right]\\
  &=\tr_q^I\left[\lb \varphi_{123}R_{12}'\varphi^{-1}_{213}R'_{13}R_{13}\varphi_{213}R_{12}\varphi^{-1}\rb^{IK_1K_2}\otimes e^{K_3}\otimes e^{K_4}\right]\\
 &=\tr_q^I\left[\rho^I\otimes\rho^{K_1}\otimes\rho^{K_2}\lb\Id\otimes \Delta \rb\lb R'R \rb \right]\otimes e^{K_3}\otimes e^{K_4}\\
 &\equiv \tr_q^I\left[\rho^I\otimes\lb \rho^{K_1}\boxtimes\rho^{K_2}\rb\lb R'R \rb \right]\otimes e^{K_3}\otimes e^{K_4}\,.
 \end{split}
 \label{eq:rep_c_l}
 \ee
 The sixth equality is obtained by using the quasi-inverse property, $\varphi\varphi^{-1}\varphi=\varphi$\footnotemark{}
 \footnotetext{{For the truncated algebra $\TUQ$, the underlying structure is a weak Hopf algebra, in which the coassociator $\varphi$ is no longer invertible, i.e., $\varphi\varphi^{-1} \neq e \otimes e\otimes e$, but is quasi-invertible.}}.}
 The detailed calculation of the fifth and sixth lines can be found in  Lemma \ref{lemma:delta_R'R} and Lemma \ref{lemma:varphi_123}.
 
The representation is now for the truncated algebra $\TUQ$ and the semi-simplicity is admitted:
\be
\rho^{K_1}\boxtimes \rho^{K_2} = \sum_{J=|K_1-K_2|}^{u(K_1,K_2)} \rho^J\,,\quad
\rho^{K_3}\boxtimes \rho^{K_4} = \sum_{J=|K_3-K_4|}^{u(K_3,K_3)} \rho^J\cong \sum_{J=|K_3-K_4|}^{u(K_3,K_3)} \rho^{\bar{J}}\,.
\ee
We also use the following results proven in \cite{Han:2023wiu} (
Lemma E.5 therein):
\begin{align}
\rho^I\otimes \rho^J (R'R) &= \sum^{u(I,J)}_{K=|I-J|} \frac{v_Iv_J}{v_K} \Id_K \,,\quad \\ 
 \sum_{K=|I-J|}^{u(I,J)} \frac{v_Iv_J}{v_K}\tr_q^K(\Id_K)&=[(2I+1)(2J+1)]_q\,.
\end{align}
Combining these facts, it is easy to see that
\be\begin{split}
D^{K_1,K_2,K_3,K_4}(c^I_\ell)|_{W^J(K_1,K_2)\otimes W^{\Bar{J}}(K_3,K_4)}
&=\tr_q^I \left[\lb \rho^I\otimes \rho^J\rb \lb R'R\rb\right] \otimes \Id_{W^{\bar{J}}(K_3,K_4)}\\
&=\frac{[(2I+1)(2J+1)]_q}{d_J} \Id_{W^J(K_1,K_2)}\otimes \Id_{W^{\Bar{J}}(K_3,K_4)}\,.
\end{split}\ee
 \end{proof}
 The Hilbert space $\mathcal{H}$ of moduli algebra on the four-punctured sphere is the multiplicity space $\mathcal{H}\equiv W^0(K_1,K_2,K_3,K_4)$ \cite{Alekseev:1995rn}. When we consider the cycle $\ell=\ell_2\ell_1 $ and the fusion algebra constructed from $\ell$, according to Lemma \ref{lemma:first_pinching_theorem}, the Hilbert space has the following decomposition $\mathcal{H}=\bigoplus_J W^J(K_1,K_2)\otimes W^{\Bar{J}}(K_3,K_4)$, where $J$ runs through all the representations satisfying the triangle inequality, $\max\lb|K_1-K_2|,|K_3-K_4|\rb\leq J \leq \min\lb u(K_1,K_2),u(K_3,K_4)\rb$ with spacing 1. The eigenvalue of element $c^I_\ell$ is calculated in the theorem \ref{Thm:eigenvalue_c_l}.

{If the label $I$ is taken to be the fundamental representation and $q$ being a root of unity, $q=e^{\frac{2\pi i}{k+2}}$, where $k$ is an integer, the case is further simplified. The eigenvalue becomes
\be
\frac{s_{\f12 J}}{d_J}=\frac{[2(2J+1)]_q}{[2J+1]_q}=\frac{\sin{\frac{2\pi}{k+2}(2J+1)}}{\sin{\frac{\pi}{k+2}(2J+1)}}=2\cos{(\frac{\pi}{k+2}(2J+1))}=(e^{\frac{\pi i}{k+2}(2J+1)}+e^{-\frac{\pi i}{k+2}(2J+1)})\;.
\ee
}

\section{Quantization of the length and twist coordinates}
\label{sec:quantize_coordinates}

 We define operators $\Tilde{\bold{x}}$ and $\Tilde{\bold{y}}$ such that their actions on the basis of the Hilbert space $\cH$ are as follows. 
\be
\Tilde{\bold{x}}(\sum_m e^{J}_{m}\otimes e^{\Bar{J}}_m)=e^{\frac{\pi i}{k+2}(2J+1)}(\sum_m e^{J}_{m}\otimes e^{\Bar{J}}_m)\quad,\quad \Tilde{\bold{y}}(\sum_m e^{J}_{m}\otimes e^{\Bar{J}}_m)=\sum_m e^{J+1}_{m}\otimes e^{\Bar{J}+1}_m\;.
\label{eq:true_x_y_action}
\ee
 From the definitions above, we observe that
 the operators $\Tilde{\bold{x}}$ and $\Tilde{\bold{y}}$ satisfy the following commutation relation:
 \be
 \Tilde{\bold{x}}\Tilde{\bold{y}}=q\Tilde{\bold{y}} \Tilde{\bold{x}}\;.
 \label{eq:Commutation_real_x_y}
 \ee
 Eqn.\eqref{eq:Commutation_real_x_y} is equivalent to the following expression:
 \be
 [\ln{\Tilde{\bold{x}}},\ln{\Tilde{\bold{y}}}]=i\frac{2\pi}{k+2}\;.
 \label{eq:commutation_real_lnx_lny}
 \ee 
 
 We now claim that the commutation relation \eqref{eq:Commutation_real_x_y} can be considered as a natural quantization of the length-twist coordinate.  
 To this end, we define an auxiliary Hilbert space $\cH_{\text{aux}}=\bigoplus_J W^J(K_1,K_2)\otimes W^{\Bar{J}}(K_3,K_4)$ where $J\in\N/2$ satisfies the triangular inequality as in $\cH$ but the spacing of neighbouring $J$'s is $1/2$ instead of 1. This means, the dimension of $\cH_{\text{aux}}$ is larger than that of $\cH$. 
 The use of $\mathcal{H}_{\text{aux}}$ will come in handy in the next section for the construction of the coherent state in $\mathcal{H}$. 
 We also define
 operators $\bold{x}$ and $\bold{y}$ that acts the basis of $\mathcal{H}_{\text{aux}}$ by
 \be
\bold{x}(\sum_m e^{J}_{m}\otimes e^{\Bar{J}}_m)=e^{\frac{\pi i}{k+2}(2J+1)}(\sum_m e^{J}_{m}\otimes e^{\Bar{J}}_m)\quad,\quad \bold{y}(\sum_m e^{J}_{m}\otimes e^{\Bar{J}}_m)=\sum_m e^{J+\frac{1}{2}}_{m}\otimes e^{\Bar{J}+\frac{1}{2}}_m\;.
\label{eq:x_y_action}
\ee
Operators $\bold{x}$ and $\bold{y}$ then satisfy the following commutation relation:
\be
\bold{x}\bold{y}=q^{\frac{1}{2}}\bold{yx}\;.
\label{eq:commutatioin_x_y}
\ee
{Comparison of \eqref{eq:Commutation_real_x_y} with \eqref{eq:x_y_action} intuitively motivates us to set $\Tilde{\bold{x}}=\bold{x}$ and $\Tilde{\bold{y}}=\bold{y}^2$. However, the actual relation, which will be clear in the next section, turns out to be $\Tilde{\bold{x}}=P\bold{x}P$ and $\Tilde{\bold{y}}=P\bold{y}^2P$, where $P$ is the projector that maps a state in $\mathcal{H}_{\text{aux}}$ into one in $\mathcal{H}$. The definition of the projector $P$ will be given in Proposition \ref{prop:projector_P} in the next section.}
  
{An immediate observation from \eqref{eq:commutatioin_x_y} is that it is equivalent to the expression:
\be
[\ln{\bold{x}},\ln{\bold{y}}]=i\hbar\;,
\label{eq:commutation_lnx_lny}
\ee
where $\hbar$ is identified with $\frac{\pi}{k+2}$ when $q=e^{i2\hbar}$ is regarded as a quantum deformation parameter.} Therefore, eqn.\eqref{eq:commutatioin_x_y} is natural quantization of the Poisson bracket of length-twist coordinates \eqref{eq:Poisson_length_twist}.

 By the comparison of 
 \eqref{eq:commutatioin_x_y} with
 \eqref{eq:Commutation_real_x_y}, we can identify $\ln{\bold{\Tilde{x}}}=\ln{\bold{x}}$ and $\ln{\bold{\Tilde{y}}}=\ln{\bold{y}^2}$. These identifications provide all the ingredients we need to show that \eqref{eq:Commutation_real_x_y} can be considered as the quantization of the length-twist coordinates. Define classical coordinates $\tilde{x}:=x$ and $\tilde{y}:=y^2$, then the semi-classical limit of \eqref{eq:commutation_real_lnx_lny} 
 recovers the Poisson bracket \eqref{eq:Poisson_length_twist} as
 \be
 \{\ln{\Tilde{x}},\ln{\Tilde{y}}\}=\{\ln{x},\ln{y^2}\}=2\;.
 \ee
 
Let us stress that the Hilbert space $\mathcal{H}$ only allows the existence of operators $\bold{\Tilde{x}}, \bold{\Tilde{y}}$. The reason for the introduction of $\mathcal{H}_{\text{aux}}$ is for our convenience in constructing the coherent state in $\mathcal{H}$, which we will illustrate in the next section.

{The strategy of coherent state's construction is that we first use operators $\bold{x}$ and $\bold{y}$ to construct the coherent state in $\mathcal{H}_{\text{aux}}$ then define a set of projectors that map the coherent state into $\mathcal{H}$.}

The $*$-operation of $\bold{x}$ and $\bold{y}$ can be induced from the $*$-operation of $c^I_\ell $. The action of $\bold{x^*}$ and $\bold{y^*}$ on the basis are given as:
\be
\bold{x^*}(\sum_m e^{J}_{m}\otimes e^{\Bar{J}}_m)=e^{-\frac{\pi i}{k+2}(2J+1)}(\sum_m e^{J}_{m}\otimes e^{\Bar{J}}_m)\quad,\quad \bold{y^*}(\sum_m e^{J}_{m}\otimes e^{\Bar{J}}_m)=(\sum_m e^{J-\frac{1}{2}}_{m}\otimes e^{\Bar{J}-\frac{1}{2}}_m)\;.
\ee

The commutation relation between $\bold{x^*}$ and $\bold{y^*}$ is the same as that between $\bold{x}$ and $\bold{y}$:
\be
\bold{x^*}\bold{y^*}=q\bold{y^*}\bold{x^*}.
\label{eq:commutatioin_xstar_ystar}
\ee
To end this section, let us point out an observation: the eigenvalue of $\bold{x}$ has a periodicity condition. This implies that the quantum states in the auxiliary Hilbert space $\mathcal{H}_{\text{aux}}$ are those defined on the torus. This motivates our construction of coherent states in the next section.

\section{Coherent states in the auxiliary Hilbert space} 
\label{sec:coherent_state}

We list the action of operators $\bold{x}, \bold{y}$ and $\bold{x^*}, \bold{y^*}$ on states in the sub-space of each representation label $J$:
\be
\begin{split}
\bold{x}(\Psi_J)=e^{\frac{\pi i}{k+2}(2J+1)}\Psi_J\quad&,\quad
\bold{y}(\Psi_J)=\Psi_{J+\frac{1}{2}}\;,\\
\bold{x^*}(\Psi_J)=e^{-\frac{\pi i}{k+2}(2J+1)}\Psi_J\quad&,\quad
\bold{y^*}(\Psi_J)=\Psi_{J-\frac{1}{2}}\;,
\end{split}
\ee
  where $\Psi_{J}$ is any state in the sub-space $W^{J}(K_1,K_2)\otimes W^{\Bar{J}}(K_3,K_4)$ for each representation label $J$. It is easy to check that \eqref{eq:commutatioin_x_y} and \eqref{eq:commutatioin_xstar_ystar} are indeed realized. 

{ As mentioned in the end of the preceding section, the auxiliary space is the quantization of the torus as the phase space. Classically, a torus is viewed as a two-dimensional phase space with canonical coordinates $(q,p)\in[0,a)\times [0,b)$. Since they are on the torus, it is naturally equipped them with periodic conditions. In quantum theory, the periodic condition yields the quantization condition for any states on the torus: }

\be
ab=2\pi\hbar N\;,
\label{eq:quantization_condition_torus}
\ee
where $N\in\mathbb{N}$.
For more details on quantum states on the torus, we refer to \cite{Gazeau:2009zz} as well as Appendix \ref{app:Quantum_State_on_the_torus}. 

{For the auxiliary Hilbert space, the irreducible representations $J\in\N/2$ fall within the range of $0$ and $(k+\frac{3}{2})$, as this range corresponds to the periodicity of the eigenvalue of the operator $\bold{x}$ as shown in \eqref{eq:x_y_action}.}

{Following \cite{Gazeau:2009zz}, we define a set of coherent states, denoted as $\psi_{(\Xt_0,\Yt_0)}$, as an $L^2(\R)$ function on a torus ${\mathbb T}_2=S^1\times S^1$ with parameters $(\Xt_0,\Yt_0)\in[0,2\pi)\times [0,2\pi)$ being the angle coordinates on ${\mathbb T}_2$. (The derivation of the coherent states can be found in Appendix \ref{app:Quantum_State_on_the_torus}.) When viewing ${\mathbb T}_2$ as a phase space, $(\Xt_0,\Yt_0)$ is a pair of canonical coordinates. $\psi_{(\Xt_0,\Yt_0)}(x)$ is defined as the superposition of basis $e_l(x)$ with coefficient $\xi_{(\Xt_0,\Yt_0)}(l)$:

\be
\psi_{(\Xt_0,\Yt_0)}(x):=\sum_{l=0}^{2k+3}
\xi_{(\Xt_0,\Yt_0)}(l) e_l(x)\,,
\label{eq:coherent_def}
\ee
where
\be
\begin{aligned}
\xi_{(\Xt_0,\Yt_0)}(l)&:=\left(\frac{1}{k+2}\right)^{\frac{1}{4}}e^{-\frac{i\Yt_0\Xt_0(k+2)}{2\pi}}\sum_{m'\in\mathbb{Z}}e^{-\frac{i(k+2)}{\pi}\Yt_0\left(l\frac{\pi}{k+2}-2\pi m'-\Xt_0\right)}e^{-\frac{k+2}{2\pi}(l\frac{\pi}{k+2}-2\pi m'-\Xt_0)^{2}}\,,\\
e_l(x)&:=\sum_{j\in\mathbb{Z}}\sqrt{\frac{\pi}{k+2}}\delta(x-2\pi\lb\f{l}{2(k+2)}-j\rb)\;.
\end{aligned}
\ee
{Here $\delta(x-2\pi\lb\f{l}{2(k+2)}-j\rb)$ is a Dirac distribution and the basis $e_{l}(x)$ is given by the Dirac combs and is an orthonormal basis in $\mathbb{C}$.}
Therefore, $\psi_{(\Xt_0,\Yt_0)}(x)$ takes non-trivial value only at discrete $x$'s. 
The resolution of identity for coherent states is \cite{Gazeau:2009zz}
\be
\frac{k+2}{2\pi^2}\int_{\mathbb{T}_2} \rd \Xt_0\rd \Yt_0 | \psi_{(\Xt_0,\Yt_0)}\rangle\langle  \psi_{(\Xt_0,\Yt_0)}| = \Id_{\mathcal{H}_{\text{aux}}}\;,
\label{eq:resolution_Haux}
\ee
where $| \psi_{(\Xt_0,\Yt_0)}\rangle$ is defined as in \eqref{eq:coherent_def} and $\langle  \psi_{(\Xt_0,\Yt_0)}|=\sum_{l=0}^{2k+3}
\bar{\xi}_{(\Xt_0,\Yt_0)}(l) e_l(x)$. 
By definition, the inner product of the coherent state is defined as 
\be
\langle\psi_{(\Xt_0,\Yt_0)}|\psi'_{(\Xt_0',\Yt_0')}\rangle=\sum^{2k+3}_{l,\tilde{l}=0} \bar{\xi}_{(\Xt_0,\Yt_0)}(l)\xi_{(\Xt_0',\Yt_0')}(\tilde{l})\,.
\label{eq:coherent_inner_product}
\ee

The final expression of the coherent state $\psi_{(\Xt_0,\Yt_0)}$ in the $\mathcal{H}_{\text{aux}} \simeq \mathbb{C}^{2(k+2)}$ is: 
\be
\sum^{2k+3}_{l=0} \xi_{(\Xt_0,\Yt_0)}(l)|l\rangle\;,
\ee
{where we have used $|l\rangle\equiv e_{\ell}(x)$ to denote the orthonormal basis in $\mathbb{C}^{2(k+2)}$.} This notation will be used in the next section. From our construction of coherent states on the torus, the $2(k+2)-$dimensional auxiliary Hilbert space $\mathcal{H}_{\text{aux}}$ is decomposed as follows: $\mathcal{H}_{\text{aux}}=\bigoplus^{2k+3}_{l=0} V^l$. 
\section{Projection into $\mathcal{H}$ and the expectation value}
\label{sec:projection_and_exp}

In the preceeding section, we have constructed the coherent state in the auxiliary Hilbert space $\mathcal{H}_{\text{aux}}$, which can be decomposed into subspace $V^l$, and where the eigenvector of ${\bf x}$ is associated to the eigenvalue $e^{\frac{i \pi}{k+2}(l+1)}$. However, the coherent state that corresponds to our quantum system are in the $\mathcal{H}$ rather than $\mathcal{H}_{\text{aux}}$. The physical Hilbert space $\mathcal{H}$ is decomposed into $\sum_{J} W^J\otimes W^{\Bar{J}}$, where $J$ can only fall within the range of $\max\lb|K_1-K_2|,|K_3-K_4|\rb\leq J \leq \min\lb u(K_1,K_2),u(K_3,K_4)\rb$ and the difference between closest $J$'s is $1$ instead of $\frac{1}{2}$.

Let us recall that the action of operators $\bold{x}$ and $\bold{y}$ on the sub-space $V^l$, where $l\in\N$ falls within the range of $0$ and $2k+3$, is defined as:
\be
\bold{x}(\Psi_l)=e^{\frac{\pi i}{k+2}(l+1)}\Psi_l\quad,\quad
\bold{y}(\Psi_l)=\Psi_{l+1}\,,\quad
\forall\,\Psi_l\in V^l\,.
\ee
The actions of $\bold{x}$ and $\bold{y}$ indeed admit the commutation relation \eqref{eq:commutatioin_x_y}. 
Denote $u(K_{\nu})\equiv\min\lb u(K_1,K_2),u(K_3,K_4)\rb$ and $m(K_{\nu})\equiv\max\lb |K_1-K_2|,|K_3-K_4|\rb$ . 
{The physical Hilbert space $\cH$ is the subspace of $\cH_{\rm aux}$ as each $J$ subspace of $\cH$ is identical to the subspace $V^{2J}$ of $\cH_{\rm aux}$ due to the fact that $\bigoplus^{u(K_{\nu})}_{J=m(K_{\nu})} V^J$ is isomorphic to $\bigoplus^{2u(K_{\nu})}_{2J=2m(K_{\nu})} V^{2J}$ and both have the same eigenvalue as $e^{\frac{\pi i}{k+2}(2J+1)}$.} We define a projection map $P_n:\cH_{\rm aux}\rightarrow \cH$ 
\be
P_n=\sum^{2u(K_{\nu})}_{2J=2m(K_{\nu})} |2J\rangle\langle2J|\,,
\ee
where $n$ denotes the dimension of $\mathcal{H}$, i.e. $n=2u(K_{\nu})-2m(K_{\nu})+1$ {and $|2J\rangle$ is an orthonormal basis in $\cH_{\rm aux}$.} 

\begin{prop}
  $P_n$ are the projectors that project the space $\mathcal{H}_{\text{aux}}$ onto $\mathcal{H}$, they satisfy the following relations:  
 \be
 P_n P_{n'}=P_{\min(n,n')}\quad,\quad (P_n)^*=P_n\,.
 \ee
\end{prop}

\begin{proof}
  \be
  \begin{aligned}
  P_nP_{n'}&=\sum_{2J}\sum_{2J'} |2J\rangle\langle 2J|2J'\rangle\langle 2J'|=\delta_{J,J'}\sum_{2J} |2J\rangle\langle 2J|= P_{\min(n,n')}\,,\\
  (P_n)^{*}&=(\sum_{2J} |2J\rangle\langle 2J|)^*=\sum_{2J} |2J\rangle\langle 2J|=P_n\,.
\end{aligned}
\ee 
\end{proof}
Given a coherent state $|\psi_{(\Xt_0,\Yt_0)}\rangle=\sum_{l=0}^{2k+3} \xi_{(\Xt_0,\Yt_0)}(l)|l\rangle\in\mathcal{H}_{\text{aux}}$, the projected coherent state $|\tpsi_{(\Xt_0,\Yt_0)}\rangle\equiv P_n |\psi_{(\Xt_0,\Yt_0)}\rangle$ is defined as 
\be
|\tpsi_{(\Xt_0,\Yt_0)}\rangle
= \sum_{2J=2m(K_\nu)}^{2u(K_\nu)} |2J\rangle\langle 2J| \psi_{(\Xt_0,\Yt_0)}\rangle
=\sum_{2J=2m(K_\nu)}^{2u(K_\nu)} \xi_{(\Xt_0,\Yt_0)}(2J)\,|2J\rangle\,. 
\ee

\begin{prop}
  Given elements $\bold{x}$ and $\bold{y}$ that form a Weyl algebra that quantizes a torus satisfying $ \bold{xy}=q^{\frac{1}{2}}\bold{yx}$, the projected elements $\bold{\Tilde{x}}:=P_n\bold{x}P_n$ and $\bold{\Tilde{y}}:=P_n\bold{y}^2P_n$ satisfy the same commutation relation of a Weyl algebra with $q^{\frac{1}{2}}\to q$. That is,  
    \be
  \bold{\Tilde{x}\Tilde{y}}=q\bold{\Tilde{y}\Tilde{x}}\,.
  \label{eq:XtYt_comm}
  \ee
  \label{prop:projector_P}
\end{prop}

\begin{proof}
We prove it by direct calculation for both sides of \eqref{eq:XtYt_comm}. 
For the {\it l.h.s.}, we have
\be
\begin{aligned}
\xbt\ybt&:=	P_n\xb P_nP_n\yb^2P_n=P_n\xb P_n\yb^2P_n
=\sum_{J,J',J''}|2J\ra\la 2J| \xb|2J'\ra\la 2J'|\yb^2
|2J''\ra\la2J''|\\
&=e^{\f{\pi i}{k+2}\lb 2J'+1\rb}\sum_{J,J',J''}|2J\ra\la2J|2J'\ra\la2J'|2J''+2\ra\la2J''|
=q e^{\f{\pi i}{k+2}\lb 2J'+1\rb}\sum_{J,J',J''}\delta_{J,J'}\delta_{J''+1,J'}|2J''+2\ra\la2J''|\\
&e^{\f{\pi i}{k+2}\lb 2J''+3\rb}\sum_{J''}|2J''+2\ra\la2J''|
\equiv q e^{\f{\pi i}{k+2}\lb 2J''+1\rb}\sum_{J''}|2J''+2\ra\la2J''|\,,
\end{aligned}
  \ee
  where the summations are from $m(K_\nu)$ to $u(K_\nu)$.
For the {\it r.h.s.},  
 \be
\begin{aligned}
q\ybt\xbt&:=q P_n\yb^2P_nP_n\xb P_n
=q P_n\yb^2P_n\xb P_n
=q\sum_{J,J',J''}|2J\ra\la2J| \yb^2 |2J'\ra\la 2J'|\xb |2J''\ra\la2J''|\\
&=qe^{\f{\pi i}{k+2}\lb 2J''+1\rb}\sum_{J,J',J''}|2J\ra\la2J|2J'+2\ra\la 2J'|2J''\ra\la2J''|
=qe^{\f{\pi i}{k+2}\lb 2J''+1\rb}\sum_{J,J',J''}\delta_{J,J'+1}\delta_{J',J''}|2J''+2\ra\la2J''|\\
&=qe^{\f{\pi i}{k+2}\lb 2J''+1\rb}\sum_{J''}|2J''+2\ra\la2J''|\equiv l.h.s.\,.
\end{aligned} 
 \ee
  \end{proof}
This shows that the operators $\xbt$ and $\ybt$ introduced in \ref{sec:quantize_coordinates} can be defined by projection of operators $\xb$ and $\yb$, which is the reason why we use the same notation to denote these operators here. 

The resolution of identity \eqref{eq:resolution_Haux} for coherent states in $\cH_{\rm aux}$ 
induces the resolution of identity for coherent states in $\cH$, which reads
\be
\frac{k+2}{2\pi^2}\int_{\mathbb{T}^2} \rd \Xt_0\rd \Yt_0 | \tpsi_{(\Xt_0,\Yt_0)}\rangle\langle  \tpsi_{(\Xt_0,\Yt_0)}| = \Id_{\mathcal{H}}\,.
\ee 
The coherent state labels $(\Xt_0,\Yt_0)$ encode the classical length-twist coordinates, which can be seen by calculating the expectation values of the operators $\xbt$ and $\ybt$ in the coherent state representation as shown in the following proposition.

\begin{prop}
In the semi-classical limit, where $k$ and the four representations $K_1, K_2, K_3, K_4$ increase at the same rate (i.e., $k=\lambda k$, $K_{\nu}=\lambda K_{\nu}$, and $\lambda\to\infty$), the expectation values of $\xbt$ and $\ybt$ in the projected coherent state representation behave as 
\be
\langle\Tilde{\bold{x}}\rangle=e^{i\Xt_0}+O\lb e^{-\lambda}/\sqrt{\lambda}\rb\,,\quad
\langle\Tilde{\bold{y}}\rangle=e^{2i\Yt_0}+ O \lb e^{-\lambda}/\sqrt{\lambda}\rb \;,
\ee
given that the position coordinate $\Xt_0$ labeling the coherent state $\tpsi_{(\Xt_0,\Yt_0)}$ satisfies the triangle inequality: $2 \max\left(\frac{|K_1-K_2|}{k+2},\frac{|K_3-K_4|}{k+2}\right) \leq \frac{\Xt_0}{\pi} \leq 2 \min\left(\frac{u(K_1,K_2)}{k+2},\frac{u(K_3,K_4)}{k+2}\right)$ with $u(K_1,K_2)= \min(K_1+K_2,k-K_1-K_2)$. Otherwise, the expectation values of $\Tilde{\bold{x}}$ and $\Tilde{\bold{y}}$ exponential decay as 
\be
\langle\Tilde{\bold{x}}\rangle= O \lb e^{-\lambda}/\sqrt{\lambda}\rb \;,\quad
\langle\Tilde{\bold{y}}\rangle= O \lb e^{-\lambda}/\sqrt{\lambda}\rb \;.
\ee
\end{prop}
\begin{proof}
We first calculate that 
\be\begin{split}
\la \xbt \ra:=\la\widetilde{\psi}_{(\Xt_0,\Yt_0)}|\xbt|\widetilde{\psi}_{(\Xt_0,\Yt_0)}\ra
&=\sum_{2J,2J'}\bar{\xi}_{(\Xt_0,\Yt_0)}(2J)\xi_{(\Xt_0,\Yt_0)}(2J')\langle2J|\tilde{{\bf x}}|2J'\rangle
=\sum_{2J}e^{\frac{\pi i}{k+2}(2J+1)}\bar{\xi}_{(\Xt_0,\Yt_0)}(2J)\xi_{(\Xt_0,\Yt_0)}(2J)\\
&=\sqrt{\frac{1}{k+2}}\sum_{s,s'\in\mathbb{Z}}e^{-i2(k+2)\Yt_0(s-s')}\sum_{2J}e^{\frac{\pi i}{k+2}(2J+1)}e^{-\frac{k+2}{2\pi}\left[\left(\pi\frac{2J}{k+2}-2\pi s-\Xt_0\right)^{2}+\left(\pi\frac{2J}{k+2}-2\pi s'-\Xt_0\right)^{2}\right]}\,,
\end{split}
\label{eq:exp_x_1}
\ee	
where the summation over $2J$ is from $2m(K_\nu)$ to $2u(K_\nu)$ as before. To proceed, we first apply the Poisson resummation formula
\be
\sum_{l=l_{i}}^{l_{f}}f(l)=\sum_{n\in\mathbb{Z}}\int_{l_{i}-\delta}^{l_{f}+1-\delta}{\rm d}l\,f(l)e^{2\pi inl}\,,\quad 
\delta>0 \text{ arbitrarily small}
\ee
to rewrite $\sum_{2J}$ in \eqref{eq:exp_x_1} in into an integral. 
Denote $m:={2m(K_\nu)}$ and $u:={2u(K_\nu)+1}$ for simplicity.  
 Then \eqref{eq:exp_x_1} becomes 
\be
\langle\tilde{{\bf x}}\rangle
=\f{1}{\sqrt{(k+2)}}\sum_{s,s'\in\mathbb{Z}}e^{-i2(k+2)\Yt_0(s-s')}\sum_{n\in\mathbb{Z}}\int_{{m}-\delta}^{u-\delta}{\rm d}(2J)\,e^{\frac{\pi i}{k+2}(2J+1)+4\pi iJn}e^{-\frac{k+2}{2\pi}\left[\left(\pi\f{2J}{k+2}-2\pi s-\Xt_0\right)^{2}+\left(\pi\f{2J}{k+2}-2\pi s'-\Xt_0\right)^{2}\right]}\,.
\label{eq:exp_x_2}
\ee

We observe that the integration can be written in terms of the error function whose asymptotic expansion behaves as
\be
{\rm erf}(x):=\frac{2}{\sqrt{\pi}}\int_{0}^{x}e^{-t^{2}}{\rm d}t=1-\frac{e^{-x^{2}}}{\sqrt{\pi}x}\left(1+O(1/x^{2})\right),\quad x\gg1\,.
\ee
Therefore, the asymptotic expansion of the integration $\int_{b}^{a}e^{-cx^{2}}{\rm d}x$ with $a> b$ and $c>0$ decays exponentially unless $a,b$ have different signs, or equivalently, the integrand $e^{-cx^{2}}$ peaks within the range $[a,b]$. More precisely,
\be
\int_{a}^{b}e^{-cx^{2}}{\rm d}x
= \begin{cases}
\sqrt{\frac{\pi}{4c}}\left({\rm erf}(b\sqrt{c})-{\rm erf}((a\sqrt{c})\right)\,\xrightarrow[\lambda\rightarrow\infty]{a\rightarrow\lambda a,\ b\rightarrow\lambda b}\ O\left(\frac{e^{-\lambda^{2}}}{\sqrt{\pi}\lambda}\right) & {\rm if\ }0\leq a<b\\
\sqrt{\frac{\pi}{4c}}\left({\rm erf}(b\sqrt{c})+{\rm erf}((a\sqrt{c})\right)\:\xrightarrow[\lambda\rightarrow\infty]{a\rightarrow\lambda a,\ b\rightarrow\lambda b}\ \sqrt{\frac{\pi}{c}}+O\left(\frac{e^{-\lambda^{2}}}{\sqrt{\pi}\lambda}\right) & {\rm if\ }a<0<b\\
-\sqrt{\frac{\pi}{4c}}\left({\rm erf}(b\sqrt{c})-{\rm erf}((a\sqrt{c})\right)\:\xrightarrow[\lambda\rightarrow\infty]{a\rightarrow\lambda a,\ b\rightarrow\lambda b}\ O\left(\frac{e^{-\lambda^{2}}}{\sqrt{\pi}\lambda}\right) & {\rm if\ }a<b\leq0
\end{cases}
\label{eq:approx_error}\,. 
\ee
The integrand of \eqref{eq:exp_x_2} peaks at $\f{2J}{k+2}=\frac{i(n(k+2)+\frac{1}{2})}{k+2}+(s+s')+\frac{x}{\pi}$, where it reads
\be
\exp\left[(k+2)\left(-\pi n^{2}-\pi(s-s')^{2}+2in(\pi(s+s')+\Xt_0)\right)-\frac{\pi\left(\frac{1}{4}-i\frac{1}{2}\right)}{k+2}+(-\pi n+i(\pi(s+s')+\Xt_0))\right]\,. 
\ee
At large $k$, the dominant contribution comes from the term $n=0,s'=s$ of the summation. In this case, the peak is $\f{2J}{k+2}=\f{\Xt_0}{\pi}+2s+i\f{1}{2(k+2)}\xrightarrow[\lambda\rightarrow\infty]{k\rightarrow\lambda k} \f{\Xt_0}{\pi}+2s$, which is within the range $[\tilde{m},\tilde{u}]$ only if $s=0$. Therefore, \eqref{eq:exp_x_2} can be simplified at the $\lambda\rightarrow\infty$ approximation (taking into account that the constant $c$ takes the value $\pi (k+2)$ when using \eqref{eq:approx_error})
\be
\langle\tilde{{\bf x}}\rangle \xrightarrow[\lambda\rightarrow\infty]{K_\nu\rightarrow\lambda K_\nu,k\rightarrow \lambda k }e^{i\Xt_0}+O\lb e^{-\lambda}/\sqrt{\lambda} \rb\,
\ee
where $2 \max\left(\frac{|K_1-K_2|}{k+2},\frac{|K_3-K_4|}{k+2}\right) \leq \frac{\Xt_0}{\pi} \leq 2 \min\left(\frac{u(K_1,K_2)}{k+2},\frac{u(K_3,K_4)}{k+2}\right)$. Otherwise, $\langle\tilde{{\bf x}}\rangle=O\lb e^{-\lambda}/\sqrt{\lambda} \rb$.

The expectation value of $\ybt$ is calculated in the same way:
\be\begin{split}
\langle\tilde{{\bf y}}\rangle&:=\langle\tilde{\psi}_{(\Xt_0,\Yt_0)}|\tilde{{\bf y}}|\tilde{\psi}_{(\Xt_0,\Yt_0)}\rangle=\sum_{2J,2J'}\bar{\xi}_{(\Xt_0,\Yt_0)}(2J')\xi_{(\Xt_0,\Yt_0)}(2J)\langle2J|\tilde{{\bf y}}|2J'\rangle=\sum_{2J}\bar{\xi}_{(\Xt_0,\Yt_0)}(2J+2)\xi_{(\Xt_0,\Yt_0)}(2J)\\
&=\left(\frac{1}{k+2}\right)^{\frac{1}{2}}\sum_{s,s'\in\mathbb{Z}}e^{-i2(k+2)\Yt_0(s-s')}e^{i2\Yt_0}\sum_{2J}e^{-\frac{k+2}{2\pi}\left[\left(\pi\frac{2J+2}{k+2}-2\pi s-\Xt_0\right)^{2}+\left(\pi\frac{2J}{k+2}-2\pi s'-\Xt_0\right)^{2}\right]}\\
&=\left(\frac{1}{k+2}\right)^{\frac{1}{2}}\sum_{s,s'\in\mathbb{Z}}e^{-i2(k+2)\Yt_0(s-s')+i2\Yt_0}\sum_{n\in\mathbb{Z}}\int_{2m-\delta}^{2u-\delta}{\rm d}(2J)\,e^{-\frac{k+2}{2\pi}\left[\left(\pi\frac{2J+2}{k+2}-2\pi s-\Xt_0\right)^{2}+\left(\pi\frac{2J}{k+2}-2\pi s'-\Xt_0\right)^{2}\right]}\,,
\end{split}\ee
where we have used the Poisson resummation in the third line. 
The peak of the integrand in the last line is at $\f{2J}{k+2}=(s+s')+\frac{\Xt_0}{\pi}-\frac{1}{k+2}\xrightarrow[\lambda\rightarrow \infty]{k\rightarrow\lambda k}(s+s')+\frac{\Xt_0}{\pi}$, which is within the range $[\f{m}{k+2},\f{u}{k+2}]$ only if $s+s'=0$ then the peak is at $\f{2J}{k+2}= \f{\Xt_0}{\pi}$. At the peak, the integrand takes the form
\be
\begin{split}
e^{-\frac{k+2}{2\pi}\left[\left(\pi\f{2J}{k+2}+\frac{2\pi}{k+2}-2\pi s-\Xt_0\right)^{2}+\left(\pi\f{2J}{k+2}-2\pi s'-\Xt_0\right)^{2}\right]}	
&\xrightarrow{\f{2J}{k+2}=\frac{\Xt_0}{\pi}}e^{-\frac{k+2}{2\pi}\left[\left(\Xt_0+\frac{2\pi}{k+2}-2\pi s-\Xt_0\right)^{2}+\left(\Xt_0-2\pi s'-\Xt_0\right)^{2}\right]}\\
&	\xrightarrow{k\rightarrow\infty}e^{-\frac{k+2}{2\pi}\left[\left(2\pi s\right)^{2}+\left(2\pi s'\right)^{2}\right]}\,,
\end{split}\ee
which decays exponentially at large $k$ unless $s=s'=0$. We therefore conclude that 
\be
\langle\tilde{{\bf y}}\rangle \xrightarrow[\lambda\rightarrow\infty]{K_\nu\rightarrow\lambda K_\nu,k\rightarrow \lambda k }e^{2i\Yt_0}+O\lb e^{-\lambda}/\sqrt{\lambda} \rb,
\ee
when $2 \max\left(\frac{|K_1-K_2|}{k+2},\frac{|K_3-K_4|}{k+2}\right) \leq \frac{\Xt_0}{\pi} \leq 2 \min\left(\frac{u(K_1,K_2)}{k+2},\frac{u(K_3,K_4)}{k+2}\right)$. Otherwise, $\langle\tilde{{\bf y}}\rangle=O\lb e^{-\lambda}/\sqrt{\lambda} \rb$.
\end{proof}

So far, we have only calculated the expectation values of $\Tilde{\bold{x}}$ and $\Tilde{\bold{y}}$. However, in general, one can find the expectation value of a polynomial in $\Tilde{\bold{x}}$ and $\Tilde{\bold{y}}$. The fact that the sum $\sum^{2u(K_{\nu})}_{2J=2m(K_{\nu})}$ has an upper limit restricts the power order $\mu$ that $\Tilde{\bold{y}}$ can have. Acting $\ybt^\mu$ on $|\tpsi_{(\Xt_0,\Yt_0)}\ra$, we have
\be
  \ybt^\mu |\tpsi_{(\Xt_0,\Yt_0)}\ra =P_n\yb^{2\mu}P_n|\psi_{(\Xt_0,\Yt_0)}\rangle
  =\sum_{2J,2J'}|2J'\rangle\langle2J'| \yb^{2\mu} |2J\rangle\langle2J|\psi_{(\Xt_0,\Yt_0)}\rangle
  =\sum_{2J,2J'}|2J'\rangle\langle2J'| 2J+2\mu\rangle\langle2J|\psi_{(\Xt_0,\Yt_0)}\rangle\;,
\ee
where sums are from $m(K_{\nu})\equiv\max\lb |K_1-K_2|,|K_3-K_4|\rb$ to $u(K_{\nu})\equiv\min\lb u(K_1,K_2),u(K_3,K_4)\rb$. 
From the calculation above, we obtain the range that $\mu$ must fall within: 
\be
0\leq \mu\leq u(K_{\nu})-m(K_{\nu})\;.
\ee
The state is set to $0$ if $\mu$ is greater than the allowed value.

We have, therefore, seen that the coherent state $|\tpsi_{(\Xt_0,\Yt_0)}\ra$ provides a complete basis spanning the intertwiner space of a quantum curved tetrahedron as a Hilbert space, and that the operators $\xbt$ and $\ybt$ quantize the length coordinate $x$ and twist coordinate $y$ respectively as multiplication and derivative operators on the Hilbert space. Let the coherent state labels $\Xt_0=\theta,\,\Yt_0=\phi$ defined in \eqref{eq:def_theta_phi}. {Then the expectation values of $\xbt$ and $\ybt$ correspond to their classical counterparts, $e^{i\theta}\equiv x$ and $e^{i2\phi}\equiv y^2$ respectively. Both $\Xt_0$ and $\Yt_0$ indeed fall within the range from $0$ to $\pi$, considering that the triangle inequality is satisfied. This inequality, $2\pi \max\left(\frac{|K_1-K_2|}{k+2},\frac{|K_3-K_4|}{k+2}\right) \leq \Xt_0 \leq 2\pi \min\left(\frac{u(K_1,K_2)}{k+2},\frac{u(K_3,K_4)}{k+2}\right)$, where $u(K_1,K_2)= \min(K_1+K_2,k-K_1-K_2)$, is considered as the quantum counterpart of \eqref{eq:range_theta_phi}.} 


\section{Conclusion and discussion}
In this work, we have constructed the algebra generated by the quantum monodromies $\bM^I_{\ell}$, where the loop $\ell$ encloses a pair of punctures. We have proved that the algebra generated by $\bM^I_{\ell}$ forms a loop algebra and that the $q$-deformed Wilson loop operators $c^I_{\ell}$ constructed from the $\bM^I_{\ell}$ form a fusion algebra. The quantum diagonal length operator is obtained from $c^{\frac{1}{2}}_\ell$. A set of coherent states is constructed directly in the intertwiner space and the expectation value of length and twist operators in the semi-classical limit peak at points of the phase space, each describing the shape of a constantly curved tetrahedron.
We have also shown that not all coherent states are geometrical states for technical reasons. Firstly, the coherent state labels $\Xt_0$ and $\Yt_0$ label a point on the torus, but the phase space of the shape of a tetrahedron is only the subspace of the torus phase space. Furthermore, the integration of the resolution of the identity is over the entire torus, which means that the coherent intertwiner is not in one-to-one correspondence with the classical tetrahedron. 

The coherent states constructed in this paper may be adapted to construct the spinfoam model with a non-zero cosmological constant $\Lambda$ similar to the one introduced in \cite{Han:2021tzw}, where the coherent state labels were chosen to be the Fock-Goncharov coordinates. These coordinates are also coordinates of the tetrahedron shape phase space, but are not Darboux coordinates and do not have natural holonomy interpretation as the length variables used in this paper, which makes it more difficult to connect to the canonical quantization approach of LQG (with $\Lambda$). We expect that the adjustment of coherent state coupling to the partition function in building the spinfoam model could lead to a more feasible model, which is easier for applications. 

The Guillemin-Sternberg theorem \cite{guillemin1982geometric} guarantees that the quantization commutes with reduction as illustrated in the following commuting diagram: 
\be\ba{ccc}
(\SU(2))^{\otimes 4}\quad & \xrightarrow{\quad \text{ quantization }\quad}\quad & 
\ba{c}V^{K_1}\otimes V^{K_2}\otimes V^{K_3}\otimes V^{K_4}\ea\\
\phantom{\text{quantization}}\left\downarrow\rule{0cm}{1cm}\right.\text{\small symplectic reduction} &&
\phantom{\text{quantization}}\left\downarrow\rule{0cm}{1cm}\right.\text{\small quantum reduction}\\\cM_{\text{flat}}(\Sigma_{0,4},\SU(2))\quad & \xrightarrow{\quad \text{ quantization }\quad}\quad & \ba{c}\text{Inv}_q(K_1,K_2,K_3,K_4)\ea\,.
\ea
\label{eq:commut_diag_quan_redu}
\ee
In our approach, we proceed the quantization first, and the reduction after. However, we may try the other route as in \cite{Conrady:2009px}, i.e. to proceed the geometric quantization directly on $\cM_{\text{flat}}(\Sigma_{0,4},\SU(2))$. The Hilbert space should be the conformal blocks of the WZW model, which is identified with the Hilbert space constructed by the combinatorial quantization \cite{Alekseev:1995rn}. One may define the coherent intertwiners in this manner. It would be interesting to investigate the isomorphism between those states and the coherent states we constructed in this paper.

\begin{acknowledgements}
The authors would like to acknowledge Muxin Han for various helpful discussions. This work receives support from the National Science Foundation through grants PHY-2207763, PHY-2110234, the Blaumann Foundation and the Jumpstart Postdoctoral Program at Florida Atlantic University.

\end{acknowledgements}

\appendix
\renewcommand\thesection{\Alph{section}}

 \section{Defining relations for $\bM^I_{\ell}$ with the substitution rules}
 \label{com_relation_truncated}
 
 For the truncated case, $\TUQ$ is a weak Hopf algebra, we apply the substitution rule over all defining relations of the lattice algebra $\mathcal{B}$, which is generated by the matrix elements of quantum holonomies of oriental edges in a lattice. The substitution rules are listed below \cite{Alekseev:1994au,Alekseev:1995rn}.
\be
\begin{split}
&{C}[IJ|K] \rightarrow  \widetilde{C}[IJ|K] := C[IJ|K](\varphi^{-1})^{IJ}\,, \qquad
{C}[IJ|K]^{*}\rightarrow  \widetilde{C}[IJ|K]^{*} := (\varphi_{213})^{IJ}C[IJ|K]^{*}\,, \\
&R^{IJ} \rightarrow \mathcal{R}^{IJ} := (\rho^{I}\otimes \rho^{J}\otimes \Id)(\varphi_{213}R\varphi^{-1})\,,\qquad
 d_I \rightarrow \tilde{d}_I := \tr^I(\rho^I(gS(\beta)\alpha))\,,\qquad
 R^I  \rightarrow  R^I \equiv(\rho^I\otimes \Id)R\,,\\
 &\tr_{q}^{I}(X)\rightarrow \widetilde{\tr}_q^I(X):=\tr^{I}(m^{I}Xw^{I}g^{I})\,,\quad
 \text{with }\; m^{I}=\rho^{I}(S(\phi^{(1)})\alpha \phi^{(2)})\phi^{(3)}\,,\quad  w^{I}=\rho^{I}(\varphi^{(2)}S^{-1}(\varphi^{(1)}\beta))\varphi^{(3)}\,,
\end{split}  
\label{eq:substitution_rule}
\ee
where $\varphi,\phi\equiv\varphi^{-1}\in\TUQ\otimes\TUQ\otimes\TUQ$ and $\alpha,\beta\in\TUQ$ are the defining elements for $\TUQ$. Note that $\phi\varphi\neq e\otimes e\otimes e$ but $\phi\varphi=(\Delta\otimes\Id)\Delta(e)$.

 By using the exchange relations of the quantum holonomies, we can derive all the commutation relations of the quantum monodromies $\bM^I_{\ell}$ from the commutation relations of the quantum holonomies. The quantum monodromies $\bM^I_{\ell}$ can be expressed as matrix product of quantum holonomies, e.g. $\bM^I_{\ell}=\kappa^{-1}_I U^I_{\mathcal{C}}U^I_{\mathcal{C'}}$.
 
 The functoriality condition of $\bM^I_{\ell}$ can be derived from the functoriality condition of quantum holonomies, i.e. $\stackrel{1}{U^I_{\mathcal{C}}}\stackrel{2}{U^J_{\mathcal{C}}}=\sum_K \widetilde{C}[IJ|K]^{*}U^K_{\mathcal{C}}\widetilde{C}[IJ|K]$ and the commutation relation $\stackrel{1}{U^I_{\mathcal{C'}}}\mathcal{R}_{x}^{IJ}\stackrel{2}{U^J_{\mathcal{C}}}=\stackrel{2}{U^J_{\mathcal{C}}}\mathcal{R'}_{y}^{IJ}\stackrel{1}{U^I_{\mathcal{C'}}}$:
 \be
 \begin{aligned}   
 \stackrel{1}{\bM^I_{\ell}}\mathcal{R}_{x}^{IJ}\stackrel{2}{\bM^J_{\ell}}&=\kappa^{-1}_I\kappa^{-1}_J\stackrel{1}{U^I_{\mathcal{C}}}\stackrel{1}{U^I_{\mathcal{C'}}}\mathcal{R}_{x}^{IJ}\stackrel{2}{U^J_{\mathcal{C}}}\stackrel{2}{U^J_{\mathcal{C'}}}\\
 &=\kappa^{-1}_I\kappa^{-1}_J\stackrel{1}{U^I_{\mathcal{C}}}\stackrel{2}{U^J_{\mathcal{C}}}\mathcal{R'}_{y}^{IJ}\stackrel{1}{U^I_{\mathcal{C'}}}\stackrel{2}{U^J_{\mathcal{C'}}}\\
 &=\kappa^{-1}_I\kappa^{-1}_J\sum_{K,K'} \widetilde{C}[IJ|K]^{*}U^K_{\mathcal{C}}\widetilde{C}[IJ|K]\mathcal{R'}_{x}^{IJ} \widetilde{C}[IJ|K']^{*}U^{K'}_{\mathcal{C'}}\widetilde{C}[IJ|K']\\
 &=\kappa^{-1}_K\sum_{K} \widetilde{C}[IJ|K]^{*}U^K_{\mathcal{C}}U^K_{\mathcal{C'}}\widetilde{C}[IJ|K]\\
 &=\sum_{K} \widetilde{C}[IJ|K]^{*}\bM^K_{\ell}\widetilde{C}[IJ|K]\;.
 \end{aligned}
 \ee
 The commutation relation of quantum monodromy $\bM^I_{\ell}$ can be derived from exchange relation of quantum holonomies:
 \be
 \begin{aligned}
 (\mathcal{R}^{-1})_{x}^{IJ}\stackrel{1}{\bM^I_{\ell}}\mathcal{R}_{x}^{IJ}\stackrel{2}{\bM^J_{\ell}}&=\kappa^{-1}_I\kappa^{-1}_J(\mathcal{R}^{-1})_{x}^{IJ}\stackrel{1}{U^I_{\mathcal{C}}}\stackrel{1}{U^I_{\mathcal{C'}}}\mathcal{R}_{x}^{IJ}\stackrel{2}{U^J_{\mathcal{C}}}\stackrel{2}{U^J_{\mathcal{C'}}}\\
 &=\kappa^{-1}_I\kappa^{-1}_J(\mathcal{R}^{-1})_{x}^{IJ}\stackrel{1}{U^I_{\mathcal{C}}}\stackrel{2}{U^J_{\mathcal{C}}}\mathcal{R'}_{y}^{IJ}\stackrel{1}{U^I_{\mathcal{C'}}}\stackrel{2}{U^J_{\mathcal{C'}}}\\
 &=\kappa^{-1}_I\kappa^{-1}_J(\mathcal{R}^{-1})_{x}^{IJ}\mathcal{R}_{x}^{IJ}\stackrel{2}{U^J_{\mathcal{C}}}\stackrel{1}{U^I_{\mathcal{C}}}(\mathcal{R'}^{-1})_{y}^{IJ}\mathcal{R'}_{y}^{IJ}\stackrel{1}{U^I_{\mathcal{C'}}}\stackrel{2}{U^J_{\mathcal{C'}}}\\
 &=\kappa^{-1}_I\kappa^{-1}_J\stackrel{2}{U^J_{\mathcal{C}}}\stackrel{1}{U^I_{\mathcal{C}}} (\mathcal{R})_{y}^{IJ}\stackrel{2}{U^J_{\mathcal{C'}}}\stackrel{1}{U^I_{\mathcal{C'}}}(\mathcal{R'}^{-1})_{x}^{IJ}\\
&=\kappa^{-1}_I\kappa^{-1}_J\stackrel{2}{U^J_{\mathcal{C}}}\stackrel{2}{U^J_{\mathcal{C'}}}(\mathcal{R'})_{x}^{IJ}\stackrel{1}{U^I_{\mathcal{C}}}\stackrel{1}{U^I_{\mathcal{C'}}}(\mathcal{R'}^{-1})_{x}^{IJ}\\
 &=\stackrel{2}{\bM^I_{\ell}}\mathcal{R'}_{x}^{IJ}\stackrel{1}{\bM^J_{\ell}}(\mathcal{R'}^{-1})_{x}^{IJ}\;.
 \end{aligned}
 \ee
 For any cycle $\ell_{\mu}$ such that $\ell \prec \ell_{\mu}$, the monodromy along the cycle $\ell_{\mu}$ can also be expressed in terms of quantum holonomies i.e. $\bM^I_{\ell_{\mu}}=\kappa^{-1}_{I}U^I_{\mathcal{\Tilde{C}}}U^I_{\mathcal{\Tilde{C'}}}$ by separating $\ell_\mu$ into $\tilde{C}\circ\tilde{C}'$. The commutation relation between $\bM^I_{\ell}$ and $\bM^I_{\ell_{\mu}}$ can be derived from the commutation relations of quantum holonomies:
 \be
 \begin{aligned}
 (\mathcal{R}^{-1})_{x}^{IJ}\stackrel{1}{\bM^I_{\ell}}\mathcal{R}_{x}^{IJ}\stackrel{2}{\bM^J_{\ell_{\mu}}}&=\kappa^{-1}_{I}\kappa^{-1}_{J}(\mathcal{R}^{-1})_{x}^{IJ}\stackrel{1}{U^I_{\mathcal{C}}}\stackrel{1}{U^I_{\mathcal{C'}}}\mathcal{R}_{x}^{IJ}\stackrel{2}{U^J_{\mathcal{\Tilde{C}}}}\stackrel{2}{U^J_{\mathcal{\Tilde{C'}}}}\\
 &=\kappa^{-1}_{I}\kappa^{-1}_{J}(\mathcal{R}^{-1})_{x}^{IJ}\stackrel{1}{U^I_{\mathcal{C}}}\stackrel{2}{U^J_{\mathcal{\Tilde{C}}}}\stackrel{1}{U^I_{\mathcal{C'}}}\stackrel{2}{U^J_{\mathcal{\Tilde{C'}}}}\\
 &=\kappa^{-1}_{I}\kappa^{-1}_{J}(\mathcal{R}^{-1})_{x}^{IJ}\mathcal{R}_{x}^{IJ}\stackrel{2}{U^J_{\mathcal{\Tilde{C}}}}\stackrel{1}{U^I_{\mathcal{C}}}\stackrel{2}{U^J_{\mathcal{\Tilde{C'}}}}\stackrel{1}{U^I_{\mathcal{C'}}}\mathcal{R}_{x}^{IJ}\\
 &=\kappa^{-1}_{I}\kappa^{-1}_{J}\stackrel{2}{U^J_{\mathcal{\Tilde{C}}}}\stackrel{2}{U^J_{\mathcal{\Tilde{C'}}}}(\mathcal{R}^{-1})_{x}^{IJ}\stackrel{1}{U^I_{\mathcal{C}}}\stackrel{1}{U^I_{\mathcal{C'}}}\mathcal{R}_{x}^{IJ}\\
 &=\stackrel{2}{\bM^J_{\ell_{\mu}}}(\mathcal{R}^{-1})_{x}^{IJ}\stackrel{1}{\bM^I_{\ell}}\mathcal{R}_{x}^{IJ}\;.
  \end{aligned}
 \ee
 On the other hand, for any cycle $\ell_{\mu}$ such that $\ell \succ \ell_{\mu}$, we have
 \be
 \begin{aligned}
 \mathcal{R'}_{x}^{IJ}\stackrel{1}{\bM^I_{\ell}}(\mathcal{R'}^{-1})_{x}^{IJ}\stackrel{2}{\bM^J_{\ell_{\mu}}}&=\kappa^{-1}_{I}\kappa^{-1}_{J}\mathcal{R'}_{x}^{IJ}\stackrel{1}{U^I_{\mathcal{C}}}\stackrel{1}{U^I_{\mathcal{C'}}}(\mathcal{R'}^{-1})_{x}^{IJ}\stackrel{2}{U^J_{\mathcal{\Tilde{C}}}}\stackrel{2}{U^J_{\mathcal{\Tilde{C'}}}} \\
 &=\kappa^{-1}_{I}\kappa^{-1}_{J}\mathcal{R'}_{x}^{IJ}\stackrel{1}{U^I_{\mathcal{C}}}\stackrel{2}{U^J_{\mathcal{\Tilde{C}}}}\stackrel{1}{U^I_{\mathcal{C'}}}\stackrel{2}{U^J_{\mathcal{\Tilde{C'}}}}\\
 &=\kappa^{-1}_{I}\kappa^{-1}_{J}\mathcal{R'}_{x}^{IJ}(\mathcal{R'}^{-1})_{x}^{IJ}\stackrel{2}{U^J_{\mathcal{\Tilde{C}}}}\stackrel{1}{U^I_{\mathcal{C}}}\stackrel{2}{U^J_{\mathcal{\Tilde{C'}}}}\stackrel{1}{U^I_{\mathcal{C'}}}(\mathcal{R'}^{-1})_{x}^{IJ}\\
 &=\kappa^{-1}_{I}\kappa^{-1}_{J}\stackrel{2}{U^J_{\mathcal{\Tilde{C}}}}\stackrel{2}{U^J_{\mathcal{\Tilde{C'}}}}\mathcal{R'}_{x}^{IJ}\stackrel{1}{U^I_{\mathcal{C}}}\stackrel{1}{U^I_{\mathcal{C'}}}(\mathcal{R'}^{-1})_{x}^{IJ}\\
 &=\stackrel{2}{\bM^J_{\ell_{\mu}}}\mathcal{R'}_{x}^{IJ}\stackrel{1}{\bM^I_{\ell}}(\mathcal{R'}^{-1})_{x}^{IJ}\;.
 \end{aligned}
 \ee
 The element $c^I_{\ell}$ is subject to change with the substitution rule \eqref{eq:substitution_rule} applied and they satisfy fusion rule \eqref{eq:c_fusion} \cite{Alekseev:1994au}. The element $c^I_{\ell}$ is expressed as 
 \be
 \kappa_I\widetilde{\tr}_q^I(\bM^I_{\ell})=\kappa_I\tr^{I}(m^{I}\bM^I_{\ell}w^{I}g^{I})\;,
 \ee
 where $m^I$ and $w^I$ can be found in \eqref{eq:substitution_rule}. The quantum character $\chi^I_{\ell}$ can be constructed by using S-matrix and elements $c^I_{\ell}$ \cite{Alekseev:1994au}:
 \be
 S_{IJ}=\mathcal{N}(\widetilde{\tr}_q^I\otimes \widetilde{\tr}_q^J)(R'R)\;,
 \ee
where $R',R$ satisfy the quasi-Yang-Baxter equation. The $S$-matrix satisfies the properties listed in \eqref{eq:S_matrix}.
 Then the quantum character $\chi^I_{\ell}$ is defined as:
 \be
 \chi^I_\ell=\mathcal{N}d_I S_{IJ}c^{\Bar{J}}=\mathcal{N}d_I S_{I\Bar{J}}c^{J}\,.
 \ee
 \section{Some detailed calculation}
 \label{app:detail_calc}
 In this appendix, we collect some detailed calculations used in the main text, which are formulated into lemmas. Lemmas \ref{lemma:commu_linear_order} and \ref{lemma:trace_identity} are used in the proof of the Propositions \ref{prop:loop_algebra} and \ref{prop:central_element} respectively. Lemmas \ref{lemma:varphi_123} and \ref{lemma:delta_R'R} are used in the proof of the Theorem \ref{Thm:eigenvalue_c_l} as the substitution rule \eqref{eq:substitution_rule} applied. 
 \begin{lemma}
Suppose $M_\ell^{I}$ is the monodromy around the $m$-th and the $(m-1)$-th punctures, where $m=2,3,4$. The quantum trace of $M_\ell^{I}$ has the same two identities as the quantum trace of monodromies around one puncture, i.e.
\begin{equation}
\tr_q^I(M_\ell^I)=\tr_q^I\left((R^{-1})^{IJ}M_\ell^IR^{IJ}\right)=\tr_q^I\left((R')^{IJ}M_\ell^I((R')^{-1})^{IJ}\right).
\end{equation}
\label{lemma:trace_identity}
\end{lemma}
\begin{proof}
 Let's prove the two quantum trace identities with the help of Sweedler's notation and quasi-triangularity: 
\begin{equation}
\begin{aligned}
\tr_q^I\left((R^{-1})^{IJ}M_\ell^IR^{IJ}\right)\quad
&=\tr^I((\sum_jS(R_{j}^{1})\otimes R_{j}^{2})^{IJ}M_\ell^I (\sum_iR_{i}^{1}\otimes R_{i}^{2})^{IJ}\rho^I(g))\\
  &=\tr^I(M_\ell^I (\sum_{ij}(R_{i}^{1})^I(S^{-1}(R_{j}^{1}))^{I}(R_{j}^{2})^J(R_{i}^{2})^{J}\rho^I(g))\\
  &=\tr^I(M_\ell^I (S^{-1}\otimes \Id)(\sum_{ij}(R_{j}^{1})^{I}(S(R_{i}^{1}))^I(R_{j}^{2})^J(R_{i}^{2})^{J}\rho^I(g))
  \\
  &=\tr^I(M_\ell^I (S^{-1}\otimes \Id)(\sum_{j}(R_{j}^{1})^{I}(R_{j}^{2})^J\sum_{i}(S(R_{i}^{1}))^I(R_{i}^{2})^{J}\rho^I(g)))\\
  &=\tr^I(M_\ell^I (S^{-1}\otimes \Id)(R)^{IJ}(R^{-1})^{IJ}\rho^I(g)))
  =\tr_q^I(M_\ell^I)\,,
\end{aligned}  
\end{equation}
\begin{equation}
\begin{aligned}
\tr_q^I((R')^{IJ}M_\ell^I((R')^{-1})^{IJ})&= \tr^I((\sum_i R_i^{1}\otimes R_i^{2})^{IJ}M_\ell^I(\sum_{j}S^{-1}(R_j^{1})\otimes R_{j}^{2})^{IJ}\rho^I(g))\\
  &=\tr^I((\sum_{ij} ((S(R_j^{1}))^{I}R_i^{1})^{I} (R_i^{2})^{J}(R_{j}^{2})^{J}M_\ell^I \rho^I(g)))\\
  &=\tr^I((S\otimes \Id)(\sum_{ij} ((S^{-1}(R_i^{1}))^{I}(R_j^{1})^{I} (R_i^{2})^{J}(R_{j}^{2})^{J}M_\ell^I \rho^I(g))))\\
  &=\tr^I((S\otimes \Id)(\sum_{i} ((S^{-1}(R_i^{1}))^{I} (R_i^{2})^{J})(\sum_{j}(R_j^{1})^{I}(R_{j}^{2})^{J})M_\ell^I \rho^I(g)))\\
  &=\tr^I((S\otimes \Id)((R')^{-1})^{IJ}(R')^{IJ}M_\ell^I \rho^I(g))
  =\tr_q^I(M_\ell^I).
\end{aligned}  
\end{equation}   
\end{proof}
The quantum trace identities still hold in the truncated case with substitution rule \eqref{eq:substitution_rule} applied \cite{Alekseev:1994au}.
\begin{lemma} 
Given a linear order at the base point of the standard graph, the exchange relation of monodromies for $\ell_{e}\succ \ell_{e'}$ is given as
\begin{equation}
 (R')^{IJ}\MleIl ((R')^{-1})^{IJ}\MlepJr = \MlepJr (R')^{IJ}\MleIl \left((R')^{-1}\right)^{IJ}\,.
 \end{equation}
 \label{lemma:commu_linear_order}
\end{lemma}
\begin{proof}
One can break loops $\ell_{e}$ and $\ell_{e'}$ into curves $C,C'$ and $\widetilde{C},\widetilde{C}'$ respectively by adding two additional vertices $y'$ an $y$ to the loops $\ell_e$ and $\ell_{e'}$ respectively as illustrated in fig.\ref{fig_decomp_add_twovertices}.
At the vertex $x$, the linear order is given as $\widetilde{C}'\prec -\widetilde{C}\prec C' \prec -C$.
The quantum monodromy $\bM^I_{\ell_{e}}$ can be decomposed into the product of quantum holonomies $\bold{U}_C^I$ and $\bold{U}_C'^I$ as $\bM^I_{\ell_{e}}=\kappa_{I}^{-1}\bold{U}_C^I\bold{U}_{C'}^I$. In the same way, the quantum monodromy $\bM^J_{\ell_{e'}}$ can be expressed in terms of quantum holonomies $\bold{U}_{\widetilde{C}}^J$ and $\bold{U}_{\widetilde{C}'}^J$ as $\bM^J_{\ell_{e'}}=\kappa_{J}^{-1}\bold{U}_{\widetilde{C}}^J\bold{U}_{\widetilde{C}'}^J$. The commutation relation of $\bM^I_{\ell_e}$ and $\bM^I_{\ell_{e'}}$ can be derived from the commutation relation of quantum holonomies.
\be
\begin{split}
\MleIl \left((R')^{-1}\right)_{x}^{IJ}\MlepJr &= \kappa_{I}^{-1}\UCl \UCpl \left((R')^{-1}\right)_{x}^{IJ}\kappa_{I}^{-1}\UCr \UCpr \\  
 &=\kappa_{I}^{-1}\kappa_{I}^{-1}\UCl \UCr \UCpl \UCpr \\
 &=\kappa_{I}^{-1}\kappa_{I}^{-1}\left((R')^{-1}\right)_{x}^{IJ}\UCr \UCl \UCpr \UCpl \left((R')^{-1}\right)_{x}^{IJ}\\
 &=\kappa_{I}^{-1}\kappa_{I}^{-1}\left((R')^{-1}\right)_{x}^{IJ}\UCr \UCpr \left(R'\right)_{x}^{IJ}\UCl \UCpl \left((R')^{-1}\right)_{}x^{IJ}\\
 &=\left((R')^{-1}\right)_{x}^{IJ}\MlepJr \left(R'\right)_{x}^{IJ}\MleIl \left((R')^{-1}\right)_{x}^{IJ}\;.
\end{split}
\ee
Using commutation relations of holonomies at vertex $x$, the derivation above is straightforward.
\end{proof}
\begin{lemma}
For $R$ and $R'$ satisfying the quasi-Yang-Baxter equation and $\varphi\equiv \varphi_{123}$ being quasi-invertible, the following relation holds
 \be
 \varphi_{123}R_{12}'\varphi^{-1}_{213}R'_{13}R_{13}\varphi_{213}R_{12}'^{-1}\varphi^{-1}\varphi_{123}R_{12}'\varphi^{-1}_{213}\varphi_{213}R_{12}\varphi^{-1}=\varphi_{123}R_{12}'\varphi^{-1}_{213}R'_{13}R_{13}\varphi_{213}R_{12}\varphi^{-1}.
 \ee
\label{lemma:varphi_123}
\end{lemma}
\begin{proof}
 \be
 \begin{split}
 \varphi_{213}R_{12}'^{-1}\varphi^{-1}\varphi_{123}R_{12}'\varphi^{-1}_{213}\varphi_{213}R_{12}\varphi^{-1}&= \varphi_{213}R_{12}'^{-1}\varphi^{-1}\varphi_{123}R_{12}'(\Delta'\otimes\Id)\Delta(e) R_{12}\varphi^{-1} \\
 &= \varphi_{213}R_{12}'^{-1}\varphi^{-1}\varphi_{123}(\Delta\otimes\Id)\Delta(e)R_{12}' R_{12}\varphi^{-1}\\
 &=\varphi_{213}R_{12}'^{-1}\varphi^{-1}\varphi_{123}R_{12}' R_{12}\varphi^{-1}\\
 &=\varphi_{213}R_{12}'^{-1}(\Delta\otimes\Id)\Delta(e)R_{12}' R_{12}\varphi^{-1}\\
 &=\varphi_{213}(\Delta'\otimes\Id)\Delta(e)R_{12}'^{-1}R_{12}' R_{12}\varphi^{-1}\\
 &=P_{12}[\varphi_{123}(\Delta\otimes\Id)\Delta(e)]R_{12}'^{-1}R_{12}' R_{12}\varphi^{-1}\\
 &=\varphi_{213}R_{12}'^{-1}R_{12}' R_{12}\varphi^{-1}\\
 &=\varphi_{213}R_{12}\varphi^{-1}\;.
 \end{split}
 \label{eq:proof_lemma_B3}
 \ee
 The first and fifth lines are obtained by quasi-triangularity, and the third and seventh lines are derived by the quasi-inverse property, $\varphi\varphi^{-1}\varphi=\varphi$:
 \be
 (\Id\otimes \Delta )\Delta(e)\varphi=
 \varphi(\Delta\otimes \Id)\Delta(e)=\varphi\;.
 \ee
 The notation $P_{12}$ used in the sixth equality of \eqref{eq:proof_lemma_B3} represents the permutation of the first and second elements. 
\end{proof}
\begin{lemma}
For $R$ and $R'$ obeying the quasi-Yang-Baxter equation, the following relation holds
 \be
 (id\otimes \Delta)(R'R)=\varphi_{123}R_{12}'\varphi^{-1}_{213}R'_{13}R_{13}\varphi_{213}R_{12}\varphi^{-1}.
 \ee
 \label{lemma:delta_R'R}
\end{lemma}
\begin{proof}
We first check the quasi-triangularity of the quasi-Hopf algebra. One can directly read out the isomorphisms of representations from the graphical illustrations shown in fig.\ref{fig:RR'} and fig.\ref{fig:enter-label}. Firstly, from fig.\ref{fig:R_13R_12}, the expression for $(\rho^I\otimes \rho^J\otimes \rho^K)(\Id\otimes \Delta)(R)$ that is the isomorphism between $(V^I\otimes(V^J\otimes V^K))$ and $((V^J\otimes V^K)\otimes V^I) $ is read as
\be
(\rho^I\otimes \rho^J\otimes \rho^K)(\varphi^{-1}_{231}R_{13}\varphi_{213}R_{12}\varphi^{-1}_{123})\;.
\ee
From the fig.\ref{fig:R_13R_23}, the expression for $(\rho^I\otimes \rho^J\otimes \rho^K)(\Delta\otimes \Id)(R)$ that is the isomorphism between $((V^I\otimes V^J)\otimes V^K))$ and $(V^K\otimes (V^I\otimes V^J)) $ is read as
\be
(\rho^I\otimes \rho^J\otimes \rho^K)(\varphi_{312}R_{13}\varphi^{-1}_{132}R_{23}\varphi_{123})\;,
\ee
where $\varphi^{IJK}$ maps $((V^I\otimes V^J)\otimes V^K)$ to $(V^I \otimes (V^J\otimes V^K))$ and $(\varphi^{-1})^{IJK}$ maps $(V^I \otimes (V^J\otimes V^K))$ to $((V^I\otimes V^J)\otimes V^K)$ for both cases. $\varphi^{KIJ}\equiv (\varphi_{312})^{IJK}$ maps $(V^K\otimes V^I)\otimes V^J)$ to $(V^K \otimes (V^I\otimes V^J))$ and $(\varphi^{-1})^{IKJ}\equiv (\varphi^{-1}_{132})^{IJK}$ maps $(V^I \otimes (V^K\otimes V^J))$ to $((V^I\otimes V^K)\otimes V^J)$

The isomorphism $(\rho^I \otimes \rho^J \boxtimes  \rho^K)(R'R)\equiv (\rho^I\otimes \rho^J \otimes \rho^K)(id\otimes \Delta)(R'R)$ maps $(V^I \otimes (V^J\otimes V^K))$ to $(V^I \otimes (V^J\otimes V^K))$. From the fig.\ref{fig:RR'} , the expression for $(\rho^I\otimes \rho^J \otimes \rho^K)(\Id\otimes \Delta)(R'R)$ is read as
\be
(\rho^I\otimes \rho^J \otimes \rho^K)(\varphi_{123}R_{12}'\varphi^{-1}_{213}R'_{13}R_{13}\varphi_{213}R_{12}\varphi^{-1})\;.
\ee
\end{proof}
\begin{figure}[h!]
    \centering
    \includegraphics[width=0.45\textwidth]{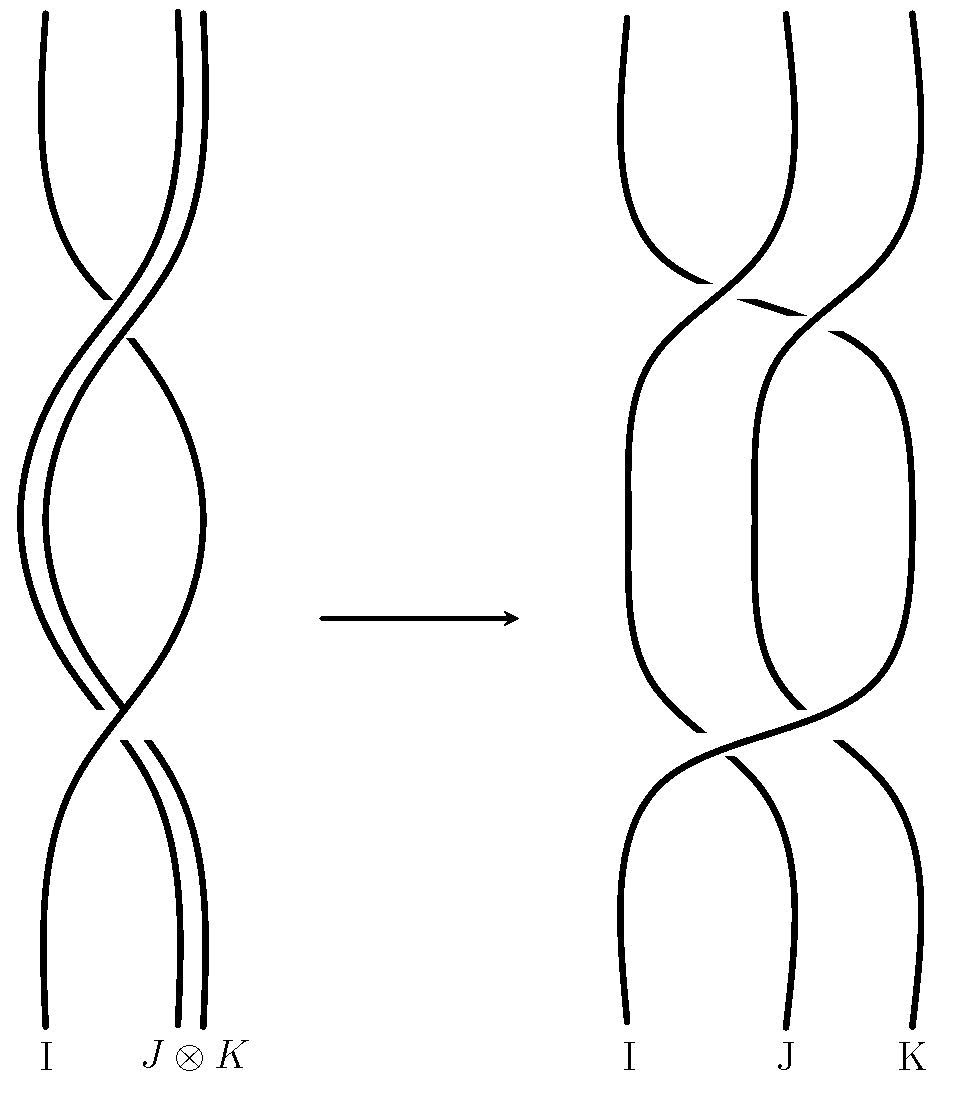}
    \caption{The action of the graph on both sides is read from bottom to top. On the left-hand side, the first action is given by $R$ and the second one is given by $R'$. On the right-hand side, the actions are given by $\varphi_{213}R\varphi^{-1}$, $R_{13}$, $R'_{13}$, and $\varphi_{123}R\varphi^{-1}_{213}$ respectively.}
    \label{fig:RR'}
\end{figure}

  \begin{figure}[th!]
      \begin{subfigure}[t]{0.38\linewidth}
      \begin{tikzpicture}[scale=1.2]
      \coordinate (A) at (0,-2.0);
\coordinate (B) at (1.0,-2.0);
\coordinate (C) at (1.1,-2.0);
\coordinate (D) at (1.0,2.0);
\coordinate (E) at (1.1,2.0);
\coordinate (F) at (1.1,2.0);
\coordinate (G) at (0.0,-0.6);
\coordinate (H) at (0.1,-0.6);
 \coordinate (O) at (0,-1.6);
 \coordinate (P) at (1.0,-1.6);
\coordinate (Q) at (1.1,-1.6);
 \coordinate (R) at (0,-0.8);
 \coordinate (S) at (0.1,-0.8);
 \coordinate (T) at (1.0,-0.8);
\coordinate (U) at (0.0,-0.2);
 \coordinate (V) at (0.1,-0.2);
\coordinate (W) at (1.0,-0.2);
 \coordinate (X) at (0.5,-1.21);
 \coordinate (Y) at (0.6,-1.17);
 \coordinate (Z) at (1.0,-0.6);
 \coordinate (I) at (0.4,-1.05);
 \coordinate (J) at (0.3,-1.1);
 \coordinate (K) at (4,-1.84);
 \draw[thick] (A) -- (O) ;
 \draw[thick] (B) -- (P) ;
 \draw[thick] (C) -- (Q) ;
 \draw[thick] (G) -- (U) ;
 \draw[thick] (H) -- (V) ;
 \draw[thick] (Z) -- (W) ;
 \draw (A) node[anchor=north]{$I$};
 \draw (1.05,-2.0) node[anchor=north]{$J\otimes K$};
 \draw (3.5,-2.0) node[anchor=north]{$I$};
 \draw (4.0,-2.0) node[anchor=north]{$J$};
 \draw (4.5,-2.0) node[anchor=north]{$K$};
\draw[thick] (P) to[out=90,in=333] 
  (X);
  \draw[thick] (Q) to[out=90,in=333] 
  (Y);
  \draw[thick] (O) to[out=90,in=270] 
  (Z);
  \draw[thick] (I) to[out=150,in=270] 
  (H);
  \draw[thick] (J) to[out=150,in=270] 
  (G);
          \draw[->, thick, >=stealth] (1.5,-1.25) -- (3,-1.25);
          \coordinate (a) at (3.5,-2.0);
\coordinate (b) at (4.0,-2.0);
\coordinate (c) at (4.5,-2.0);
\coordinate (o) at (3.5,-1.6);
 \coordinate (p) at (4.0,-1.6);
\coordinate (q) at (4.5,-1.6);
\coordinate (x) at (3.8,-1.29);
 \coordinate (y) at (4.2,-1.10);
 \coordinate (z) at (4.5,-0.6);

\coordinate (i) at (3.7,-1.19);
 \coordinate (j) at (4.1,-1.0);
 \coordinate (h) at (3.89,-0.6);
 \coordinate (g) at (3.5,-0.6);
\coordinate (u) at (3.5,-0.2);
 \coordinate (v) at (3.89,-0.2);
\coordinate (w) at (4.5,-0.2);
 \draw[thick] (a) -- (o) ;
 \draw[thick] (b) -- (p) ;
 \draw[thick] (c) -- (q) ;
 \draw[thick] (g) -- (u) ;
 \draw[thick] (h) -- (v) ;
 \draw[thick] (z) -- (w) ;

 \draw[thick] (p) to[out=90,in=333] 
  (x);
  \draw[thick] (q) to[out=90,in=333] 
  (y);
  \draw[thick] (o) to[out=90,in=270] 
  (z);
  \draw[thick] (i) to[out=150,in=270] 
  (g);
  \draw[thick] (j) to[out=150,in=270] 
  (h);
      \end{tikzpicture}
      \caption{}
\label{fig:R_13R_12}
      \end{subfigure}
      \begin{subfigure}[t]{0.38\linewidth}
      \begin{tikzpicture}[scale=1.2]
      \coordinate (A) at (0,-2.0);
\coordinate (B) at (0.1,-2.0);
\coordinate (C) at (1.1,-2.0);
\coordinate (D) at (1.0,2.0);
\coordinate (E) at (1.1,2.0);
\coordinate (F) at (1.1,2.0);
\coordinate (G) at (1.1,-0.6);
\coordinate (H) at (0.1,-0.6);
 \coordinate (O) at (0,-1.6);
 \coordinate (P) at (0.1,-1.6);
\coordinate (Q) at (1.1,-1.6);
 \coordinate (R) at (0,-0.8);
 \coordinate (S) at (0.1,-0.8);
 \coordinate (T) at (1.0,-0.8);
\coordinate (U) at (1.1,-0.2);
 \coordinate (V) at (0.1,-0.2);
\coordinate (W) at (1.0,-0.2);
 \coordinate (X) at (0.5,-1.21);
 \coordinate (Y) at (0.67,-1.23);
 \coordinate (Z) at (1.0,-0.6);
 \coordinate (I) at (0.4,-1.05);
 \coordinate (J) at (0.3,-1.1);
 \coordinate (K) at (4,-1.84);
 \draw[thick] (A) -- (O) ;
 \draw[thick] (B) -- (P) ;
 \draw[thick] (C) -- (Q) ;
 \draw[thick] (G) -- (U) ;
 \draw[thick] (H) -- (V) ;
 \draw[thick] (Z) -- (W) ;
 \draw (0.05,-2.0) node[anchor=north]{$I\otimes J$};
 \draw (C) node[anchor=north]{$K$};
  \draw (3.5,-2.0) node[anchor=north]{$I$};
  \draw (4.0,-2.0) node[anchor=north]{$J$};
  \draw (4.5,-2.0) node[anchor=north]{$K$};
\draw[thick] (P) to[out=90,in=210] 
  (X);
  \draw[thick] (Q) to[out=90,in=333] 
  (Y);
  \draw[thick] (O) to[out=90,in=270] 
  (Z);
  \draw[thick] (I) to[out=150,in=270] 
  (H);
  \draw[thick] (X) to[out=28,in=270] 
  (G);
  
  \draw[->, thick, >=stealth] (1.5,-1.25) -- (3,-1.25);
\coordinate (a) at (3.5,-2.0);
\coordinate (b) at (4.0,-2.0);
\coordinate (c) at (4.5,-2.0);

\coordinate (o) at (3.5,-1.6);
 \coordinate (p) at (4.0,-1.6);
\coordinate (q) at (4.5,-1.6);

\coordinate (x) at (4.5,-0.2);
 \coordinate (y) at (4.2,-1.10);
 \coordinate (z) at (3.95,-0.2);

\coordinate (i) at (3.75,-0.75);
 \coordinate (j) at (4.1,-1.0);
 \coordinate (h) at (3.89,-0.8);
 \coordinate (g) at (3.5,-0.6);
\coordinate (u) at (3.5,-0.2);
 \coordinate (v) at (3.89,-0.2);
\coordinate (w) at (4.5,-0.2);


 \draw[thick] (a) -- (o) ;
 \draw[thick] (b) -- (p) ;
 \draw[thick] (c) -- (q) ;
 
 \draw[thick] (p) to[out=90,in=270] 
  (x);
  \draw[thick] (q) to[out=90,in=313] 
  (y);
  \draw[thick] (o) to[out=90,in=270] 
  (z);
  \draw[thick] (i) to[out=150,in=270] 
  (u);
  \draw[thick] (j) to[out=140,in=340] 
  (h);
      \end{tikzpicture}
      \caption{}
\label{fig:R_13R_23}
      \end{subfigure}
      \caption{{\it (a)} The actions of the graph on both sides are read from bottom to top. On the left-hand side, the action is given by $R$. On the right-hand side, the actions are given by $\varphi_{213}R\varphi^{-1}$, $R_{13}$, and $\varphi^{-1}_{231}$ respectively.  {\it (b)} The actions of the graph on both sides are read from bottom to top. On the left-hand side, the action is given by $R$. On the right-hand side, the actions are given by $\varphi^{-1}_{132}R_{23}\varphi$, $R_{13}$, and $\varphi_{312}$ respectively.}
      \label{fig:enter-label}
  \end{figure}
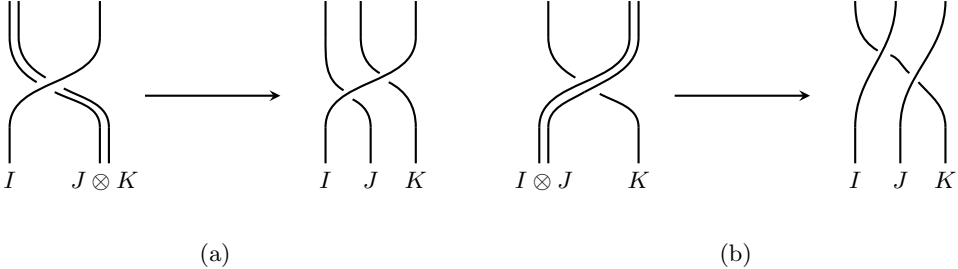  

\section{Quantum State on the torus and coherent state}
\label{app:Quantum_State_on_the_torus}
In this appendix, we briefly review the construction and features of the generic quantum states and coherent states on the torus \cite{Gazeau:2009zz} and give the derivation explicitly. The auxiliary space $\mathcal{H}_{\text{aux}}$ is the quantization of the torus as the phase space. Suppose that we have a pair of operators $(\bold{Q},\bold{P})$ with the commutation relation $[\bold{Q},\bold{P}]=i\hbar$. Since they are on the torus it is natural to impose the periodic condition. The periodic condition is defined in \cite{Gazeau:2009zz} as
\be
U(a,0)\Psi=e^{-i\kappa_1a}\Psi\quad,\quad U(0,b)\Psi=e^{-i\kappa_2b}\Psi, 
\ee
where the operator is defined as $U(\alpha,\beta)=e^{\frac{i}{\hbar}(\beta \bold{Q}-\alpha \bold{P})}$ and $(\kappa_1,\kappa_2)\in$ $[0,\frac{2\pi}{a})  \times [0,\frac{2\pi}{b})$. 
A state on the torus can be obtained from arbitrary Schwarz function by the map $\mathrm{P}_{\kappa}$ \cite{Gazeau:2009zz}. Given any Schwarz function $\psi(x),\;x\in\mathbb{R}$, one can construct a state on the torus from it as \cite{Gazeau:2009zz}:
\be
\begin{aligned}
\psi(x)\mapsto\mathrm{P}_{\kappa}\psi(x)&=\sum_{m,n\in\mathbb{Z}}(-1)^{Nmn}e^{i\left(\kappa_{1}ma-\kappa_{2}nb\right)}U(ma,nb)\psi(x)=\sum_{m,n\in\mathbb{Z}}(-1)^{Nmn}e^{i\left(\kappa_{1}ma-\kappa_{2}nb\right)}e^{-\frac{i}{2\hbar}mnab}e^{\frac{i}{\hbar}nbx}\psi(x-ma)\\
&=\sum_{m\in\mathbb{Z}}e^{i\kappa_{1}ma}\psi(x-ma)\sum_{n\in\mathbb{Z}}e^{\frac{i}{\hbar}nb\left(x-\kappa_{2}\hbar-\frac{1}{2}ma\right)+i\pi Nmn}\\
&=\sum_{m\in\mathbb{Z}}e^{i\kappa_{1}ma}\psi(x-ma)\sum_{n\in\mathbb{Z}}e^{i2\pi n\left(\frac{N}{a}(x-\kappa_{2}\hbar-\frac{1}{2}ma)+\frac{1}{2}Nm\right)}\\
&=\sum_{m\in\mathbb{Z}}e^{i\kappa_{1}ma}\psi(x-ma)\sum_{l\in\mathbb{Z}}\delta\left(\frac{N}{a}(x-\kappa_{2}\hbar-\frac{1}{2}ma)+\frac{1}{2}Nm-l\right)\\
&=\sum_{m\in\mathbb{Z}}e^{i\kappa_{1}ma}\psi(x-ma)\sum_{l\in\mathbb{Z}}\frac{a}{N}\delta\left(x-\kappa_{2}\hbar+\frac{1}{2}ma-\frac{1}{2}ma-l\frac{a}{N}\right)\\
&=\sum_{l\in\mathbb{Z}}\sum_{m\in\mathbb{Z}}e^{i\kappa_{1}ma}\psi(\frac{ab}{2\pi N}\kappa_{2}+l\frac{a}{N}-ma)\frac{a}{N}\delta\left(x-\frac{ab}{2\pi N}\kappa_{2}-l\frac{a}{N}\right)\\
&=\sum_{l\in\mathbb{Z}}\left(\sum_{m\in\mathbb{Z}}\sqrt{\frac{a}{N}}e^{i\kappa_{1}ma}\psi(\frac{ab}{2\pi N}\kappa_{2}+l\frac{a}{N}-ma)\right)\left(\sqrt{\frac{a}{N}}\delta\left(x-\frac{ab}{2\pi N}\kappa_{2}-l\frac{a}{N}\right)\right)\\
&=\sum_{l\in\mathbb{Z}}c_{l}(\psi)\delta_{l}(x)\,,
\end{aligned}
\ee
where $c_{l}(\psi):=\sum_{m\in\mathbb{Z}}\sqrt{\frac{a}{N}}e^{i\kappa_{1}ma}\psi(\frac{ab}{2\pi N}\kappa_{2}+l\frac{a}{N}-ma)$ and $\delta_{l}(x):=\sqrt{\frac{a}{N}}\delta\left(x-\frac{ab}{2\pi N}\kappa_{2}-l\frac{a}{N}\right)$. 
The state $\mathrm{P}_{\kappa}\psi(x)$ is in an $N$ dimensional vector space denoted as $S'(\mathbf{\kappa}, N)$ \cite{Gazeau:2009zz}. By definition, the inner product is defined as
\be
\langle\mathrm{P}_{\kappa}\psi|\mathrm{P}_{\kappa}\psi'\rangle=\sum^{N-1}_{l=0} \bar{c_l}(\psi)c_l(\psi').
\ee
 By applying the periodicity condition on the state, one has the periodicity condition for the coefficients $c_{l}(\psi)$, which we now show. Firstly, 
\be
\begin{aligned}
U(a,0)\delta_{l}(x)&=U(a,0)\sqrt{\frac{a}{N}}\delta\left(x-\frac{ab}{2\pi N}\kappa_{2}-l\frac{a}{N}\right)=e^{-\frac{i}{\hbar}aP}\sqrt{\frac{a}{N}}\delta\left(x-\frac{ab}{2\pi N}\kappa_{2}-l\frac{a}{N}\right)\\
&=\sqrt{\frac{a}{N}}\delta\left(x-\frac{ab}{2\pi N}\kappa_{2}-(l+N)\frac{a}{N}\right)\,.
\end{aligned}
\ee
Let $l'=l+N$, then
\be
\begin{aligned}
U(a,0)\delta_{l}(x)&=\sum_{l\in\mathbb{Z}}\left(\sum_{m\in\mathbb{Z}}\sqrt{\frac{a}{N}}e^{i\kappa_{1}ma}\psi(\frac{ab}{2\pi N}\kappa_{2}+l\frac{a}{N}-ma)\right)\left(\sqrt{\frac{a}{N}}\delta\left(x-\frac{ab}{2\pi N}\kappa_{2}-(l+N)\frac{a}{N}\right)\right)\\
&=\sum_{l'\in\mathbb{Z}}\left(\sum_{m\in\mathbb{Z}}\sqrt{\frac{a}{N}}e^{i\kappa_{1}ma}\psi(\frac{ab}{2\pi N}\kappa_{2}+(l'-N)\frac{a}{N}-ma)\right)\left(\sqrt{\frac{a}{N}}\delta\left(x-\frac{ab}{2\pi N}\kappa_{2}-(l')\frac{a}{N}\right)\right)\\
&=\sum_{l\in\mathbb{Z}}c_{l-N}(\psi)\delta_{l}(x)\,.
\end{aligned}
\ee
Therefore, 
\be
 U(a,0)\psi=e^{-i\kappa_{1}a}\psi=\sum_{l\in\mathbb{Z}}e^{-i\kappa_{1}a}c_{l}(\psi)\delta_{l}(x)\,.
\ee
We then obtain the periodicity condition on the coefficients themselves:
\be
c_{l-N}(\psi)=e^{-i\kappa_{1}a}c_{l}(\psi)\quad ,\quad
c_{l+N}(\psi)=e^{i\kappa_{1}a}c_{l}(\psi)\,.
\ee
Since we are dealing with the $N$-dimensional Hilbert space, one can establish a one-to-one correspondence between $S'(\mathbf{\kappa}, N)$ and $\mathbb{C}^N$ with the canonical basis $e^{\mathbf{\kappa}}_{j}=\sqrt{\frac{a}{N}}\sum_{l\in\mathbb{Z}} e^{-ima\kappa_1}\delta(x-x^j_l)\in\bC^N$, where $x^j_l=\frac{ab}{2\pi N}\kappa_2+j\frac{a}{N}-la$ \cite{Gazeau:2009zz}.

The coherent states on the torus can be defined by the displacement operator $U(\alpha,\beta)$ acting on the ground state. To construct the displacement operator, it is necessary to define the Weyl-Heisenberg group. In our setup, we can choose $(\ln{\bold{x}},\ln{\bold{y}},i\operatorname{id})$ to be the generators of the Weyl-Heisenberg Lie algebra, with the Lie bracket defined in \eqref{eq:commutation_lnx_lny}, i.e.  $[\ln{\bold{x}},\ln{\bold{y}}]=i\hbar\equiv i \frac{\pi}{k+2}$. Since $\ln{\bold{x}},\ln{\bold{y}}$ are anti-self-adjoint, one can set $\ln{\bold{x}}=i\bold{Q}$ and $\ln{\bold{y}}=
i\bold{P}$, where $\bold{Q}$ and $\bold{P}$ are self-adjoint operators with the commutation as $[\bold{Q},\bold{P}]=-i\hbar$. This commutation relation is derived from the commutation relation of $\ln{\bold{x}},\ln{\bold{y}}$. We then have $(i\bold{Q},-i\bold{P},i\operatorname{id})$ as generators of the Weyl-Heisenberg Lie algebra.  To follow the construction in \cite{Gazeau:2009zz}, we define $\bold{\Xt_0}=\bold{Q}-\frac{\pi}{k+2}, \bold{\Yt_0}=\bold{P}-p_0$, maintaining the same commutation relation $[\bold{\Xt_0},\bold{\Yt_0}]=-i\hbar$, and the position $\Xt_0$ is now in the range of $(0,2\pi)$.} 

 The exponential map of them gives us the Weyl-Heisenberg group elements, which are unitary operators and defines unitary irreducible representations.
Since the quantum states are defined on the torus, they must satisfy the quantization condition \eqref{eq:quantization_condition_torus}.
 Following \cite{Gazeau:2009zz}, the displacement operator in our case is defined as $U(\alpha, \beta)=e^{\frac{i}{\hbar}(\beta \bold{\Xt_0}+\alpha \bold{\Yt_0})}$.
\medskip

  The Schwartz function we choose is the standard coherent state for harmonic oscillator in the position representation, where the ground state is defined as 
\be
\psi^{\mathcal{Z}}_{(0,0)}(x)=\left(\frac{\mathfrak{I}\mathcal{Z}}{\pi\hbar}\right)^{\frac{1}{4}} 
e^{-\frac{\mathcal{JZ}}{2\hbar}x^2}\;.
\ee
Here we use the same notation as in \cite{Gazeau:2009zz} and $\mathfrak{I}$ and $\mathcal{Z}$ are $-i$ and $i$ respectively. The regular coherent state is then obtained from the action of the displacement operator on the ground state:
\be
\begin{split}
 \psi^{\mathcal{Z}}_{(\Xt_0,\Yt_0)}(x):=e^{\frac{i}{\hbar}(\Yt_0\bold{\Xt_0}+\Xt_0\bold{\Yt_0})} \psi^{\mathcal{Z}}_{(0,0)}(x)=\left(\frac{\mathfrak{I}\mathcal{Z}}{\pi\hbar}\right)^{\frac{1}{4}}e^{-\frac{i\Yt_0\Xt_0}{2\hbar}}e^{\frac{i}{\hbar}\Yt_0 x}e^{i\frac{\mathcal{Z}}{2\hbar}(x-\Xt_0)^{2}}\;. 
 \end{split}
\ee
The coherent state on the torus is defined as \cite{Gazeau:2009zz}:
\be
\begin{split}
 \psi^{\mathcal{Z},\mathbf{\kappa}}_{(\Xt_0,\Yt_0)}(x)&:= P_k \psi^{\mathcal{Z}}_{(\Xt_0,\Yt_0)}(x) =\sum_{m,n\in\mathbb{Z}}(-1)^{Nmn}e^{i\left(\kappa_{1}ma-\kappa_{2}nb\right)}U(ma,nb)\psi_{(\Xt_0,\Yt_0)}^{\mathcal{Z}}(x)\\
 &=\sum_{m,n\in\mathbb{Z}}(-1)^{Nmn}e^{i\left(\kappa_{1}ma-\kappa_{2}nb\right)}e^{\frac{i}{\hbar}(nb \bold{\Xt_0}+ ma \bold{\Yt_0})}(\left(\frac{\mathfrak{I}\mathcal{Z}}{\pi\hbar}\right)^{\frac{1}{4}}e^{-\frac{i\Yt_0\Xt_0}{2\hbar}}e^{\frac{i}{\hbar}\Yt_0x}e^{i\frac{\mathcal{Z}}{2\hbar}(x-\Xt_0)^{2}})\\
 &=\sum_{m,n\in\mathbb{Z}}(-1)^{Nmn}e^{i\left(\kappa_{1}ma-\kappa_{2}nb\right)}e^{-\frac{i}{2\hbar}(mnab )}e^{\frac{i}{\hbar}(nb x)}(\left(\frac{\mathfrak{I}\mathcal{Z}}{\pi\hbar}\right)^{\frac{1}{4}}e^{-\frac{i\Yt_0\Xt_0}{2\hbar}}e^{\frac{i}{\hbar}\Yt_0(x-ma)}e^{i\frac{\mathcal{Z}}{2\hbar}(x-ma-\Xt_0)^{2}})\\
 &=\left(\frac{\mathfrak{I}\mathcal{Z}}{\pi\hbar}\right)^{\frac{1}{4}} e^{-\frac{i\Yt_0\Xt_0}{2\hbar}}\sum_m e^{i(\kappa_{1}ma)} e^{\frac{i}{\hbar}\Yt_0(x-ma)}e^{i\frac{\mathcal{Z}}{2\hbar}(x-ma-\Xt_0)^{2}}\sum_n e^{\frac{i}{\hbar}nb(x-\frac{1}{2}ma-\hbar\kappa_2)+i\pi Nmn}\\
 &=\left(\frac{\mathfrak{I}\mathcal{Z}}{\pi\hbar}\right)^{\frac{1}{4}} e^{-\frac{i\Yt_0\Xt_0}{2\hbar}}\sum_m e^{i(\kappa_{1}ma)} e^{\frac{i}{\hbar}\Yt_0(x-ma)}e^{i\frac{\mathcal{Z}}{2\hbar}(x-ma-\Xt_0)^{2}}\sum_n e^{i2\pi n(\frac{N}{a}(x-\frac{1}{2}ma-\hbar\kappa_2)+\frac{1}{2} Nm)}\\
 &=\left(\frac{\mathfrak{I}\mathcal{Z}}{\pi\hbar}\right)^{\frac{1}{4}} e^{-\frac{i\Yt_0\Xt_0}{2\hbar}}\sum_m e^{i(\kappa_{1}ma)} e^{\frac{i}{\hbar}\Yt_0(x-ma)}e^{i\frac{\mathcal{Z}}{2\hbar}(x-ma-\Xt_0)^{2}} \frac{a}{N} \sum_l \delta(x-\frac{ab}{2\pi N}\kappa_2-l\frac{a}{N})\\
 &=\left(\frac{\mathfrak{I}\mathcal{Z}}{\pi\hbar}\right)^{\frac{1}{4}} e^{-\frac{i\Yt_0\Xt_0}{2\hbar}}\sum_m \sum_l e^{i(\kappa_{1}ma)} e^{\frac{i}{\hbar}\Yt_0(\frac{ab}{2\pi N}\kappa_2+l\frac{a}{N}-ma)}e^{i\frac{\mathcal{Z}}{2\hbar}(\frac{ab}{2\pi N}\kappa_2+l\frac{a}{N}-ma-\Xt_0)^{2}} \frac{a}{N}  \delta(x-\frac{ab}{2\pi N}\kappa_2-l\frac{a}{N})\;.
 \end{split}
\ee
In the fourth equality, we obtain the result by inserting $(-1)^{Nmn}=(e^{i\pi})^{Nmn}$. We can choose the sum of $l$ to be within one fixed periodicity by introducing another sum over $j\in \mathbb{Z}$.
\be
\begin{split}
 =&\left(\frac{\mathfrak{I}\mathcal{Z}}{\pi\hbar}\right)^{\frac{1}{4}}\sqrt{\frac{a}{N}}e^{-\frac{i\Yt_0\Xt_0}{2\hbar}}\sum_{l=0}^{N-1}\sum_{j\in\mathbb{Z}}\sum_{m\in\mathbb{Z}}e^{\frac{i}{\hbar}\Yt_0\left(\frac{ab}{2\pi N}\kappa_{2}+\left(jN+l\right)\frac{a}{N}-ma\right)}e^{i\frac{\mathcal{Z}}{2\hbar}(\frac{ab}{2\pi N}\kappa_{2}+\left(jN+l\right)\frac{a}{N}-ma-\Xt_0)^{2}}\\ 
 &\sqrt{\frac{a}{N}}e^{i\kappa_{1}ma}\delta\left(x-\frac{ab}{2\pi N}\kappa_{2}-\left(jN+l\right)\frac{a}{N}\right)\\
 =&\left(\frac{\mathfrak{I}\mathcal{Z}}{\pi\hbar}\right)^{\frac{1}{4}}\sqrt{\frac{a}{N}}e^{-\frac{i\Yt_0\Xt_0}{2\hbar}}\sum_{l=0}^{N-1}\sum_{j\in\mathbb{Z}}\sum_{m\in\mathbb{Z}}e^{\frac{i}{\hbar}\Yt_0\left(\frac{ab}{2\pi N}\kappa_{2}+\left(l\right)\frac{a}{N}-(m-j)a\right)}e^{i\frac{\mathcal{Z}}{2\hbar}(\frac{ab}{2\pi N}\kappa_{2}+\left(l\right)\frac{a}{N}-(m-j)a-\Xt_0)^{2}}\\ 
 &\sqrt{\frac{a}{N}}e^{i\kappa_{1}(m-j)a}e^{i\kappa_{1}ja}\delta\left(x-\frac{ab}{2\pi N}\kappa_{2}-\left(jN+l\right)\frac{a}{N}\right)\\
 =&\left(\frac{\mathfrak{I}\mathcal{Z}}{\pi\hbar}\right)^{\frac{1}{4}}\sqrt{\frac{a}{N}}e^{-\frac{i\Yt_0\Xt_0}{2\hbar}}\sum_{l=0}^{N-1}\sum_{m'\in\mathbb{Z}}e^{\frac{i}{\hbar}\Yt_0\left(\frac{ab}{2\pi N}\kappa_{2}+\left(l\right)\frac{a}{N}-(m')a\right)}e^{i\frac{\mathcal{Z}}{2\hbar}(\frac{ab}{2\pi N}\kappa_{2}+\left(l\right)\frac{a}{N}-(m')a-\Xt_0)^{2}}e^{i\kappa_{1}(m')a}\\ 
 &\sum_{j\in\mathbb{Z}}\sqrt{\frac{a}{N}}e^{i\kappa_{1}ja}\delta\left(x-\frac{ab}{2\pi N}\kappa_{2}-\left(jN+l\right)\frac{a}{N}\right)\\
 =&\left(\frac{\mathfrak{I}\mathcal{Z}}{\pi\hbar}\right)^{\frac{1}{4}}\sqrt{\frac{a}{N}}e^{-\frac{i\Yt_0\Xt_0}{2\hbar}}\sum_{l=0}^{N-1}\sum_{m'\in\mathbb{Z}}e^{\frac{i}{\hbar}\Yt_0\left(\frac{ab}{2\pi N}\kappa_{2}+\left(l\right)\frac{a}{N}-(m')a\right)}e^{i\frac{\mathcal{Z}}{2\hbar}(\frac{ab}{2\pi N}\kappa_{2}+\left(l\right)\frac{a}{N}-(m')a-\Xt_0)^{2}}e^{i\kappa_{1}(m')a}\\ 
 &\sum_{j\in\mathbb{Z}}\sqrt{\frac{a}{N}}e^{-i\kappa_{1}ja}\delta\left(x-\frac{ab}{2\pi N}\kappa_{2}-\left(-jN+l\right)\frac{a}{N}\right)\\
 =&\left(\frac{\mathfrak{I}\mathcal{Z}}{\pi\hbar}\right)^{\frac{1}{4}}\sqrt{\frac{a}{N}}e^{-\frac{i\Yt_0\Xt_0}{2\hbar}}\sum_{l=0}^{N-1}\sum_{m'\in\mathbb{Z}}e^{\frac{i}{\hbar}\Yt_0\left(\frac{ab}{2\pi N}\kappa_{2}+\left(l\right)\frac{a}{N}-(m')a\right)}e^{i\frac{\mathcal{Z}}{2\hbar}(\frac{ab}{2\pi N}\kappa_{2}+\left(l\right)\frac{a}{N}-(m')a-\Xt_0)^{2}}e^{i\kappa_{1}(m')a}\\ 
 &\sum_{j\in\mathbb{Z}}\sqrt{\frac{a}{N}}e^{-i\kappa_{1}ja}\delta\left(x-\frac{ab}{2\pi N}\kappa_{2}-l\frac{a}{N}+ja\right)\\
 =&\sum_{l=0}^{N-1} c_{(q',p')}(l)e^{\mathbf{\kappa}}_{l}(x) \;.
\end{split}
\ee

In our convention, we set: $a=2\pi=b$, $\; \kappa_1=0=\kappa_2,\; N=2(k+2),\; \Xt_0=\Xt_0,\; \Yt_0=\Yt_0,\; \hbar=\frac{\pi}{k+2},\; \mathfrak{I}=-i,\,\mathcal{Z}=i$. The final expressions become:
\be
\begin{split}
\xi_{(q',-p')}(l)&=\overline{\xi_l(q',p')}=c_{(q',p')}(l)=\left(\frac{1}{k+2}\right)^{\frac{1}{4}}e^{\frac{ip'q'(k+2)}{2\pi}}\sum_{m'\in\mathbb{Z}}e^{\frac{i(k+2)}{\pi}p'\left(l\frac{2\pi}{k+2}-2\pi m'-q'\right)}e^{-\frac{k+2}{2\pi}(l\frac{2\pi}{k+2}-2\pi m'-q')^{2}}\\
\xi_{(q',p')}(l)&=\overline{c_{(q',p')}(l)}=\left(\frac{1}{k+2}\right)^{\frac{1}{4}}e^{\frac{-ip'q'(k+2)}{2\pi}}\sum_{m'\in\mathbb{Z}}e^{\frac{-i(k+2)}{\pi}p'\left(l\frac{2\pi}{k+2}-2\pi m'-q'\right)}e^{-\frac{k+2}{2\pi}(l\frac{2\pi}{k+2}-2\pi m'-q')^{2}}\\
e_l(x)&=\sum_{j\in\mathbb{Z}}\sqrt{\frac{\pi}{k+2}}\delta\left(x-l\frac{\pi}{k+2}+j2\pi\right)\;.
\end{split}
\ee

\nocite*{}
\bibliographystyle{bib-style} 
\bibliography{QI.bib}

\end{document}